\documentclass[%
reprint,
superscriptaddress,
bibnotes,longbibliography
amsmath,amssymb,
aps,
floatfix
]{revtex4-2}
\usepackage[dvipdfmx]{graphicx}
\usepackage{graphicx}
\usepackage{color}
\usepackage{natbib}
\usepackage{amsmath,amssymb,amsthm,mathrsfs,amsfonts,dsfont}
\usepackage{subfigure, epsfig}
\usepackage{braket}
\usepackage{bm}
\usepackage{bbm}
\usepackage{enumerate}
\usepackage{here}
\usepackage{comment}
\usepackage{appendix} % For customizing appendices
\usepackage[colorlinks,linkcolor=blue,citecolor=blue]{hyperref}

\usepackage{makecell} % makecellパッケージをインクルード

\newtheorem{theorem}{Theorem}

\newtheorem{lemma}{Lemma}
\newtheorem{proposition}{Proposition}

\theoremstyle{definition}
\newtheorem{dfn}{Definition}
\theoremstyle{remark}

\theoremstyle{example}

\newcommand{\Exnu}{\mathbb{E}_{U \sim\nu}}

\newcommand{\poly}{\mathrm{poly}}
\newcommand{\polylog}{\mathrm{polylog}}

\usepackage[utf8]{inputenc}
\usepackage{array}       % For column width and text alignment
\usepackage{cellspace}   % For cell padding
\usepackage{graphicx}    % For scaling the table

\newcolumntype{C}[1]{>{\centering\arraybackslash}m{#1}}
\addparagraphcolumntypes{C}   % Integrate cellspace with new column type if needed

\allowdisplaybreaks[4]

\begin{document}
% --------------------  TITLE  --------------------
\title{Non-Haar random circuits form unitary designs as fast as Haar random circuits}

% ------------  AUTHORS AND AFFILIATIONS ----------
\author{Toshihiro Yada}
\email{yada@noneq.t.u-tokyo.ac.jp}
\affiliation{\mbox{Department of Applied Physics, The University of Tokyo, 7-3-1 Hongo, Bunkyo-ku, Tokyo 113-8656, Japan}}

\author{Ryotaro Suzuki}
\affiliation{\mbox{Dahlem Center for Complex Quantum Systems, Freie Universität Berlin, Berlin 14195, Germany}}

\author{Yosuke Mitsuhashi}
\affiliation{\mbox{Department of Basic Science, The University of Tokyo, 3-8-1 Komaba, Meguro-ku, Tokyo 153-8902, Japan}}

\author{Nobuyuki Yoshioka}
\affiliation{\mbox{International Center for Elementary Particle Physics, The University of Tokyo, 7-3-1 Hongo, Bunkyo-ku, Tokyo 113-0033, Japan}}

%\email{nyoshioka@ap.t.u-tokyo.ac.jp}

% --------------------  ABSTRACT  --------------------

\begin{abstract}
The unitary design formation in random circuits has attracted considerable attention due to its wide range of practical applications and relevance to fundamental physics. While the formation rates in Haar random circuits have been extensively studied in previous works, it remains an open question how these rates are affected by the choice of local randomizers. 
In this work, we prove that the circuit depths required for general non-Haar random circuits to form unitary designs are upper bounded by those for the corresponding Haar random circuits, up to a constant factor independent of the system size. This result is derived in a broad range of circuit structures, including one- and higher-dimensional lattices, geometrically non-local configurations, and even extremely shallow circuits with patchwork architectures. We provide specific applications of these results in randomized benchmarking and random circuit sampling, and also discuss their implications for quantum many-body physics. Our work lays the foundation for flexible and robust randomness generation in real-world experiments, and offers new insights into chaotic dynamics in complex quantum systems.
\end{abstract}
\maketitle

\vspace{1mm}\noindent
\textbf{Introduction:}
Random quantum circuits have been intensively studied as a powerful tool for various quantum information processing tasks such as randomized benchmarking \cite{helsen2022general} and quantum tomography \cite{elben2023randomized,huang2020predicting}, and also have served as tractable models for the nonequilibrium dynamics of quantum many-body systems \cite{fisher2023random,nahum2017quantum}.  
A central focus in these investigations is the rate at which the global randomness is generated.
While generating an exact Haar random ensemble in an $N$-qudit system requires a circuit depth of at least $e^{\Omega(N)}$, an ensemble that mimics Haar randomness up to the $t$-th moment, called a unitary $t$-design, is known to be created much more efficiently, with a depth of $\poly N$ \cite{harrow2009random,brown2010convergence,brandao2016local,haferkamp2022random,chen2024incompressibility,haferkamp2021improved,hunter2019unitary,harrow2023approximate}.
Furthermore, the recent intensive investigations reveal that $\poly N$-depth formation of unitary $t$-design is ubiquitous in generic circuit architectures \cite{mittal2023local,belkin2023approximate}, and even $O(\log N)$-depth formation is achievable in certain special geometries \cite{schuster2024random,laracuente2024approximate}, which we call patchwork circuit in this Letter. 
However, these previous works have mainly focused on scenarios where each two-qudit unitary gate that constitutes the random circuit is sampled from the Haar measure, or at least where the unitary gates have a special structure, such as forming a unitary subgroup \cite{webb2015clifford,zhu2017multiqubit}.

A prominent question then arises: \textit{To what extent is the rate of unitary design formation affected by the choice of local randomizers?}
This question is interesting not only from a fundamental perspective but also from an experimental standpoint, since in real experiments unitary gates are typically chosen non-Haar randomly, for example, from a discrete gate set \cite{boixo2018characterizing,arute2019quantum,morvan2024phase,zhu2022quantum,gao2025establishing}. While this question has been addressed in seminal works \cite{oszmaniec2021epsilon,oszmaniec2022saturation}, their depth upper bounds are loose by \(\poly N\)-factors and are valid only for circuits with specific structures, excluding any fixed-architecture circuits such as brickwork circuits. Therefore, the critical open problem is to clarify whether unitary designs can be formed with the same circuit depth order even when the local randomizers are modified, across a wide range of circuit structures.

\begin{figure*}[ht]
\begin{center}
\includegraphics[width=0.99\textwidth]{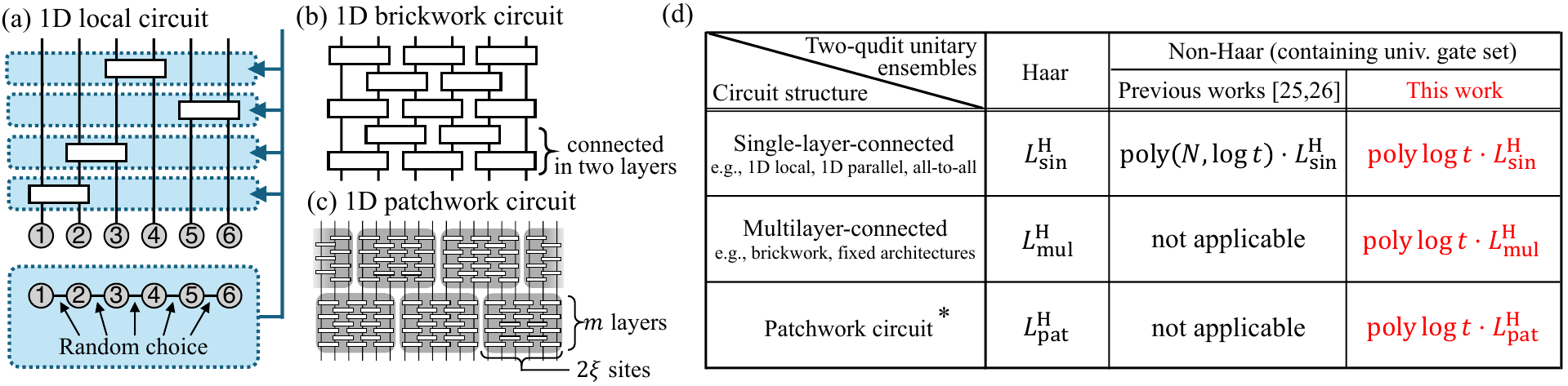}
\caption{Schematics of (a) 1D local circuit, (b) 1D brickwork circuit, and (c) 1D patchwork circuit. The patchwork circuit is a special geometry introduced in Ref.~\cite{schuster2024random} to generate a unitary design with extremely shallow circuit depth. (d) The best known upper bounds for the circuit depths required to form \(\varepsilon\)-approximate unitary \(t\)-designs in \(N\)-qudit systems. The upper bound for a single-layer-connected circuit obtained in this work is tighter than those in Refs.~\cite{oszmaniec2021epsilon,oszmaniec2022saturation} by a $\poly N$ factor. Furthermore, we also derive depth bounds for a multilayer-connected circuit and a patchwork circuit composed of them, which were not covered in those works. The asterisk in the patchwork row indicates that this bound is based on the relative error definition of an approximate \(t\)-design, while the other rows use Definition~\ref{def:approx_design}.}
\label{fig:abst_circuit}
\end{center}
\end{figure*}

In this work, we comprehensively study various circuit structures and prove that the circuit depths required to form unitary \(t\)-designs in general non-Haar random circuits are upper bounded by those in the corresponding Haar-random circuits, up to a constant prefactor. This prefactor is determined by the randomness strength of the two-qudit unitary ensembles constituting the circuit, and is independent of the system size \(N\). When these ensembles contain universal gate sets, it can be further shown that this prefactor exhibits at most \(\polylog t\) scaling in \(t\), indicating the insensitivity of the \(t\)-design formation depth to the choice of local randomizers in both \(N\) and \(t\).

In deriving these results, we classify the circuit structures into three types and develop different proof methods for each class. Specifically, the three classes are: $(\mathrm{i})$ \textit{single-layer-connected circuit}, where the two-qudit gates available in a single layer form a connected graph covering all sites, as shown in Fig.~\ref{fig:abst_circuit}(a); $(\mathrm{ii})$ \textit{multilayer-connected circuit}, where multiple layers are required to connect all sites, as shown in Fig.~\ref{fig:abst_circuit}(b); and $(\mathrm{iii})$ \textit{patchwork circuit}, where small patches of normal random circuits are glued together, as shown in Fig.~\ref{fig:abst_circuit}(c).
The proof methods developed in this work significantly improve both the tightness and the comprehensiveness of the existing upper bounds, as summarized in Fig.~\ref{fig:abst_circuit}(d).

These results have diverse applications and rich implications in various fields.
In terms of practical applications, our work enhances the flexibility of global randomness generation in real experiments and ensures that the generation rate is largely unaffected by deformations of local gates, for example, those caused by coherent errors. We raise a few specific examples that could benefit from this enhanced flexibility, in randomized benchmarking and random circuit sampling.
From a fundamental physics perspective, our work implies that various chaotic phenomena, such as information scrambling \cite{roberts2017chaos,nahum2018operator,von2018operator} and complexity growth \cite{oszmaniec2022saturation,brandao2021models,haferkamp2022linear}, universally emerge with the same circuit depth order, independent of the choice of local randomness. Thus, our work lays the foundation for randomness generation in quantum many-body systems and opens up numerous avenues for future applications.

\vspace{1mm}\noindent
\textbf{Preliminary:} 
Before addressing the random circuit, we briefly review the basic tools for an approximate unitary $t$-design.
We here consider a probability distribution $\nu$ on the set of unitary operations on $q$-dimensional Hilbert space, denoted as $\mathrm{U}(q)$. The distribution $\nu$ is said to form an approximate unitary $t$-design if the $t$-th moment of $\nu$ is close to that of the uniform distribution (i.e., the Haar measure) on $\mathrm{U}(q)$.
The $t$-th moment of the ensemble $\nu$ is completely characterized by the $t$-th moment operator, defined as $M_{\nu}^{(t)} \equiv \Exnu [U^{\otimes t, t}]$. Here, $\Exnu$ represents taking average over the ensemble $\nu$, and $U^{\otimes t, t} \equiv U^{\otimes t} \otimes U^{* \otimes t}$, where $U^{*}$ is the complex conjugation of $U$. 

The moment operator for the Haar measure, denoted as $M_{\rm Haar}^{(t)}$, is known to be identical to a projectior on the invariant subspace $\mathrm{Inv}(\mathrm{U}(q)^{\otimes t, t})$ of $\mathrm{U}(q)^{\otimes t, t}\equiv\{U^{\otimes t, t} | U\in \mathrm{U}(q)\}$ \cite{mele2024introduction}. We denote this projector as $P^{(t)}$, which implies $M_{\rm Haar}^{(t)} = P^{(t)}$.
The moment operator for a general (non-Haar) unitary ensemble $\nu$ is known to be represented as $M_{\nu}^{(t)} = P^{(t)} + R_{\nu}^{(t)}$, where the residual part $R_{\nu}^{(t)}$ acts on the orthogonal subspace of $\mathrm{Inv}(\mathrm{U}(q)^{\otimes t, t})$, and therefore satisfies $P^{(t)}R_{\nu}^{(t)} = R_{\nu}^{(t)}P^{(t)} =0$.

Using the moment operator, an approximate unitary design is defined as follows: 
\vspace{-1.5mm}
\begin{dfn}\label{def:approx_design} 
For $\varepsilon \geq 0$ and $t \in \mathbb{N}$, a probability distribution $\nu$ on $\mathrm{U}(q)$ is an $\varepsilon$-approximate unitary $t$-design if and only if 
\begin{equation} 
\label{eq:def_approx_design} 
\|M_{\nu}^{(t)} - M_{\rm Haar}^{(t)} \| \leq \frac{\varepsilon}{q^{2t}}, 
\end{equation}
where $\|\cdot\|$ represents the operator norm. 
\end{dfn} 
\vspace{-1.5mm}
\noindent 
While the several definitions of approximate unitary designs exist depending on how to quantify the difference between $\nu$ and the Haar measure, Definition~\ref{def:approx_design} implies other standard definitions such as the additive error and relative error $\varepsilon$-approximate $t$-design, as proved in Ref.~\cite{brandao2016local}.
Definition~\ref{def:approx_design} is directly compatible with Lemma~\ref{lem:basic_tool_design_formation} introduced below, making it particularly suitable for evaluating the design formation depth in random circuits. For this reason, it has been adopted in most prior works \cite{brandao2016local,haferkamp2022random,chen2024incompressibility,haferkamp2021improved,mittal2023local,belkin2023approximate}, either as the primary definition or as a baseline that implies other standard definitions.

When the unitary ensembles are applied sequentially, the composite ensemble monotonically approaches the Haar measure.
To quantify this approaching rate regarding the $t$-th moment, the spectral gap defined as $\Delta_{\nu}^{(t)} \equiv 1- \| M_{\nu}^{(t)} - M_{\rm Haar}^{(t)} \|$ plays an essential role. Since the moment operator for the ensemble $\nu_2*\nu_1$, which is the convolution of $\nu_1$ and $\nu_2$, becomes $M_{\nu_2*\nu_1}^{(t)} = M_{\nu_2}^{(t)}M_{\nu_1}^{(t)}$, we can derive the following lemma, which is frequently used in this work:
\vspace{-1.5mm}
\begin{lemma} \label{lem:basic_tool_design_formation}
The moment operator of the $L$-fold convolution of unitary ensemble $\chi_L \equiv \nu_L*\cdots*\nu_2*\nu_1$ satisfies 
\begin{equation}
    \|M_{\chi_L}^{(t)} -M_{\rm Haar}^{(t)} \| \leq \exp\left(-\sum_{i=1}^L\Delta_{\nu_i}^{(t)} \right).
\end{equation}
\end{lemma}
\vspace{-1.5mm}
\noindent
This lemma has been widely used in prior works~\cite{brandao2016local,haferkamp2022random,chen2024incompressibility,haferkamp2021improved,mittal2023local,belkin2023approximate}. The derivation proceeds as follows:
\begin{align}
    \|M_{\chi_L}^{(t)} -M_{\rm Haar}^{(t)} \| &= \left\|R_{\nu_L}^{(t)}R_{\nu_{L-1}}^{(t)}\cdots R_{\nu_2}^{(t)}R_{\nu_1}^{(t)} \right\| \nonumber \\
    &\leq \prod_{i=1}^L \|R_{\nu_i}^{(t)}\| \nonumber \\
    &\leq  e^{-\sum_{i=1}^{L} \Delta_{\nu_i}^{(t)}}, \nonumber 
\end{align}
where we used the relation $P^{(t)}R_{\nu_i}^{(t)} =R_{\nu_i}^{(t)}P^{(t)}=0$ in the first line, the submultiplicativity of the operator norm in the second line, and the inequality $1-x \leq e^{-x}$ in the third line.
Since $0\leq \Delta_{\nu_i}^{(t)} \leq 1$ always holds, this lemma indicates that the unitary ensemble monotonically approaches unitary design by the convolutions and its approaching rate is characterized by the spectral gap.

\begin{comment}
This lemma follows from the inequality
\begin{align}
    \|M_{\chi_L}^{(t)} -M_{\rm Haar}^{(t)} \| &= \left\|R_{\nu_L}^{(t)}R_{\nu_{L-1}}^{(t)}\cdots R_{\nu_2}^{(t)}R_{\nu_1}^{(t)} \right\| \nonumber \\
    &\leq \prod_{i=1}^L \|R_{\nu_i}^{(t)}\| \nonumber \\
    &\leq  e^{-\sum_{i=1}^{L} \Delta_{\nu_i}^{(t)}}, \nonumber 
\end{align}
where we used the relation $P^{(t)}R_{\nu_i}^{(t)} =R_{\nu_i}^{(t)}P^{(t)}=0$ in the first line, the submultiplicativity of the operator norm in the second line, and the inequality $1-x \leq e^{-x}$ in the third line.
Since $0\leq \Delta_{\nu_i}^{(t)} \leq 1$ always holds, this lemma indicates that the unitary ensemble monotonically approaches unitary design by the convolutions and its approaching rate is characterized by the spectral gap.    
\end{comment}

\vspace{1mm}\noindent
\textbf{Results for single-layer-connected circuit:}
In this work, we discuss the formation of unitary designs in random circuits on $N$-qudit systems, where the circuits are composed of two-qudit gates drawn from non-Haar ensembles. 
Specifically, in a single-layer-connected circuit, a unitary ensemble for each layer, denoted by $\nu$, is applied consecutively. For example, in a 1D local circuit, the ensemble \(\nu\) is defined such that a neighboring two-qudit pair is randomly selected and a two-qudit gate drawn from a general non-Haar ensemble is applied to that pair.
In such circuits, the key for evaluating the unitary $t$-design formation rate is to lower bound the spectral gap $\Delta_{\nu}^{(t)}$, as clear from Lemma~\ref{lem:basic_tool_design_formation}.

For lower bounding \(\Delta_{\nu}^{(t)}\), we relate it to the spectral gap of the corresponding Haar random circuit, denoted as \(\nu^{\rm H}\). The ensemble \(\nu^{\rm H}\) is defined by replacing all the two-qudit unitary ensembles constituting \(\nu\) with the Haar measure, while keeping the same circuit structure as \(\nu\). The spectral gap of Haar random circuits has been extensively studied in previous works for various circuit structures such as 1D local circuits \cite{brandao2016local,haferkamp2022random,chen2024incompressibility}, 1D parallel circuits \cite{brandao2016local}, all-to-all circuits \cite{haferkamp2021improved,chen2024incompressibility}, and circuits with general connectivity \cite{mittal2023local}. In these works, by applying Lemma~\ref{lem:basic_tool_design_formation}, the circuit depths required to form \(\varepsilon\)-approximate \(t\)-designs in Haar random circuits were evaluated as
\begin{equation}
\label{eq:depth_single_layer_connected_Haar}
    L^{\rm H} = (\Delta_{\nu^{\rm H}}^{(t)})^{-1} \cdot (2Nt \log d - \log \varepsilon),
\end{equation}
where \(d\) is the local dimension on each site.

By utilizing $\Delta_{\nu^{\rm H}}^{(t)}$, we can derive the following inequality, which provides the lower bound for $\Delta_{\nu}^{(t)}$:
\begin{equation}
\label{eq:spectral_gap_single_layer}
    \Delta_{\nu}^{(t)} \geq \frac{\Delta_{\mathrm{loc},\nu}^{(t)}}{2} \cdot\Delta_{\nu^{\rm H}}^{(t)}.
\end{equation}
Here, $\Delta_{\mathrm{loc},\nu}^{(t)}$ is defined as the minimum value of the spectral gaps for the two-qudit unitary ensembles within the single layer.
For example, in the case of 1D local random circuit, this value is explicitly described as $\Delta_{\mathrm{loc},\nu}^{(t)} =  \min_{1\leq i < N} \Delta_{i,i+1}^{(t)}$, where $\Delta_{i,i+1}^{(t)}$ is the spectral gap for the two-qudit unitary on the $i$-th and $(i+1)$-th qudits. When all the two-qudit unitaries are drawn from the same ensembles, $\Delta_{\mathrm{loc},\nu}^{(t)}$ is reduced to the spectral gap of that ensemble, and therefore becomes constant independent of the system size $N$. 
Furthermore, Refs.~\cite{varju2013random,oszmaniec2021epsilon} established that \(\Delta_{\mathrm{loc},\nu}^{(t)} \geq \Omega((\log t)^{-2})\) holds when the local ensembles contain universal gate sets.
We note that throughout this Letter, we say that a unitary ensemble ``contains a universal gate set'' if it has support on a universal gate set together with the identity operator. This convention is adopted for technical reasons and reflects the practical feasibility of implementing the identity gate.

From Eq.~\eqref{eq:spectral_gap_single_layer} and Lemma~\ref{lem:basic_tool_design_formation}, we can immediately derive the following theorem:
\vspace{-1.5mm}
\begin{theorem}\label{thm:single_layer_connected_depth}
A single-layer-connected random circuit on an $N$-qudit system forms an $\varepsilon$-approximate unitary $t$-design at a circuit depth 
\begin{equation}
    L = 2 (\Delta_{\mathrm{loc},\nu}^{(t)})^{-1} L^{\rm H},
\end{equation}
where $L^{\rm H}$ is defined as Eq.~\eqref{eq:depth_single_layer_connected_Haar} and $\Delta_{\mathrm{loc},\nu}^{(t)}$ is the minimum value of the spectral gaps of two-qudit unitary ensembles.
\end{theorem}
\vspace{-1.5mm}
\noindent
The essence of the derivations of Eq.~\eqref{eq:spectral_gap_single_layer} and Theorem~\ref{thm:single_layer_connected_depth} is presented in Appendix~\ref{app:1D_local} in the 1D local circuit case, and the complete proof for general single-layer-connected circuit is provided in the Supplemental Material (SM) \cite{supplement}.

\vspace{1mm}\noindent
\textbf{Results for multilayer-connected circuit:}
In circuits with fixed architectures, where the order and positions of the unitary gate applications are predetermined, the number of available gates in a single layer is at most $\frac{N}{2}$, and they cannot form a connected graph as shown in Fig.~\ref{fig:abst_circuit}(b). Therefore, it becomes necessary to evaluate the spectral gap for the multilayer-connected block.
For Haar-random circuits, the spectral gap of such a multilayer block has been successfully lower bounded by relating it to that of a certain single-layer-connected circuit, using the detectability lemma \cite{aharonov2009detectability,anshu2016simple}. However, this lemma is not applicable when the two-qudit gates are drawn from non-Haar ensembles, and thus it remains an open problem to establish a spectral gap lower bound for multilayer blocks composed of non-Haar random gates.

We overcome this technical issue by developing a method to extend the applicability of the detectability lemma. 
For example, applying this method to the 1D brickwork circuit, we can obtain the following lower bound for the spectral gap:
\begin{equation}
\label{eq:spectral_gap_1Dbw}
    \Delta_{\nu_{\rm bw}}^{(t)} \geq \bigl(\Delta_{\mathrm{loc},\nu_{\rm bw}}^{(t)}\bigr)^{2} \cdot \Delta_{\nu_{\rm bw}^{\rm H}}^{(t)}.
\end{equation}
Here, $\nu_{\rm bw}$ and \(\nu_{\rm bw}^{\rm H}\) are the unitary ensembles of the two-layer block for the Haar and non-Haar random 1D brickwork circuits, respectively, and \(\Delta_{\mathrm{loc},\nu_{\rm bw}}^{(t)}\) represents the minimum spectral gap of the two-qudit unitary ensembles within $\nu_{\rm bw}$.

This method is even applicable to a circuit with a generic fixed architecture requiring $l$ layers to cover all sites; this minimal $l$ defines the connection depth.
Defining the unitary ensemble for the connected \(l\)-layer block as \(\nu_A\), we can lower bound its spectral gap as
\begin{equation}
\label{eq:spectral_gap_multi}
    \Delta_{\nu_A}^{(t)} \geq \bigl(\Delta_{\mathrm{loc},\nu_A}^{(t)}\bigr)^l \cdot \mathrm{LB}_{A}^{(t)},
\end{equation}
where \(\Delta_{\mathrm{loc},\nu_A}^{(t)}\) denotes the minimum spectral gap of the two-qudit ensembles within the \(l\)-layer block.
Here, \(\mathrm{LB}_{A}^{(t)}\) is a function of the architecture \(A\) and moment order \(t\), introduced in Ref.~\cite{belkin2023approximate} as a lower bound for the spectral gap of the Haar random circuit with architecture \(A\), i.e., \(\Delta_{\nu_{A}^{\rm H}}^{(t)} \geq \mathrm{LB}_{A}^{(t)}\). It is defined based on the spectral gap of the 1D brickwork circuit, and its explicit form is provided in the SM~\cite{supplement}. By further applying Lemma~\ref{lem:basic_tool_design_formation}, Ref.~\cite{belkin2023approximate} also showed that such a circuit forms an \(\varepsilon\)-approximate unitary \(t\)-design at circuit depth
\begin{equation}
\label{eq:depth_multi_layer_connected_Haar}
    L^{\rm H}_A = \bigl(\mathrm{LB}_{A}^{(t)}\bigr)^{-1} \cdot l (2Nt \log d - \log \varepsilon).
\end{equation}
This is the tightest known depth bound for circuits with generic fixed architecture.

From Eq.~\eqref{eq:spectral_gap_multi} and Lemma~\ref{lem:basic_tool_design_formation}, we can derive the following theorem:
\vspace{-1.5mm}\noindent
\begin{theorem}\label{thm:multilayer_connected_depth}
A multilayer-connected random circuit on an $N$-qudit system with a fixed architecture forms an \(\varepsilon\)-approximate unitary \(t\)-design at a circuit depth
\begin{equation}
    L_A = \bigl(\Delta_{\mathrm{loc},\nu_A}^{(t)}\bigr)^{-l} L^{\rm H}_A,
\end{equation}
where \(L^{\rm H}_A\) is defined as Eq.~\eqref{eq:depth_multi_layer_connected_Haar}, $l$ is the connection depth of $A$, and \(\Delta_{\mathrm{loc},\nu_A}^{(t)}\) is the minimum value of the spectral gaps of two-qudit unitary ensembles.
\end{theorem}
\vspace{-1.5mm}\noindent
The essence of the derivation of Theorem~\ref{thm:multilayer_connected_depth} is provided in Appendix~\ref{app:1D_bw} in the case of 1D brickwork circuit, and the full proof for generic fixed-architecture circuits is presented in the SM~\cite{supplement}.

\vspace{1mm}\noindent
\textbf{Results for patchwork circuit:}
The patchwork circuit is the circuit with the special geometry, where the small patches of normal random circuits are glued together as shown in Fig.~\ref{fig:abst_circuit}(c).
In Ref.~\cite{schuster2024random}, it was proved that such circuits form a unitary design in $O(\log N)$ depth. 
The proof strategy for this depth upper bound is different from the above two cases, where Lemma~\ref{lem:basic_tool_design_formation} and the spectral gap lower bound are crucial; rather, the following lemma plays an essential role:
\vspace{-1.5mm}\noindent
\begin{lemma}[Informal version of Theorem~1 in Ref.~\cite{schuster2024random}] \label{lem:patchwork_essence}
Consider a 1D patchwork circuit on an $N$-qubit system as shown in Fig.~\ref{fig:abst_circuit}(c), where the subsystem size of each small patch is represented as $2\xi$. When $\frac{\varepsilon}{N}$-approximate unitary $t$-design is formed within each small patch, and the inequality $\xi \geq \log_2(Nt^2/\varepsilon)$ is satisfied, this circuit forms an $\varepsilon$-approximate unitary $t$-design on entire $N$-qubit system.
\end{lemma}
\vspace{-1.5mm}\noindent
Roughly speaking, this lemma states that unitary design on the entire system can be formed by gluing up sufficiently random unitaries on small subsystems. 
A simple way to fulfill the requirements of Lemma~\ref{lem:patchwork_essence} is to take $\xi = \log_2(Nt^2/\varepsilon)$ and construct each small patch by the Haar random circuit with sufficient circuit depth $m$.
For example, when the structures of small patches are 1D brickwork, $m$ can be taken as $m = O(\xi t + \log (N/\varepsilon))$, and therefore unitary design on an entire $N$-qubit system is formed within $O(\log N)$-depth.

Such $O(\log N)$-depth formations of unitary designs can be realized even with non-Haar random circuits, since Lemma~\ref{lem:patchwork_essence} does not depend on how each small patch is constructed.
Specifically, it suffices to multiply the depth \(m\) of each small patch by a constant factor; in the case of the 1D brickwork structure, this factor is \((\Delta_{\mathrm{loc,\nu_{\rm bw}}})^{-2}\), which follows from Eq.~\eqref{eq:spectral_gap_1Dbw}. 
The same argument holds even when different structures are used for each patch, by using Theorems~\ref{thm:single_layer_connected_depth} and \ref{thm:multilayer_connected_depth}.

We here note that the definition of an $\varepsilon$-approximate $t$-design used in this section (Lemma~\ref{lem:patchwork_essence}) is not Definition~\ref{def:approx_design}, which is mainly used in the rest part of this Letter, but the relative error $\varepsilon$-approximate unitary $t$-design, which is based on twirling channel.
However, since Definition~\ref{def:approx_design} implies the relative error $\varepsilon$-approximate $t$-design, Theorems~\ref{thm:single_layer_connected_depth} and \ref{thm:multilayer_connected_depth} can be directly used to determine the circuit depth of each small patch.

\vspace{1mm}\noindent
\textbf{Applications and implications of our results:}
%In this work, we show that the unitary design formation depth remains unchanged up to a constant factor even when the local randomizers are varied. Since unitary design formation is crucial for both practical applications and fundamental physics, our results have significant implications in both fields. In the following, we outline three specific applications of our work in quantum many-body physics, computational hardness discussion, and randomized benchmarking.
We here provide three specific applications of our results in quantum many-body physics, computational hardness, and randomized benchmarking, in order.
First, our results imply the universal emergence of chaotic phenomena in random quantum circuits, regardless of the choice of local randomizers. Previous works showed that achieving a unitary design leads to various manifestations of quantum chaos, such as complexity growth \cite{brandao2021models,oszmaniec2022saturation}, information scrambling \cite{roberts2017chaos,nahum2018operator}, and quantum thermalization \cite{choi2023preparing,cotler2023emergent,ho2022exact,kaneko2020characterizing}. However, the robustness of these phenomena under diverse types of local randomness has not been established. Our work demonstrates that these phenomena emerge at the same circuit depth order, independent of the choice of local randomizers.

Second, our results also imply that anticoncentration occurs in \(O(\log N)\)-depth non-Haar random circuits. 
Anticoncentration is an essential ingredient in the hardness argument for random circuit sampling, which is a leading candidate for demonstrating quantum advantage over classical computers. 
In particular, the random circuits implemented in real experiments \cite{boixo2018characterizing,arute2019quantum,zhu2022quantum,gao2025establishing} are composed of non-Haar local randomizers, making it crucial to understand whether anticoncentration holds in such settings.
However, previous works have established \(O(\log N)\)-depth anticoncentration only for Haar-random circuits \cite{dalzell2022random,schuster2024random}.
Our result for patchwork circuits demonstrates that non-Haar random circuits can also exhibit anticoncentration at \(O(\log N)\) depth, based on the fact that approximate unitary 2-designs are known to imply anticoncentration \cite{hangleiter2018anticoncentration}.

%Second, our results also imply that anticoncentration occurs in \(O(\log N)\)-depth non-Haar random circuits. Anticoncentration is an essential ingredient in the hardness argument of random circuit sampling, which is a leading candidate for demonstrating quantum advantage over classical computers \cite{boixo2018characterizing,arute2019quantum,zhu2022quantum,gao2025establishing}. Recently, anticoncentration at shallow circuit depth has been attracted much attention, in the context of near-term quantum device implementations. While a few works have highlighted \(O(\log N)\)-depth anticoncentration in Haar random circuits \cite{dalzell2022random,schuster2024random}, \(O(\log N)\)-depth anticoncentration with general non-Haar local randomizers has not been well-investigated, \red{which casts question on the hardness of some experiments for cross entropy benchmarking~\cite{arute2019quantum, zhu2022quantum, morvan2024phase, gao2025establishing}}. Since an approximate unitary \(2\)-design is known to imply anticoncentration \cite{hangleiter2018anticoncentration}, our result in the patchwork circuit demonstrates that non-Haar local randomizers can also exhibit anticoncentration in \(O(\log N)\) depth, thereby significantly enhancing the flexibility of circuit constructions.

Third, our results allow for greater flexibility in the choice of local gates in randomized benchmarking.
Randomized benchmarking is a widely used technique for assessing the quality of quantum gate implementations \cite{helsen2022general}. In particular, in a scalable protocol known as non-uniform filtered randomized benchmarking, error rates of specific unitary gates are estimated using random circuits composed of those gates.
Reference.~\cite{heinrich2022randomized} showed that the required circuit depth for this protocol can be upper bounded by the inverse of the spectral gap of the corresponding random circuit.
By combining this with our results, Eqs.~\eqref{eq:spectral_gap_single_layer}, \eqref{eq:spectral_gap_1Dbw}, and \eqref{eq:spectral_gap_multi}, we find that the required depth is unaffected by the choice of local gates up to a constant factor, thereby broadening the applicability of the benchmarking protocol.

\vspace{1mm}\noindent
\textbf{Summary and outlook:}
In this work, we reveal that the circuit depths required to form unitary \(t\)-designs in non-Haar random circuits are upper bounded by those for the corresponding Haar-random circuits, up to a constant factor independent of the qudit count. The dependence on the choice of local randomizers appears solely through this constant factor, which is determined by the local spectral gap, as detailed in Theorems~\ref{thm:single_layer_connected_depth} and \ref{thm:multilayer_connected_depth}. These results are established across a wide range of circuit structures, including lattices, non-local geometries, and extremely shallow-depth patchwork architectures. They highlight the robustness of the global randomness generation rate against local gate deformations and pave the way for diverse applications in quantum information science.

Finally, we discuss several future research directions related to our work.
First, it is important to extend the study of flexible unitary design formation to bosonic \cite{aaronson2011computational} and fermionic systems \cite{wan2023matchgate}.
Another promising direction is to clarify how the design formation rate depends on the choice of local randomizers under symmetry constraints \cite{mitsuhashi2024unitary,mitsuhashi2024characterization,liu2024unitary}.
It is also intriguing to explore whether non-Haar local randomizers can enable faster scrambling than their Haar counterparts. While previous works \cite{suzuki2024more,kong2024convergence,riddell2025quantum} have observed such acceleration, a comprehensive understanding of this phenomenon remains an open challenge.

\begin{acknowledgments}
The authors thank Jonas Haferkamp and Soumik Ghosh for the fruitful discussions.
T.Y. is supported by World-leading Innovative Graduate Study Program for Materials Research, Information, and Technology (MERIT-WINGS) of the University of Tokyo. T.Y. is also supported by JSPS KAKENHI Grant No. JP23KJ0672.
R.S. is supported by the BMBF (PhoQuant, Grant Number 13N16103).
Y.M. is supported by JSPS KAKENHI Grant No. JP23KJ0421.
N.Y. is supported by JST PRESTO No. JPMJPR2119, JST Grant Number JPMJPF2221, JST CREST Grant Number JPMJCR23I4, IBM Quantum, JST ASPIRE Grant Number JPMJAP2316, JST ERATO Grant Number JPMJER2302, and Institute of AI and Beyond of the University of Tokyo.
\end{acknowledgments}

\bibliography{bibliography}

\begin{thebibliography}{51}%
\makeatletter
\providecommand \@ifxundefined [1]{%
 \@ifx{#1\undefined}
}%
\providecommand \@ifnum [1]{%
 \ifnum #1\expandafter \@firstoftwo
 \else \expandafter \@secondoftwo
 \fi
}%
\providecommand \@ifx [1]{%
 \ifx #1\expandafter \@firstoftwo
 \else \expandafter \@secondoftwo
 \fi
}%
\providecommand \natexlab [1]{#1}%
\providecommand \enquote  [1]{``#1''}%
\providecommand \bibnamefont  [1]{#1}%
\providecommand \bibfnamefont [1]{#1}%
\providecommand \citenamefont [1]{#1}%
\providecommand \href@noop [0]{\@secondoftwo}%
\providecommand \href [0]{\begingroup \@sanitize@url \@href}%
\providecommand \@href[1]{\@@startlink{#1}\@@href}%
\providecommand \@@href[1]{\endgroup#1\@@endlink}%
\providecommand \@sanitize@url [0]{\catcode `\\12\catcode `\$12\catcode
  `\&12\catcode `\#12\catcode `\^12\catcode `\_12\catcode `\%12\relax}%
\providecommand \@@startlink[1]{}%
\providecommand \@@endlink[0]{}%
\providecommand \url  [0]{\begingroup\@sanitize@url \@url }%
\providecommand \@url [1]{\endgroup\@href {#1}{\urlprefix }}%
\providecommand \urlprefix  [0]{URL }%
\providecommand \Eprint [0]{\href }%
\providecommand \doibase [0]{https://doi.org/}%
\providecommand \selectlanguage [0]{\@gobble}%
\providecommand \bibinfo  [0]{\@secondoftwo}%
\providecommand \bibfield  [0]{\@secondoftwo}%
\providecommand \translation [1]{[#1]}%
\providecommand \BibitemOpen [0]{}%
\providecommand \bibitemStop [0]{}%
\providecommand \bibitemNoStop [0]{.\EOS\space}%
\providecommand \EOS [0]{\spacefactor3000\relax}%
\providecommand \BibitemShut  [1]{\csname bibitem#1\endcsname}%
\let\auto@bib@innerbib\@empty
%</preamble>
\bibitem [{\citenamefont {Helsen}\ \emph {et~al.}(2022)\citenamefont {Helsen},
  \citenamefont {Roth}, \citenamefont {Onorati}, \citenamefont {Werner},\ and\
  \citenamefont {Eisert}}]{helsen2022general}%
  \BibitemOpen
  \bibfield  {author} {\bibinfo {author} {\bibfnamefont {J.}~\bibnamefont
  {Helsen}}, \bibinfo {author} {\bibfnamefont {I.}~\bibnamefont {Roth}},
  \bibinfo {author} {\bibfnamefont {E.}~\bibnamefont {Onorati}}, \bibinfo
  {author} {\bibfnamefont {A.~H.}\ \bibnamefont {Werner}},\ and\ \bibinfo
  {author} {\bibfnamefont {J.}~\bibnamefont {Eisert}},\ }\bibfield  {title}
  {\bibinfo {title} {General framework for randomized benchmarking},\
  }\href@noop {} {\bibfield  {journal} {\bibinfo  {journal} {PRX Quantum}\
  }\textbf {\bibinfo {volume} {3}},\ \bibinfo {pages} {020357} (\bibinfo {year}
  {2022})}\BibitemShut {NoStop}%
\bibitem [{\citenamefont {Elben}\ \emph {et~al.}(2023)\citenamefont {Elben},
  \citenamefont {Flammia}, \citenamefont {Huang}, \citenamefont {Kueng},
  \citenamefont {Preskill}, \citenamefont {Vermersch},\ and\ \citenamefont
  {Zoller}}]{elben2023randomized}%
  \BibitemOpen
  \bibfield  {author} {\bibinfo {author} {\bibfnamefont {A.}~\bibnamefont
  {Elben}}, \bibinfo {author} {\bibfnamefont {S.~T.}\ \bibnamefont {Flammia}},
  \bibinfo {author} {\bibfnamefont {H.-Y.}\ \bibnamefont {Huang}}, \bibinfo
  {author} {\bibfnamefont {R.}~\bibnamefont {Kueng}}, \bibinfo {author}
  {\bibfnamefont {J.}~\bibnamefont {Preskill}}, \bibinfo {author}
  {\bibfnamefont {B.}~\bibnamefont {Vermersch}},\ and\ \bibinfo {author}
  {\bibfnamefont {P.}~\bibnamefont {Zoller}},\ }\bibfield  {title} {\bibinfo
  {title} {The randomized measurement toolbox},\ }\href@noop {} {\bibfield
  {journal} {\bibinfo  {journal} {Nature Reviews Physics}\ }\textbf {\bibinfo
  {volume} {5}},\ \bibinfo {pages} {9} (\bibinfo {year} {2023})}\BibitemShut
  {NoStop}%
\bibitem [{\citenamefont {Huang}\ \emph {et~al.}(2020)\citenamefont {Huang},
  \citenamefont {Kueng},\ and\ \citenamefont {Preskill}}]{huang2020predicting}%
  \BibitemOpen
  \bibfield  {author} {\bibinfo {author} {\bibfnamefont {H.-Y.}\ \bibnamefont
  {Huang}}, \bibinfo {author} {\bibfnamefont {R.}~\bibnamefont {Kueng}},\ and\
  \bibinfo {author} {\bibfnamefont {J.}~\bibnamefont {Preskill}},\ }\bibfield
  {title} {\bibinfo {title} {Predicting many properties of a quantum system
  from very few measurements},\ }\href@noop {} {\bibfield  {journal} {\bibinfo
  {journal} {Nature Physics}\ }\textbf {\bibinfo {volume} {16}},\ \bibinfo
  {pages} {1050} (\bibinfo {year} {2020})}\BibitemShut {NoStop}%
\bibitem [{\citenamefont {Fisher}\ \emph {et~al.}(2023)\citenamefont {Fisher},
  \citenamefont {Khemani}, \citenamefont {Nahum},\ and\ \citenamefont
  {Vijay}}]{fisher2023random}%
  \BibitemOpen
  \bibfield  {author} {\bibinfo {author} {\bibfnamefont {M.~P.}\ \bibnamefont
  {Fisher}}, \bibinfo {author} {\bibfnamefont {V.}~\bibnamefont {Khemani}},
  \bibinfo {author} {\bibfnamefont {A.}~\bibnamefont {Nahum}},\ and\ \bibinfo
  {author} {\bibfnamefont {S.}~\bibnamefont {Vijay}},\ }\bibfield  {title}
  {\bibinfo {title} {Random quantum circuits},\ }\href@noop {} {\bibfield
  {journal} {\bibinfo  {journal} {Annu. Rev. Condens. Matter Phys.}\ }\textbf
  {\bibinfo {volume} {14}},\ \bibinfo {pages} {335} (\bibinfo {year}
  {2023})}\BibitemShut {NoStop}%
\bibitem [{\citenamefont {Nahum}\ \emph {et~al.}(2017)\citenamefont {Nahum},
  \citenamefont {Ruhman}, \citenamefont {Vijay},\ and\ \citenamefont
  {Haah}}]{nahum2017quantum}%
  \BibitemOpen
  \bibfield  {author} {\bibinfo {author} {\bibfnamefont {A.}~\bibnamefont
  {Nahum}}, \bibinfo {author} {\bibfnamefont {J.}~\bibnamefont {Ruhman}},
  \bibinfo {author} {\bibfnamefont {S.}~\bibnamefont {Vijay}},\ and\ \bibinfo
  {author} {\bibfnamefont {J.}~\bibnamefont {Haah}},\ }\bibfield  {title}
  {\bibinfo {title} {Quantum entanglement growth under random unitary
  dynamics},\ }\href@noop {} {\bibfield  {journal} {\bibinfo  {journal} {Phys.
  Rev. X}\ }\textbf {\bibinfo {volume} {7}},\ \bibinfo {pages} {031016}
  (\bibinfo {year} {2017})}\BibitemShut {NoStop}%
\bibitem [{\citenamefont {Harrow}\ and\ \citenamefont
  {Low}(2009)}]{harrow2009random}%
  \BibitemOpen
  \bibfield  {author} {\bibinfo {author} {\bibfnamefont {A.~W.}\ \bibnamefont
  {Harrow}}\ and\ \bibinfo {author} {\bibfnamefont {R.~A.}\ \bibnamefont
  {Low}},\ }\bibfield  {title} {\bibinfo {title} {Random quantum circuits are
  approximate 2-designs},\ }\href@noop {} {\bibfield  {journal} {\bibinfo
  {journal} {Commun. Math. Phys.}\ }\textbf {\bibinfo {volume} {291}},\
  \bibinfo {pages} {257} (\bibinfo {year} {2009})}\BibitemShut {NoStop}%
\bibitem [{\citenamefont {Brown}\ and\ \citenamefont
  {Viola}(2010)}]{brown2010convergence}%
  \BibitemOpen
  \bibfield  {author} {\bibinfo {author} {\bibfnamefont {W.~G.}\ \bibnamefont
  {Brown}}\ and\ \bibinfo {author} {\bibfnamefont {L.}~\bibnamefont {Viola}},\
  }\bibfield  {title} {\bibinfo {title} {Convergence rates for arbitrary
  statistical moments of random quantum circuits},\ }\href@noop {} {\bibfield
  {journal} {\bibinfo  {journal} {Phys. Rev. Lett.}\ }\textbf {\bibinfo
  {volume} {104}},\ \bibinfo {pages} {250501} (\bibinfo {year}
  {2010})}\BibitemShut {NoStop}%
\bibitem [{\citenamefont {Brand{\~a}o}\ \emph {et~al.}(2016)\citenamefont
  {Brand{\~a}o}, \citenamefont {Harrow},\ and\ \citenamefont
  {Horodecki}}]{brandao2016local}%
  \BibitemOpen
  \bibfield  {author} {\bibinfo {author} {\bibfnamefont {F.~G.}\ \bibnamefont
  {Brand{\~a}o}}, \bibinfo {author} {\bibfnamefont {A.~W.}\ \bibnamefont
  {Harrow}},\ and\ \bibinfo {author} {\bibfnamefont {M.}~\bibnamefont
  {Horodecki}},\ }\bibfield  {title} {\bibinfo {title} {Local random quantum
  circuits are approximate polynomial-designs},\ }\href@noop {} {\bibfield
  {journal} {\bibinfo  {journal} {Commun. Math. Phys.}\ }\textbf {\bibinfo
  {volume} {346}},\ \bibinfo {pages} {397} (\bibinfo {year}
  {2016})}\BibitemShut {NoStop}%
\bibitem [{\citenamefont {Haferkamp}(2022)}]{haferkamp2022random}%
  \BibitemOpen
  \bibfield  {author} {\bibinfo {author} {\bibfnamefont {J.}~\bibnamefont
  {Haferkamp}},\ }\bibfield  {title} {\bibinfo {title} {Random quantum circuits
  are approximate unitary {$t$}-designs in depth
  {$O\left(nt^{5+o(1)}\right)$}},\ }\href@noop {} {\bibfield  {journal}
  {\bibinfo  {journal} {{Quantum}}\ }\textbf {\bibinfo {volume} {6}},\ \bibinfo
  {pages} {795} (\bibinfo {year} {2022})}\BibitemShut {NoStop}%
\bibitem [{\citenamefont {Chen}\ \emph {et~al.}(2024)\citenamefont {Chen},
  \citenamefont {Haah}, \citenamefont {Haferkamp}, \citenamefont {Liu},
  \citenamefont {Metger},\ and\ \citenamefont
  {Tan}}]{chen2024incompressibility}%
  \BibitemOpen
  \bibfield  {author} {\bibinfo {author} {\bibfnamefont {C.-F.}\ \bibnamefont
  {Chen}}, \bibinfo {author} {\bibfnamefont {J.}~\bibnamefont {Haah}}, \bibinfo
  {author} {\bibfnamefont {J.}~\bibnamefont {Haferkamp}}, \bibinfo {author}
  {\bibfnamefont {Y.}~\bibnamefont {Liu}}, \bibinfo {author} {\bibfnamefont
  {T.}~\bibnamefont {Metger}},\ and\ \bibinfo {author} {\bibfnamefont
  {X.}~\bibnamefont {Tan}},\ }\bibfield  {title} {\bibinfo {title}
  {Incompressibility and spectral gaps of random circuits},\ }\href@noop {}
  {\bibfield  {journal} {\bibinfo  {journal} {arXiv preprint arXiv:2406.07478}\
  } (\bibinfo {year} {2024})}\BibitemShut {NoStop}%
\bibitem [{\citenamefont {Haferkamp}\ and\ \citenamefont
  {Hunter-Jones}(2021)}]{haferkamp2021improved}%
  \BibitemOpen
  \bibfield  {author} {\bibinfo {author} {\bibfnamefont {J.}~\bibnamefont
  {Haferkamp}}\ and\ \bibinfo {author} {\bibfnamefont {N.}~\bibnamefont
  {Hunter-Jones}},\ }\bibfield  {title} {\bibinfo {title} {Improved spectral
  gaps for random quantum circuits: Large local dimensions and all-to-all
  interactions},\ }\href@noop {} {\bibfield  {journal} {\bibinfo  {journal}
  {Phys. Rev. A}\ }\textbf {\bibinfo {volume} {104}},\ \bibinfo {pages}
  {022417} (\bibinfo {year} {2021})}\BibitemShut {NoStop}%
\bibitem [{\citenamefont {Hunter-Jones}(2019)}]{hunter2019unitary}%
  \BibitemOpen
  \bibfield  {author} {\bibinfo {author} {\bibfnamefont {N.}~\bibnamefont
  {Hunter-Jones}},\ }\bibfield  {title} {\bibinfo {title} {Unitary designs from
  statistical mechanics in random quantum circuits},\ }\href@noop {} {\bibfield
   {journal} {\bibinfo  {journal} {arXiv preprint arXiv:1905.12053}\ }
  (\bibinfo {year} {2019})}\BibitemShut {NoStop}%
\bibitem [{\citenamefont {Harrow}\ and\ \citenamefont
  {Mehraban}(2023)}]{harrow2023approximate}%
  \BibitemOpen
  \bibfield  {author} {\bibinfo {author} {\bibfnamefont {A.~W.}\ \bibnamefont
  {Harrow}}\ and\ \bibinfo {author} {\bibfnamefont {S.}~\bibnamefont
  {Mehraban}},\ }\bibfield  {title} {\bibinfo {title} {Approximate unitary
  t-designs by short random quantum circuits using nearest-neighbor and
  long-range gates},\ }\href@noop {} {\bibfield  {journal} {\bibinfo  {journal}
  {Commun. Math. Phys.}\ }\textbf {\bibinfo {volume} {401}},\ \bibinfo {pages}
  {1531} (\bibinfo {year} {2023})}\BibitemShut {NoStop}%
\bibitem [{\citenamefont {Mittal}\ and\ \citenamefont
  {Hunter-Jones}(2023)}]{mittal2023local}%
  \BibitemOpen
  \bibfield  {author} {\bibinfo {author} {\bibfnamefont {S.}~\bibnamefont
  {Mittal}}\ and\ \bibinfo {author} {\bibfnamefont {N.}~\bibnamefont
  {Hunter-Jones}},\ }\bibfield  {title} {\bibinfo {title} {Local random quantum
  circuits form approximate designs on arbitrary architectures},\ }\href@noop
  {} {\bibfield  {journal} {\bibinfo  {journal} {arXiv preprint
  arXiv:2310.19355}\ } (\bibinfo {year} {2023})}\BibitemShut {NoStop}%
\bibitem [{\citenamefont {Belkin}\ \emph {et~al.}(2024)\citenamefont {Belkin},
  \citenamefont {Allen}, \citenamefont {Ghosh}, \citenamefont {Kang},
  \citenamefont {Lin}, \citenamefont {Sud}, \citenamefont {Chong},
  \citenamefont {Fefferman},\ and\ \citenamefont
  {Clark}}]{belkin2023approximate}%
  \BibitemOpen
  \bibfield  {author} {\bibinfo {author} {\bibfnamefont {D.}~\bibnamefont
  {Belkin}}, \bibinfo {author} {\bibfnamefont {J.}~\bibnamefont {Allen}},
  \bibinfo {author} {\bibfnamefont {S.}~\bibnamefont {Ghosh}}, \bibinfo
  {author} {\bibfnamefont {C.}~\bibnamefont {Kang}}, \bibinfo {author}
  {\bibfnamefont {S.}~\bibnamefont {Lin}}, \bibinfo {author} {\bibfnamefont
  {J.}~\bibnamefont {Sud}}, \bibinfo {author} {\bibfnamefont {F.~T.}\
  \bibnamefont {Chong}}, \bibinfo {author} {\bibfnamefont {B.}~\bibnamefont
  {Fefferman}},\ and\ \bibinfo {author} {\bibfnamefont {B.~K.}\ \bibnamefont
  {Clark}},\ }\bibfield  {title} {\bibinfo {title} {Approximate t-designs in
  generic circuit architectures},\ }\href@noop {} {\bibfield  {journal}
  {\bibinfo  {journal} {PRX Quantum}\ }\textbf {\bibinfo {volume} {5}},\
  \bibinfo {pages} {040344} (\bibinfo {year} {2024})}\BibitemShut {NoStop}%
\bibitem [{\citenamefont {Schuster}\ \emph {et~al.}(2024)\citenamefont
  {Schuster}, \citenamefont {Haferkamp},\ and\ \citenamefont
  {Huang}}]{schuster2024random}%
  \BibitemOpen
  \bibfield  {author} {\bibinfo {author} {\bibfnamefont {T.}~\bibnamefont
  {Schuster}}, \bibinfo {author} {\bibfnamefont {J.}~\bibnamefont
  {Haferkamp}},\ and\ \bibinfo {author} {\bibfnamefont {H.-Y.}\ \bibnamefont
  {Huang}},\ }\bibfield  {title} {\bibinfo {title} {Random unitaries in
  extremely low depth},\ }\href@noop {} {\bibfield  {journal} {\bibinfo
  {journal} {arXiv preprint arXiv:2407.07754}\ } (\bibinfo {year}
  {2024})}\BibitemShut {NoStop}%
\bibitem [{\citenamefont {LaRacuente}\ and\ \citenamefont
  {Leditzky}(2024)}]{laracuente2024approximate}%
  \BibitemOpen
  \bibfield  {author} {\bibinfo {author} {\bibfnamefont {N.}~\bibnamefont
  {LaRacuente}}\ and\ \bibinfo {author} {\bibfnamefont {F.}~\bibnamefont
  {Leditzky}},\ }\bibfield  {title} {\bibinfo {title} {Approximate unitary $ k
  $-designs from shallow, low-communication circuits},\ }\href@noop {}
  {\bibfield  {journal} {\bibinfo  {journal} {arXiv preprint arXiv:2407.07876}\
  } (\bibinfo {year} {2024})}\BibitemShut {NoStop}%
\bibitem [{\citenamefont {Webb}(2015)}]{webb2015clifford}%
  \BibitemOpen
  \bibfield  {author} {\bibinfo {author} {\bibfnamefont {Z.}~\bibnamefont
  {Webb}},\ }\bibfield  {title} {\bibinfo {title} {The clifford group forms a
  unitary 3-design},\ }\href@noop {} {\bibfield  {journal} {\bibinfo  {journal}
  {arXiv preprint arXiv:1510.02769}\ } (\bibinfo {year} {2015})}\BibitemShut
  {NoStop}%
\bibitem [{\citenamefont {Zhu}(2017)}]{zhu2017multiqubit}%
  \BibitemOpen
  \bibfield  {author} {\bibinfo {author} {\bibfnamefont {H.}~\bibnamefont
  {Zhu}},\ }\bibfield  {title} {\bibinfo {title} {Multiqubit clifford groups
  are unitary 3-designs},\ }\href@noop {} {\bibfield  {journal} {\bibinfo
  {journal} {Phys. Rev. A}\ }\textbf {\bibinfo {volume} {96}},\ \bibinfo
  {pages} {062336} (\bibinfo {year} {2017})}\BibitemShut {NoStop}%
\bibitem [{\citenamefont {Boixo}\ \emph {et~al.}(2018)\citenamefont {Boixo},
  \citenamefont {Isakov}, \citenamefont {Smelyanskiy}, \citenamefont {Babbush},
  \citenamefont {Ding}, \citenamefont {Jiang}, \citenamefont {Bremner},
  \citenamefont {Martinis},\ and\ \citenamefont
  {Neven}}]{boixo2018characterizing}%
  \BibitemOpen
  \bibfield  {author} {\bibinfo {author} {\bibfnamefont {S.}~\bibnamefont
  {Boixo}}, \bibinfo {author} {\bibfnamefont {S.~V.}\ \bibnamefont {Isakov}},
  \bibinfo {author} {\bibfnamefont {V.~N.}\ \bibnamefont {Smelyanskiy}},
  \bibinfo {author} {\bibfnamefont {R.}~\bibnamefont {Babbush}}, \bibinfo
  {author} {\bibfnamefont {N.}~\bibnamefont {Ding}}, \bibinfo {author}
  {\bibfnamefont {Z.}~\bibnamefont {Jiang}}, \bibinfo {author} {\bibfnamefont
  {M.~J.}\ \bibnamefont {Bremner}}, \bibinfo {author} {\bibfnamefont {J.~M.}\
  \bibnamefont {Martinis}},\ and\ \bibinfo {author} {\bibfnamefont
  {H.}~\bibnamefont {Neven}},\ }\bibfield  {title} {\bibinfo {title}
  {Characterizing quantum supremacy in near-term devices},\ }\href@noop {}
  {\bibfield  {journal} {\bibinfo  {journal} {Nature Physics}\ }\textbf
  {\bibinfo {volume} {14}},\ \bibinfo {pages} {595} (\bibinfo {year}
  {2018})}\BibitemShut {NoStop}%
\bibitem [{\citenamefont {Arute}\ \emph {et~al.}(2019)\citenamefont {Arute},
  \citenamefont {Arya}, \citenamefont {Babbush}, \citenamefont {Bacon},
  \citenamefont {Bardin}, \citenamefont {Barends}, \citenamefont {Biswas},
  \citenamefont {Boixo}, \citenamefont {Brand{\~a}o}, \citenamefont {Buell}
  \emph {et~al.}}]{arute2019quantum}%
  \BibitemOpen
  \bibfield  {author} {\bibinfo {author} {\bibfnamefont {F.}~\bibnamefont
  {Arute}}, \bibinfo {author} {\bibfnamefont {K.}~\bibnamefont {Arya}},
  \bibinfo {author} {\bibfnamefont {R.}~\bibnamefont {Babbush}}, \bibinfo
  {author} {\bibfnamefont {D.}~\bibnamefont {Bacon}}, \bibinfo {author}
  {\bibfnamefont {J.~C.}\ \bibnamefont {Bardin}}, \bibinfo {author}
  {\bibfnamefont {R.}~\bibnamefont {Barends}}, \bibinfo {author} {\bibfnamefont
  {R.}~\bibnamefont {Biswas}}, \bibinfo {author} {\bibfnamefont
  {S.}~\bibnamefont {Boixo}}, \bibinfo {author} {\bibfnamefont {F.~G. S.~L.}\
  \bibnamefont {Brand{\~a}o}}, \bibinfo {author} {\bibfnamefont {D.~A.}\
  \bibnamefont {Buell}}, \emph {et~al.},\ }\bibfield  {title} {\bibinfo {title}
  {Quantum supremacy using a programmable superconducting processor},\
  }\href@noop {} {\bibfield  {journal} {\bibinfo  {journal} {Nature}\ }\textbf
  {\bibinfo {volume} {574}},\ \bibinfo {pages} {505} (\bibinfo {year}
  {2019})}\BibitemShut {NoStop}%
\bibitem [{\citenamefont {Morvan}\ \emph {et~al.}(2024)\citenamefont {Morvan},
  \citenamefont {Villalonga}, \citenamefont {Mi}, \citenamefont {Mandrà},
  \citenamefont {Bengtsson}, \citenamefont {Klimov}, \citenamefont {Chen},
  \citenamefont {Hong}, \citenamefont {Erickson}, \citenamefont {Drozdov} \emph
  {et~al.}}]{morvan2024phase}%
  \BibitemOpen
  \bibfield  {author} {\bibinfo {author} {\bibfnamefont {A.}~\bibnamefont
  {Morvan}}, \bibinfo {author} {\bibfnamefont {B.}~\bibnamefont {Villalonga}},
  \bibinfo {author} {\bibfnamefont {X.}~\bibnamefont {Mi}}, \bibinfo {author}
  {\bibfnamefont {S.}~\bibnamefont {Mandrà}}, \bibinfo {author} {\bibfnamefont
  {A.}~\bibnamefont {Bengtsson}}, \bibinfo {author} {\bibfnamefont {P.~V.}\
  \bibnamefont {Klimov}}, \bibinfo {author} {\bibfnamefont {Z.}~\bibnamefont
  {Chen}}, \bibinfo {author} {\bibfnamefont {S.}~\bibnamefont {Hong}}, \bibinfo
  {author} {\bibfnamefont {C.}~\bibnamefont {Erickson}}, \bibinfo {author}
  {\bibfnamefont {I.~K.}\ \bibnamefont {Drozdov}}, \emph {et~al.},\ }\bibfield
  {title} {\bibinfo {title} {Phase transitions in random circuit sampling},\
  }\href@noop {} {\bibfield  {journal} {\bibinfo  {journal} {Nature}\ }\textbf
  {\bibinfo {volume} {634}},\ \bibinfo {pages} {328} (\bibinfo {year}
  {2024})}\BibitemShut {NoStop}%
\bibitem [{\citenamefont {Zhu}\ \emph {et~al.}(2022)\citenamefont {Zhu},
  \citenamefont {Cao}, \citenamefont {Chen}, \citenamefont {Chen},
  \citenamefont {Chen}, \citenamefont {Chung}, \citenamefont {Deng},
  \citenamefont {Du}, \citenamefont {Fan}, \citenamefont {Gong} \emph
  {et~al.}}]{zhu2022quantum}%
  \BibitemOpen
  \bibfield  {author} {\bibinfo {author} {\bibfnamefont {Q.}~\bibnamefont
  {Zhu}}, \bibinfo {author} {\bibfnamefont {S.}~\bibnamefont {Cao}}, \bibinfo
  {author} {\bibfnamefont {F.}~\bibnamefont {Chen}}, \bibinfo {author}
  {\bibfnamefont {M.-C.}\ \bibnamefont {Chen}}, \bibinfo {author}
  {\bibfnamefont {X.}~\bibnamefont {Chen}}, \bibinfo {author} {\bibfnamefont
  {T.-H.}\ \bibnamefont {Chung}}, \bibinfo {author} {\bibfnamefont
  {H.}~\bibnamefont {Deng}}, \bibinfo {author} {\bibfnamefont {Y.}~\bibnamefont
  {Du}}, \bibinfo {author} {\bibfnamefont {D.}~\bibnamefont {Fan}}, \bibinfo
  {author} {\bibfnamefont {M.}~\bibnamefont {Gong}}, \emph {et~al.},\
  }\bibfield  {title} {\bibinfo {title} {Quantum computational advantage via
  60-qubit 24-cycle random circuit sampling},\ }\href@noop {} {\bibfield
  {journal} {\bibinfo  {journal} {Science bulletin}\ }\textbf {\bibinfo
  {volume} {67}},\ \bibinfo {pages} {240} (\bibinfo {year} {2022})}\BibitemShut
  {NoStop}%
\bibitem [{\citenamefont {Gao}\ \emph {et~al.}(2025)\citenamefont {Gao},
  \citenamefont {Fan}, \citenamefont {Zha}, \citenamefont {Bei}, \citenamefont
  {Cai}, \citenamefont {Cai}, \citenamefont {Cao}, \citenamefont {Chen},
  \citenamefont {Chen}, \citenamefont {Chen} \emph
  {et~al.}}]{gao2025establishing}%
  \BibitemOpen
  \bibfield  {author} {\bibinfo {author} {\bibfnamefont {D.}~\bibnamefont
  {Gao}}, \bibinfo {author} {\bibfnamefont {D.}~\bibnamefont {Fan}}, \bibinfo
  {author} {\bibfnamefont {C.}~\bibnamefont {Zha}}, \bibinfo {author}
  {\bibfnamefont {J.}~\bibnamefont {Bei}}, \bibinfo {author} {\bibfnamefont
  {G.}~\bibnamefont {Cai}}, \bibinfo {author} {\bibfnamefont {J.}~\bibnamefont
  {Cai}}, \bibinfo {author} {\bibfnamefont {S.}~\bibnamefont {Cao}}, \bibinfo
  {author} {\bibfnamefont {F.}~\bibnamefont {Chen}}, \bibinfo {author}
  {\bibfnamefont {J.}~\bibnamefont {Chen}}, \bibinfo {author} {\bibfnamefont
  {K.}~\bibnamefont {Chen}}, \emph {et~al.},\ }\bibfield  {title} {\bibinfo
  {title} {Establishing a new benchmark in quantum computational advantage with
  105-qubit zuchongzhi 3.0 processor},\ }\href@noop {} {\bibfield  {journal}
  {\bibinfo  {journal} {Phys. Rev. Lett.}\ }\textbf {\bibinfo {volume} {134}},\
  \bibinfo {pages} {090601} (\bibinfo {year} {2025})}\BibitemShut {NoStop}%
\bibitem [{\citenamefont {Oszmaniec}\ \emph {et~al.}(2021)\citenamefont
  {Oszmaniec}, \citenamefont {Sawicki},\ and\ \citenamefont
  {Horodecki}}]{oszmaniec2021epsilon}%
  \BibitemOpen
  \bibfield  {author} {\bibinfo {author} {\bibfnamefont {M.}~\bibnamefont
  {Oszmaniec}}, \bibinfo {author} {\bibfnamefont {A.}~\bibnamefont {Sawicki}},\
  and\ \bibinfo {author} {\bibfnamefont {M.}~\bibnamefont {Horodecki}},\
  }\bibfield  {title} {\bibinfo {title} {Epsilon-nets, unitary designs, and
  random quantum circuits},\ }\href@noop {} {\bibfield  {journal} {\bibinfo
  {journal} {IEEE Trans. Inf. Theory.}\ }\textbf {\bibinfo {volume} {68}},\
  \bibinfo {pages} {989} (\bibinfo {year} {2021})}\BibitemShut {NoStop}%
\bibitem [{\citenamefont {Oszmaniec}\ \emph {et~al.}(2024)\citenamefont
  {Oszmaniec}, \citenamefont {Kotowski}, \citenamefont {Horodecki},\ and\
  \citenamefont {Hunter-Jones}}]{oszmaniec2022saturation}%
  \BibitemOpen
  \bibfield  {author} {\bibinfo {author} {\bibfnamefont {M.}~\bibnamefont
  {Oszmaniec}}, \bibinfo {author} {\bibfnamefont {M.}~\bibnamefont {Kotowski}},
  \bibinfo {author} {\bibfnamefont {M.}~\bibnamefont {Horodecki}},\ and\
  \bibinfo {author} {\bibfnamefont {N.}~\bibnamefont {Hunter-Jones}},\
  }\bibfield  {title} {\bibinfo {title} {Saturation and recurrence of quantum
  complexity in random local quantum dynamics},\ }\href@noop {} {\bibfield
  {journal} {\bibinfo  {journal} {Phys. Rev. X}\ }\textbf {\bibinfo {volume}
  {14}},\ \bibinfo {pages} {041068} (\bibinfo {year} {2024})}\BibitemShut
  {NoStop}%
\bibitem [{\citenamefont {Roberts}\ and\ \citenamefont
  {Yoshida}(2017)}]{roberts2017chaos}%
  \BibitemOpen
  \bibfield  {author} {\bibinfo {author} {\bibfnamefont {D.~A.}\ \bibnamefont
  {Roberts}}\ and\ \bibinfo {author} {\bibfnamefont {B.}~\bibnamefont
  {Yoshida}},\ }\bibfield  {title} {\bibinfo {title} {Chaos and complexity by
  design},\ }\href@noop {} {\bibfield  {journal} {\bibinfo  {journal} {J. High
  Energ. Phys.}\ }\textbf {\bibinfo {volume} {2017}}\bibinfo  {number} {
  (04)},\ \bibinfo {pages} {121}}\BibitemShut {NoStop}%
\bibitem [{\citenamefont {Nahum}\ \emph {et~al.}(2018)\citenamefont {Nahum},
  \citenamefont {Vijay},\ and\ \citenamefont {Haah}}]{nahum2018operator}%
  \BibitemOpen
\bibfield  {number} {  }\bibfield  {author} {\bibinfo {author} {\bibfnamefont
  {A.}~\bibnamefont {Nahum}}, \bibinfo {author} {\bibfnamefont
  {S.}~\bibnamefont {Vijay}},\ and\ \bibinfo {author} {\bibfnamefont
  {J.}~\bibnamefont {Haah}},\ }\bibfield  {title} {\bibinfo {title} {Operator
  spreading in random unitary circuits},\ }\href@noop {} {\bibfield  {journal}
  {\bibinfo  {journal} {Phys. Rev. X}\ }\textbf {\bibinfo {volume} {8}},\
  \bibinfo {pages} {021014} (\bibinfo {year} {2018})}\BibitemShut {NoStop}%
\bibitem [{\citenamefont {von Keyserlingk}\ \emph {et~al.}(2018)\citenamefont
  {von Keyserlingk}, \citenamefont {Rakovszky}, \citenamefont {Pollmann},\ and\
  \citenamefont {Sondhi}}]{von2018operator}%
  \BibitemOpen
  \bibfield  {author} {\bibinfo {author} {\bibfnamefont {C.~W.}\ \bibnamefont
  {von Keyserlingk}}, \bibinfo {author} {\bibfnamefont {T.}~\bibnamefont
  {Rakovszky}}, \bibinfo {author} {\bibfnamefont {F.}~\bibnamefont
  {Pollmann}},\ and\ \bibinfo {author} {\bibfnamefont {S.~L.}\ \bibnamefont
  {Sondhi}},\ }\bibfield  {title} {\bibinfo {title} {Operator hydrodynamics,
  otocs, and entanglement growth in systems without conservation laws},\
  }\href@noop {} {\bibfield  {journal} {\bibinfo  {journal} {Phys. Rev. X}\
  }\textbf {\bibinfo {volume} {8}},\ \bibinfo {pages} {021013} (\bibinfo {year}
  {2018})}\BibitemShut {NoStop}%
\bibitem [{\citenamefont {Brand{\~a}o}\ \emph {et~al.}(2021)\citenamefont
  {Brand{\~a}o}, \citenamefont {Chemissany}, \citenamefont {Hunter-Jones},
  \citenamefont {Kueng},\ and\ \citenamefont {Preskill}}]{brandao2021models}%
  \BibitemOpen
  \bibfield  {author} {\bibinfo {author} {\bibfnamefont {F.~G.}\ \bibnamefont
  {Brand{\~a}o}}, \bibinfo {author} {\bibfnamefont {W.}~\bibnamefont
  {Chemissany}}, \bibinfo {author} {\bibfnamefont {N.}~\bibnamefont
  {Hunter-Jones}}, \bibinfo {author} {\bibfnamefont {R.}~\bibnamefont
  {Kueng}},\ and\ \bibinfo {author} {\bibfnamefont {J.}~\bibnamefont
  {Preskill}},\ }\bibfield  {title} {\bibinfo {title} {Models of quantum
  complexity growth},\ }\href@noop {} {\bibfield  {journal} {\bibinfo
  {journal} {PRX Quantum}\ }\textbf {\bibinfo {volume} {2}},\ \bibinfo {pages}
  {030316} (\bibinfo {year} {2021})}\BibitemShut {NoStop}%
\bibitem [{\citenamefont {Haferkamp}\ \emph {et~al.}(2022)\citenamefont
  {Haferkamp}, \citenamefont {Faist}, \citenamefont {Kothakonda}, \citenamefont
  {Eisert},\ and\ \citenamefont {Yunger~Halpern}}]{haferkamp2022linear}%
  \BibitemOpen
  \bibfield  {author} {\bibinfo {author} {\bibfnamefont {J.}~\bibnamefont
  {Haferkamp}}, \bibinfo {author} {\bibfnamefont {P.}~\bibnamefont {Faist}},
  \bibinfo {author} {\bibfnamefont {N.~B.}\ \bibnamefont {Kothakonda}},
  \bibinfo {author} {\bibfnamefont {J.}~\bibnamefont {Eisert}},\ and\ \bibinfo
  {author} {\bibfnamefont {N.}~\bibnamefont {Yunger~Halpern}},\ }\bibfield
  {title} {\bibinfo {title} {Linear growth of quantum circuit complexity},\
  }\href@noop {} {\bibfield  {journal} {\bibinfo  {journal} {Nature Physics}\
  }\textbf {\bibinfo {volume} {18}},\ \bibinfo {pages} {528} (\bibinfo {year}
  {2022})}\BibitemShut {NoStop}%
\bibitem [{\citenamefont {Mele}(2024)}]{mele2024introduction}%
  \BibitemOpen
  \bibfield  {author} {\bibinfo {author} {\bibfnamefont {A.~A.}\ \bibnamefont
  {Mele}},\ }\bibfield  {title} {\bibinfo {title} {Introduction to haar measure
  tools in quantum information: A beginner's tutorial},\ }\href@noop {}
  {\bibfield  {journal} {\bibinfo  {journal} {Quantum}\ }\textbf {\bibinfo
  {volume} {8}},\ \bibinfo {pages} {1340} (\bibinfo {year} {2024})}\BibitemShut
  {NoStop}%
\bibitem [{\citenamefont {Varj{\'u}}(2013)}]{varju2013random}%
  \BibitemOpen
  \bibfield  {author} {\bibinfo {author} {\bibfnamefont {P.~P.}\ \bibnamefont
  {Varj{\'u}}},\ }\bibfield  {title} {\bibinfo {title} {Random walks in compact
  groups},\ }\href@noop {} {\bibfield  {journal} {\bibinfo  {journal}
  {Documenta Mathematica}\ }\textbf {\bibinfo {volume} {18}},\ \bibinfo {pages}
  {1137} (\bibinfo {year} {2013})}\BibitemShut {NoStop}%
\bibitem [{sup()}]{supplement}%
  \BibitemOpen
  \href@noop {} {\bibinfo  {journal} {See Supplemental Material (SM) for the
  details}\ }\BibitemShut {NoStop}%
\bibitem [{\citenamefont {Aharonov}\ \emph {et~al.}(2009)\citenamefont
  {Aharonov}, \citenamefont {Arad}, \citenamefont {Landau},\ and\ \citenamefont
  {Vazirani}}]{aharonov2009detectability}%
  \BibitemOpen
\bibfield  {journal} {  }\bibfield  {author} {\bibinfo {author} {\bibfnamefont
  {D.}~\bibnamefont {Aharonov}}, \bibinfo {author} {\bibfnamefont
  {I.}~\bibnamefont {Arad}}, \bibinfo {author} {\bibfnamefont {Z.}~\bibnamefont
  {Landau}},\ and\ \bibinfo {author} {\bibfnamefont {U.}~\bibnamefont
  {Vazirani}},\ }\bibfield  {title} {\bibinfo {title} {The detectability lemma
  and quantum gap amplification},\ }in\ \href@noop {} {\emph {\bibinfo
  {booktitle} {Proceedings of the forty-first annual ACM symposium on Theory of
  computing}}}\ (\bibinfo {year} {2009})\ pp.\ \bibinfo {pages}
  {417--426}\BibitemShut {NoStop}%
\bibitem [{\citenamefont {Anshu}\ \emph {et~al.}(2016)\citenamefont {Anshu},
  \citenamefont {Arad},\ and\ \citenamefont {Vidick}}]{anshu2016simple}%
  \BibitemOpen
  \bibfield  {author} {\bibinfo {author} {\bibfnamefont {A.}~\bibnamefont
  {Anshu}}, \bibinfo {author} {\bibfnamefont {I.}~\bibnamefont {Arad}},\ and\
  \bibinfo {author} {\bibfnamefont {T.}~\bibnamefont {Vidick}},\ }\bibfield
  {title} {\bibinfo {title} {Simple proof of the detectability lemma and
  spectral gap amplification},\ }\href@noop {} {\bibfield  {journal} {\bibinfo
  {journal} {Phys. Rev. B}\ }\textbf {\bibinfo {volume} {93}},\ \bibinfo
  {pages} {205142} (\bibinfo {year} {2016})}\BibitemShut {NoStop}%
\bibitem [{\citenamefont {Choi}\ \emph {et~al.}(2023)\citenamefont {Choi},
  \citenamefont {Shaw}, \citenamefont {Madjarov}, \citenamefont {Xie},
  \citenamefont {Finkelstein}, \citenamefont {Covey}, \citenamefont {Cotler},
  \citenamefont {Mark}, \citenamefont {Huang}, \citenamefont {Kale},
  \citenamefont {Pichler}, \citenamefont {Brandão}, \citenamefont {Choi},\
  and\ \citenamefont {Endres}}]{choi2023preparing}%
  \BibitemOpen
  \bibfield  {author} {\bibinfo {author} {\bibfnamefont {J.}~\bibnamefont
  {Choi}}, \bibinfo {author} {\bibfnamefont {A.~L.}\ \bibnamefont {Shaw}},
  \bibinfo {author} {\bibfnamefont {I.~S.}\ \bibnamefont {Madjarov}}, \bibinfo
  {author} {\bibfnamefont {X.}~\bibnamefont {Xie}}, \bibinfo {author}
  {\bibfnamefont {R.}~\bibnamefont {Finkelstein}}, \bibinfo {author}
  {\bibfnamefont {J.~P.}\ \bibnamefont {Covey}}, \bibinfo {author}
  {\bibfnamefont {J.~S.}\ \bibnamefont {Cotler}}, \bibinfo {author}
  {\bibfnamefont {D.~K.}\ \bibnamefont {Mark}}, \bibinfo {author}
  {\bibfnamefont {H.-Y.}\ \bibnamefont {Huang}}, \bibinfo {author}
  {\bibfnamefont {A.}~\bibnamefont {Kale}}, \bibinfo {author} {\bibfnamefont
  {H.}~\bibnamefont {Pichler}}, \bibinfo {author} {\bibfnamefont {F.~G. S.~L.}\
  \bibnamefont {Brandão}}, \bibinfo {author} {\bibfnamefont {S.}~\bibnamefont
  {Choi}},\ and\ \bibinfo {author} {\bibfnamefont {M.}~\bibnamefont {Endres}},\
  }\bibfield  {title} {\bibinfo {title} {Preparing random states and
  benchmarking with many-body quantum chaos},\ }\href@noop {} {\bibfield
  {journal} {\bibinfo  {journal} {Nature}\ }\textbf {\bibinfo {volume} {613}},\
  \bibinfo {pages} {468} (\bibinfo {year} {2023})}\BibitemShut {NoStop}%
\bibitem [{\citenamefont {Cotler}\ \emph {et~al.}(2023)\citenamefont {Cotler},
  \citenamefont {Mark}, \citenamefont {Huang}, \citenamefont {Hernandez},
  \citenamefont {Choi}, \citenamefont {Shaw}, \citenamefont {Endres},\ and\
  \citenamefont {Choi}}]{cotler2023emergent}%
  \BibitemOpen
  \bibfield  {author} {\bibinfo {author} {\bibfnamefont {J.~S.}\ \bibnamefont
  {Cotler}}, \bibinfo {author} {\bibfnamefont {D.~K.}\ \bibnamefont {Mark}},
  \bibinfo {author} {\bibfnamefont {H.-Y.}\ \bibnamefont {Huang}}, \bibinfo
  {author} {\bibfnamefont {F.}~\bibnamefont {Hernandez}}, \bibinfo {author}
  {\bibfnamefont {J.}~\bibnamefont {Choi}}, \bibinfo {author} {\bibfnamefont
  {A.~L.}\ \bibnamefont {Shaw}}, \bibinfo {author} {\bibfnamefont
  {M.}~\bibnamefont {Endres}},\ and\ \bibinfo {author} {\bibfnamefont
  {S.}~\bibnamefont {Choi}},\ }\bibfield  {title} {\bibinfo {title} {Emergent
  quantum state designs from individual many-body wave functions},\ }\href@noop
  {} {\bibfield  {journal} {\bibinfo  {journal} {PRX Quantum}\ }\textbf
  {\bibinfo {volume} {4}},\ \bibinfo {pages} {010311} (\bibinfo {year}
  {2023})}\BibitemShut {NoStop}%
\bibitem [{\citenamefont {Ho}\ and\ \citenamefont {Choi}(2022)}]{ho2022exact}%
  \BibitemOpen
  \bibfield  {author} {\bibinfo {author} {\bibfnamefont {W.~W.}\ \bibnamefont
  {Ho}}\ and\ \bibinfo {author} {\bibfnamefont {S.}~\bibnamefont {Choi}},\
  }\bibfield  {title} {\bibinfo {title} {Exact emergent quantum state designs
  from quantum chaotic dynamics},\ }\href@noop {} {\bibfield  {journal}
  {\bibinfo  {journal} {Phys. Rev. Lett.}\ }\textbf {\bibinfo {volume} {128}},\
  \bibinfo {pages} {060601} (\bibinfo {year} {2022})}\BibitemShut {NoStop}%
\bibitem [{\citenamefont {Kaneko}\ \emph {et~al.}(2020)\citenamefont {Kaneko},
  \citenamefont {Iyoda},\ and\ \citenamefont
  {Sagawa}}]{kaneko2020characterizing}%
  \BibitemOpen
  \bibfield  {author} {\bibinfo {author} {\bibfnamefont {K.}~\bibnamefont
  {Kaneko}}, \bibinfo {author} {\bibfnamefont {E.}~\bibnamefont {Iyoda}},\ and\
  \bibinfo {author} {\bibfnamefont {T.}~\bibnamefont {Sagawa}},\ }\bibfield
  {title} {\bibinfo {title} {Characterizing complexity of many-body quantum
  dynamics by higher-order eigenstate thermalization},\ }\href@noop {}
  {\bibfield  {journal} {\bibinfo  {journal} {Phys. Rev. A}\ }\textbf {\bibinfo
  {volume} {101}},\ \bibinfo {pages} {042126} (\bibinfo {year}
  {2020})}\BibitemShut {NoStop}%
\bibitem [{\citenamefont {Dalzell}\ \emph {et~al.}(2022)\citenamefont
  {Dalzell}, \citenamefont {Hunter-Jones},\ and\ \citenamefont
  {Brand{\~a}o}}]{dalzell2022random}%
  \BibitemOpen
  \bibfield  {author} {\bibinfo {author} {\bibfnamefont {A.~M.}\ \bibnamefont
  {Dalzell}}, \bibinfo {author} {\bibfnamefont {N.}~\bibnamefont
  {Hunter-Jones}},\ and\ \bibinfo {author} {\bibfnamefont {F.~G.}\ \bibnamefont
  {Brand{\~a}o}},\ }\bibfield  {title} {\bibinfo {title} {Random quantum
  circuits anticoncentrate in log depth},\ }\href@noop {} {\bibfield  {journal}
  {\bibinfo  {journal} {PRX Quantum}\ }\textbf {\bibinfo {volume} {3}},\
  \bibinfo {pages} {010333} (\bibinfo {year} {2022})}\BibitemShut {NoStop}%
\bibitem [{\citenamefont {Hangleiter}\ \emph {et~al.}(2018)\citenamefont
  {Hangleiter}, \citenamefont {Bermejo-Vega}, \citenamefont {Schwarz},\ and\
  \citenamefont {Eisert}}]{hangleiter2018anticoncentration}%
  \BibitemOpen
  \bibfield  {author} {\bibinfo {author} {\bibfnamefont {D.}~\bibnamefont
  {Hangleiter}}, \bibinfo {author} {\bibfnamefont {J.}~\bibnamefont
  {Bermejo-Vega}}, \bibinfo {author} {\bibfnamefont {M.}~\bibnamefont
  {Schwarz}},\ and\ \bibinfo {author} {\bibfnamefont {J.}~\bibnamefont
  {Eisert}},\ }\bibfield  {title} {\bibinfo {title} {Anticoncentration theorems
  for schemes showing a quantum speedup},\ }\href@noop {} {\bibfield  {journal}
  {\bibinfo  {journal} {Quantum}\ }\textbf {\bibinfo {volume} {2}},\ \bibinfo
  {pages} {65} (\bibinfo {year} {2018})}\BibitemShut {NoStop}%
\bibitem [{\citenamefont {Heinrich}\ \emph {et~al.}(2022)\citenamefont
  {Heinrich}, \citenamefont {Kliesch},\ and\ \citenamefont
  {Roth}}]{heinrich2022randomized}%
  \BibitemOpen
  \bibfield  {author} {\bibinfo {author} {\bibfnamefont {M.}~\bibnamefont
  {Heinrich}}, \bibinfo {author} {\bibfnamefont {M.}~\bibnamefont {Kliesch}},\
  and\ \bibinfo {author} {\bibfnamefont {I.}~\bibnamefont {Roth}},\ }\bibfield
  {title} {\bibinfo {title} {Randomized benchmarking with random quantum
  circuits},\ }\href@noop {} {\bibfield  {journal} {\bibinfo  {journal} {arXiv
  preprint arXiv:2212.06181}\ } (\bibinfo {year} {2022})}\BibitemShut {NoStop}%
\bibitem [{\citenamefont {Aaronson}\ and\ \citenamefont
  {Arkhipov}(2011)}]{aaronson2011computational}%
  \BibitemOpen
  \bibfield  {author} {\bibinfo {author} {\bibfnamefont {S.}~\bibnamefont
  {Aaronson}}\ and\ \bibinfo {author} {\bibfnamefont {A.}~\bibnamefont
  {Arkhipov}},\ }\bibfield  {title} {\bibinfo {title} {The computational
  complexity of linear optics},\ }in\ \href@noop {} {\emph {\bibinfo
  {booktitle} {Proceedings of the forty-third annual ACM symposium on Theory of
  computing}}}\ (\bibinfo {year} {2011})\ pp.\ \bibinfo {pages}
  {333--342}\BibitemShut {NoStop}%
\bibitem [{\citenamefont {Wan}\ \emph {et~al.}(2023)\citenamefont {Wan},
  \citenamefont {Huggins}, \citenamefont {Lee},\ and\ \citenamefont
  {Babbush}}]{wan2023matchgate}%
  \BibitemOpen
  \bibfield  {author} {\bibinfo {author} {\bibfnamefont {K.}~\bibnamefont
  {Wan}}, \bibinfo {author} {\bibfnamefont {W.~J.}\ \bibnamefont {Huggins}},
  \bibinfo {author} {\bibfnamefont {J.}~\bibnamefont {Lee}},\ and\ \bibinfo
  {author} {\bibfnamefont {R.}~\bibnamefont {Babbush}},\ }\bibfield  {title}
  {\bibinfo {title} {Matchgate shadows for fermionic quantum simulation},\
  }\href@noop {} {\bibfield  {journal} {\bibinfo  {journal} {Commun. Math.
  Phys.}\ }\textbf {\bibinfo {volume} {404}},\ \bibinfo {pages} {629} (\bibinfo
  {year} {2023})}\BibitemShut {NoStop}%
\bibitem [{\citenamefont {Mitsuhashi}\ \emph
  {et~al.}(2024{\natexlab{a}})\citenamefont {Mitsuhashi}, \citenamefont
  {Suzuki}, \citenamefont {Soejima},\ and\ \citenamefont
  {Yoshioka}}]{mitsuhashi2024unitary}%
  \BibitemOpen
  \bibfield  {author} {\bibinfo {author} {\bibfnamefont {Y.}~\bibnamefont
  {Mitsuhashi}}, \bibinfo {author} {\bibfnamefont {R.}~\bibnamefont {Suzuki}},
  \bibinfo {author} {\bibfnamefont {T.}~\bibnamefont {Soejima}},\ and\ \bibinfo
  {author} {\bibfnamefont {N.}~\bibnamefont {Yoshioka}},\ }\bibfield  {title}
  {\bibinfo {title} {Unitary designs of symmetric local random circuits},\
  }\href@noop {} {\bibfield  {journal} {\bibinfo  {journal} {arXiv preprint
  arXiv:2408.13472}\ } (\bibinfo {year} {2024}{\natexlab{a}})}\BibitemShut
  {NoStop}%
\bibitem [{\citenamefont {Mitsuhashi}\ \emph
  {et~al.}(2024{\natexlab{b}})\citenamefont {Mitsuhashi}, \citenamefont
  {Suzuki}, \citenamefont {Soejima},\ and\ \citenamefont
  {Yoshioka}}]{mitsuhashi2024characterization}%
  \BibitemOpen
  \bibfield  {author} {\bibinfo {author} {\bibfnamefont {Y.}~\bibnamefont
  {Mitsuhashi}}, \bibinfo {author} {\bibfnamefont {R.}~\bibnamefont {Suzuki}},
  \bibinfo {author} {\bibfnamefont {T.}~\bibnamefont {Soejima}},\ and\ \bibinfo
  {author} {\bibfnamefont {N.}~\bibnamefont {Yoshioka}},\ }\bibfield  {title}
  {\bibinfo {title} {Characterization of randomness in quantum circuits of
  continuous gate sets},\ }\href@noop {} {\bibfield  {journal} {\bibinfo
  {journal} {arXiv preprint arXiv:2408.13475}\ } (\bibinfo {year}
  {2024}{\natexlab{b}})}\BibitemShut {NoStop}%
\bibitem [{\citenamefont {Liu}\ \emph {et~al.}(2024)\citenamefont {Liu},
  \citenamefont {Hulse},\ and\ \citenamefont {Marvian}}]{liu2024unitary}%
  \BibitemOpen
  \bibfield  {author} {\bibinfo {author} {\bibfnamefont {H.}~\bibnamefont
  {Liu}}, \bibinfo {author} {\bibfnamefont {A.}~\bibnamefont {Hulse}},\ and\
  \bibinfo {author} {\bibfnamefont {I.}~\bibnamefont {Marvian}},\ }\bibfield
  {title} {\bibinfo {title} {Unitary designs from random symmetric quantum
  circuits},\ }\href@noop {} {\bibfield  {journal} {\bibinfo  {journal} {arXiv
  preprint arXiv:2408.14463}\ } (\bibinfo {year} {2024})}\BibitemShut {NoStop}%
\bibitem [{\citenamefont {Suzuki}\ \emph {et~al.}(2024)\citenamefont {Suzuki},
  \citenamefont {Katsura}, \citenamefont {Mitsuhashi}, \citenamefont {Soejima},
  \citenamefont {Eisert},\ and\ \citenamefont {Yoshioka}}]{suzuki2024more}%
  \BibitemOpen
  \bibfield  {author} {\bibinfo {author} {\bibfnamefont {R.}~\bibnamefont
  {Suzuki}}, \bibinfo {author} {\bibfnamefont {H.}~\bibnamefont {Katsura}},
  \bibinfo {author} {\bibfnamefont {Y.}~\bibnamefont {Mitsuhashi}}, \bibinfo
  {author} {\bibfnamefont {T.}~\bibnamefont {Soejima}}, \bibinfo {author}
  {\bibfnamefont {J.}~\bibnamefont {Eisert}},\ and\ \bibinfo {author}
  {\bibfnamefont {N.}~\bibnamefont {Yoshioka}},\ }\bibfield  {title} {\bibinfo
  {title} {More global randomness from less random local gates},\ }\href@noop
  {} {\bibfield  {journal} {\bibinfo  {journal} {arXiv preprint
  arXiv:2410.24127}\ } (\bibinfo {year} {2024})}\BibitemShut {NoStop}%
\bibitem [{\citenamefont {Kong}\ \emph {et~al.}(2024)\citenamefont {Kong},
  \citenamefont {Li},\ and\ \citenamefont {Liu}}]{kong2024convergence}%
  \BibitemOpen
  \bibfield  {author} {\bibinfo {author} {\bibfnamefont {L.}~\bibnamefont
  {Kong}}, \bibinfo {author} {\bibfnamefont {Z.}~\bibnamefont {Li}},\ and\
  \bibinfo {author} {\bibfnamefont {Z.-W.}\ \bibnamefont {Liu}},\ }\bibfield
  {title} {\bibinfo {title} {Convergence efficiency of quantum gates and
  circuits},\ }\href@noop {} {\bibfield  {journal} {\bibinfo  {journal} {arXiv
  preprint arXiv:2411.04898}\ } (\bibinfo {year} {2024})}\BibitemShut {NoStop}%
\bibitem [{\citenamefont {Riddell}\ \emph {et~al.}(2025)\citenamefont
  {Riddell}, \citenamefont {Klobas},\ and\ \citenamefont
  {Bertini}}]{riddell2025quantum}%
  \BibitemOpen
  \bibfield  {author} {\bibinfo {author} {\bibfnamefont {J.}~\bibnamefont
  {Riddell}}, \bibinfo {author} {\bibfnamefont {K.}~\bibnamefont {Klobas}},\
  and\ \bibinfo {author} {\bibfnamefont {B.}~\bibnamefont {Bertini}},\
  }\bibfield  {title} {\bibinfo {title} {Quantum state designs from minimally
  random quantum circuits},\ }\href@noop {} {\bibfield  {journal} {\bibinfo
  {journal} {arXiv preprint arXiv:2503.05698}\ } (\bibinfo {year}
  {2025})}\BibitemShut {NoStop}%
\end{thebibliography}%


\begin{thebibliography}{15}%
\makeatletter
\providecommand \@ifxundefined [1]{%
 \@ifx{#1\undefined}
}%
\providecommand \@ifnum [1]{%
 \ifnum #1\expandafter \@firstoftwo
 \else \expandafter \@secondoftwo
 \fi
}%
\providecommand \@ifx [1]{%
 \ifx #1\expandafter \@firstoftwo
 \else \expandafter \@secondoftwo
 \fi
}%
\providecommand \natexlab [1]{#1}%
\providecommand \enquote  [1]{``#1''}%
\providecommand \bibnamefont  [1]{#1}%
\providecommand \bibfnamefont [1]{#1}%
\providecommand \citenamefont [1]{#1}%
\providecommand \href@noop [0]{\@secondoftwo}%
\providecommand \href [0]{\begingroup \@sanitize@url \@href}%
\providecommand \@href[1]{\@@startlink{#1}\@@href}%
\providecommand \@@href[1]{\endgroup#1\@@endlink}%
\providecommand \@sanitize@url [0]{\catcode `\\12\catcode `\$12\catcode
  `\&12\catcode `\#12\catcode `\^12\catcode `\_12\catcode `\%12\relax}%
\providecommand \@@startlink[1]{}%
\providecommand \@@endlink[0]{}%
\providecommand \url  [0]{\begingroup\@sanitize@url \@url }%
\providecommand \@url [1]{\endgroup\@href {#1}{\urlprefix }}%
\providecommand \urlprefix  [0]{URL }%
\providecommand \Eprint [0]{\href }%
\providecommand \doibase [0]{https://doi.org/}%
\providecommand \selectlanguage [0]{\@gobble}%
\providecommand \bibinfo  [0]{\@secondoftwo}%
\providecommand \bibfield  [0]{\@secondoftwo}%
\providecommand \translation [1]{[#1]}%
\providecommand \BibitemOpen [0]{}%
\providecommand \bibitemStop [0]{}%
\providecommand \bibitemNoStop [0]{.\EOS\space}%
\providecommand \EOS [0]{\spacefactor3000\relax}%
\providecommand \BibitemShut  [1]{\csname bibitem#1\endcsname}%
\let\auto@bib@innerbib\@empty
%</preamble>
\bibitem [{\citenamefont {Mele}(2024)}]{mele2024introduction}%
  \BibitemOpen
  \bibfield  {author} {\bibinfo {author} {\bibfnamefont {A.~A.}\ \bibnamefont
  {Mele}},\ }\bibfield  {title} {\bibinfo {title} {Introduction to haar measure
  tools in quantum information: A beginner's tutorial},\ }\href@noop {}
  {\bibfield  {journal} {\bibinfo  {journal} {Quantum}\ }\textbf {\bibinfo
  {volume} {8}},\ \bibinfo {pages} {1340} (\bibinfo {year} {2024})}\BibitemShut
  {NoStop}%
\bibitem [{\citenamefont {Brand{\~a}o}\ \emph {et~al.}(2016)\citenamefont
  {Brand{\~a}o}, \citenamefont {Harrow},\ and\ \citenamefont
  {Horodecki}}]{brandao2016local}%
  \BibitemOpen
  \bibfield  {author} {\bibinfo {author} {\bibfnamefont {F.~G.}\ \bibnamefont
  {Brand{\~a}o}}, \bibinfo {author} {\bibfnamefont {A.~W.}\ \bibnamefont
  {Harrow}},\ and\ \bibinfo {author} {\bibfnamefont {M.}~\bibnamefont
  {Horodecki}},\ }\bibfield  {title} {\bibinfo {title} {Local random quantum
  circuits are approximate polynomial-designs},\ }\href@noop {} {\bibfield
  {journal} {\bibinfo  {journal} {Commun. Math. Phys.}\ }\textbf {\bibinfo
  {volume} {346}},\ \bibinfo {pages} {397} (\bibinfo {year}
  {2016})}\BibitemShut {NoStop}%
\bibitem [{\citenamefont {Haferkamp}(2022)}]{haferkamp2022random}%
  \BibitemOpen
  \bibfield  {author} {\bibinfo {author} {\bibfnamefont {J.}~\bibnamefont
  {Haferkamp}},\ }\bibfield  {title} {\bibinfo {title} {Random quantum circuits
  are approximate unitary {$t$}-designs in depth
  {$O\left(nt^{5+o(1)}\right)$}},\ }\href@noop {} {\bibfield  {journal}
  {\bibinfo  {journal} {{Quantum}}\ }\textbf {\bibinfo {volume} {6}},\ \bibinfo
  {pages} {795} (\bibinfo {year} {2022})}\BibitemShut {NoStop}%
\bibitem [{\citenamefont {Chen}\ \emph {et~al.}(2024)\citenamefont {Chen},
  \citenamefont {Haah}, \citenamefont {Haferkamp}, \citenamefont {Liu},
  \citenamefont {Metger},\ and\ \citenamefont
  {Tan}}]{chen2024incompressibility}%
  \BibitemOpen
  \bibfield  {author} {\bibinfo {author} {\bibfnamefont {C.-F.}\ \bibnamefont
  {Chen}}, \bibinfo {author} {\bibfnamefont {J.}~\bibnamefont {Haah}}, \bibinfo
  {author} {\bibfnamefont {J.}~\bibnamefont {Haferkamp}}, \bibinfo {author}
  {\bibfnamefont {Y.}~\bibnamefont {Liu}}, \bibinfo {author} {\bibfnamefont
  {T.}~\bibnamefont {Metger}},\ and\ \bibinfo {author} {\bibfnamefont
  {X.}~\bibnamefont {Tan}},\ }\bibfield  {title} {\bibinfo {title}
  {Incompressibility and spectral gaps of random circuits},\ }\href@noop {}
  {\bibfield  {journal} {\bibinfo  {journal} {arXiv preprint arXiv:2406.07478}\
  } (\bibinfo {year} {2024})}\BibitemShut {NoStop}%
\bibitem [{\citenamefont {Haferkamp}\ and\ \citenamefont
  {Hunter-Jones}(2021)}]{haferkamp2021improved}%
  \BibitemOpen
  \bibfield  {author} {\bibinfo {author} {\bibfnamefont {J.}~\bibnamefont
  {Haferkamp}}\ and\ \bibinfo {author} {\bibfnamefont {N.}~\bibnamefont
  {Hunter-Jones}},\ }\bibfield  {title} {\bibinfo {title} {Improved spectral
  gaps for random quantum circuits: Large local dimensions and all-to-all
  interactions},\ }\href@noop {} {\bibfield  {journal} {\bibinfo  {journal}
  {Phys. Rev. A}\ }\textbf {\bibinfo {volume} {104}},\ \bibinfo {pages}
  {022417} (\bibinfo {year} {2021})}\BibitemShut {NoStop}%
\bibitem [{\citenamefont {Mittal}\ and\ \citenamefont
  {Hunter-Jones}(2023)}]{mittal2023local}%
  \BibitemOpen
  \bibfield  {author} {\bibinfo {author} {\bibfnamefont {S.}~\bibnamefont
  {Mittal}}\ and\ \bibinfo {author} {\bibfnamefont {N.}~\bibnamefont
  {Hunter-Jones}},\ }\bibfield  {title} {\bibinfo {title} {Local random quantum
  circuits form approximate designs on arbitrary architectures},\ }\href@noop
  {} {\bibfield  {journal} {\bibinfo  {journal} {arXiv preprint
  arXiv:2310.19355}\ } (\bibinfo {year} {2023})}\BibitemShut {NoStop}%
\bibitem [{\citenamefont {Belkin}\ \emph {et~al.}(2024)\citenamefont {Belkin},
  \citenamefont {Allen}, \citenamefont {Ghosh}, \citenamefont {Kang},
  \citenamefont {Lin}, \citenamefont {Sud}, \citenamefont {Chong},
  \citenamefont {Fefferman},\ and\ \citenamefont
  {Clark}}]{belkin2023approximate}%
  \BibitemOpen
  \bibfield  {author} {\bibinfo {author} {\bibfnamefont {D.}~\bibnamefont
  {Belkin}}, \bibinfo {author} {\bibfnamefont {J.}~\bibnamefont {Allen}},
  \bibinfo {author} {\bibfnamefont {S.}~\bibnamefont {Ghosh}}, \bibinfo
  {author} {\bibfnamefont {C.}~\bibnamefont {Kang}}, \bibinfo {author}
  {\bibfnamefont {S.}~\bibnamefont {Lin}}, \bibinfo {author} {\bibfnamefont
  {J.}~\bibnamefont {Sud}}, \bibinfo {author} {\bibfnamefont {F.~T.}\
  \bibnamefont {Chong}}, \bibinfo {author} {\bibfnamefont {B.}~\bibnamefont
  {Fefferman}},\ and\ \bibinfo {author} {\bibfnamefont {B.~K.}\ \bibnamefont
  {Clark}},\ }\bibfield  {title} {\bibinfo {title} {Approximate t-designs in
  generic circuit architectures},\ }\href@noop {} {\bibfield  {journal}
  {\bibinfo  {journal} {PRX Quantum}\ }\textbf {\bibinfo {volume} {5}},\
  \bibinfo {pages} {040344} (\bibinfo {year} {2024})}\BibitemShut {NoStop}%
\bibitem [{\citenamefont {Schuster}\ \emph {et~al.}(2024)\citenamefont
  {Schuster}, \citenamefont {Haferkamp},\ and\ \citenamefont
  {Huang}}]{schuster2024random}%
  \BibitemOpen
  \bibfield  {author} {\bibinfo {author} {\bibfnamefont {T.}~\bibnamefont
  {Schuster}}, \bibinfo {author} {\bibfnamefont {J.}~\bibnamefont
  {Haferkamp}},\ and\ \bibinfo {author} {\bibfnamefont {H.-Y.}\ \bibnamefont
  {Huang}},\ }\bibfield  {title} {\bibinfo {title} {Random unitaries in
  extremely low depth},\ }\href@noop {} {\bibfield  {journal} {\bibinfo
  {journal} {arXiv preprint arXiv:2407.07754}\ } (\bibinfo {year}
  {2024})}\BibitemShut {NoStop}%
\bibitem [{\citenamefont {LaRacuente}\ and\ \citenamefont
  {Leditzky}(2024)}]{laracuente2024approximate}%
  \BibitemOpen
  \bibfield  {author} {\bibinfo {author} {\bibfnamefont {N.}~\bibnamefont
  {LaRacuente}}\ and\ \bibinfo {author} {\bibfnamefont {F.}~\bibnamefont
  {Leditzky}},\ }\bibfield  {title} {\bibinfo {title} {Approximate unitary $ k
  $-designs from shallow, low-communication circuits},\ }\href@noop {}
  {\bibfield  {journal} {\bibinfo  {journal} {arXiv preprint arXiv:2407.07876}\
  } (\bibinfo {year} {2024})}\BibitemShut {NoStop}%
\bibitem [{\citenamefont {Oszmaniec}\ \emph {et~al.}(2021)\citenamefont
  {Oszmaniec}, \citenamefont {Sawicki},\ and\ \citenamefont
  {Horodecki}}]{oszmaniec2021epsilon}%
  \BibitemOpen
  \bibfield  {author} {\bibinfo {author} {\bibfnamefont {M.}~\bibnamefont
  {Oszmaniec}}, \bibinfo {author} {\bibfnamefont {A.}~\bibnamefont {Sawicki}},\
  and\ \bibinfo {author} {\bibfnamefont {M.}~\bibnamefont {Horodecki}},\
  }\bibfield  {title} {\bibinfo {title} {Epsilon-nets, unitary designs, and
  random quantum circuits},\ }\href@noop {} {\bibfield  {journal} {\bibinfo
  {journal} {IEEE Trans. Inf. Theory.}\ }\textbf {\bibinfo {volume} {68}},\
  \bibinfo {pages} {989} (\bibinfo {year} {2021})}\BibitemShut {NoStop}%
\bibitem [{\citenamefont {Oszmaniec}\ \emph {et~al.}(2024)\citenamefont
  {Oszmaniec}, \citenamefont {Kotowski}, \citenamefont {Horodecki},\ and\
  \citenamefont {Hunter-Jones}}]{oszmaniec2022saturation}%
  \BibitemOpen
  \bibfield  {author} {\bibinfo {author} {\bibfnamefont {M.}~\bibnamefont
  {Oszmaniec}}, \bibinfo {author} {\bibfnamefont {M.}~\bibnamefont {Kotowski}},
  \bibinfo {author} {\bibfnamefont {M.}~\bibnamefont {Horodecki}},\ and\
  \bibinfo {author} {\bibfnamefont {N.}~\bibnamefont {Hunter-Jones}},\
  }\bibfield  {title} {\bibinfo {title} {Saturation and recurrence of quantum
  complexity in random local quantum dynamics},\ }\href@noop {} {\bibfield
  {journal} {\bibinfo  {journal} {Phys. Rev. X}\ }\textbf {\bibinfo {volume}
  {14}},\ \bibinfo {pages} {041068} (\bibinfo {year} {2024})}\BibitemShut
  {NoStop}%
\bibitem [{\citenamefont {Lloyd}(1995)}]{lloyd1995almost}%
  \BibitemOpen
  \bibfield  {author} {\bibinfo {author} {\bibfnamefont {S.}~\bibnamefont
  {Lloyd}},\ }\bibfield  {title} {\bibinfo {title} {Almost any quantum logic
  gate is universal},\ }\href@noop {} {\bibfield  {journal} {\bibinfo
  {journal} {Phys. Rev. Lett.}\ }\textbf {\bibinfo {volume} {75}},\ \bibinfo
  {pages} {346} (\bibinfo {year} {1995})}\BibitemShut {NoStop}%
\bibitem [{\citenamefont {Varj{\'u}}(2013)}]{varju2013random}%
  \BibitemOpen
  \bibfield  {author} {\bibinfo {author} {\bibfnamefont {P.~P.}\ \bibnamefont
  {Varj{\'u}}},\ }\bibfield  {title} {\bibinfo {title} {Random walks in compact
  groups},\ }\href@noop {} {\bibfield  {journal} {\bibinfo  {journal}
  {Documenta Mathematica}\ }\textbf {\bibinfo {volume} {18}},\ \bibinfo {pages}
  {1137} (\bibinfo {year} {2013})}\BibitemShut {NoStop}%
\bibitem [{\citenamefont {Harrow}\ and\ \citenamefont
  {Low}(2009)}]{harrow2009random}%
  \BibitemOpen
  \bibfield  {author} {\bibinfo {author} {\bibfnamefont {A.~W.}\ \bibnamefont
  {Harrow}}\ and\ \bibinfo {author} {\bibfnamefont {R.~A.}\ \bibnamefont
  {Low}},\ }\bibfield  {title} {\bibinfo {title} {Random quantum circuits are
  approximate 2-designs},\ }\href@noop {} {\bibfield  {journal} {\bibinfo
  {journal} {Commun. Math. Phys.}\ }\textbf {\bibinfo {volume} {291}},\
  \bibinfo {pages} {257} (\bibinfo {year} {2009})}\BibitemShut {NoStop}%
\bibitem [{\citenamefont {Brown}\ and\ \citenamefont
  {Viola}(2010)}]{brown2010convergence}%
  \BibitemOpen
  \bibfield  {author} {\bibinfo {author} {\bibfnamefont {W.~G.}\ \bibnamefont
  {Brown}}\ and\ \bibinfo {author} {\bibfnamefont {L.}~\bibnamefont {Viola}},\
  }\bibfield  {title} {\bibinfo {title} {Convergence rates for arbitrary
  statistical moments of random quantum circuits},\ }\href@noop {} {\bibfield
  {journal} {\bibinfo  {journal} {Phys. Rev. Lett.}\ }\textbf {\bibinfo
  {volume} {104}},\ \bibinfo {pages} {250501} (\bibinfo {year}
  {2010})}\BibitemShut {NoStop}%
\end{thebibliography}%

\clearpage

\appendix

\section{1D local circuit} \label{app:1D_local}
We here provide proof sketch of Theorem~\ref{thm:single_layer_connected_depth}, in the case of 1D local rancom circuit.
While this circuit is just one example of the diverse circuits covered in Theorem~\ref{thm:single_layer_connected_depth}, the discussion for this example includes the essence of the proof of general cases. 

We denote the unitary ensemble for a single layer of 1D local random circuit as $\nu_{\rm lr}$, and abbreviate the superscript $(t)$ of the moment operator and the spectral gap.
The moment operator of $\nu_{\rm lr}$ is denoted as $M_{\nu_{\rm lr}} = \sum_{i=1}^{N-1} \frac{1}{N-1} M_{i,i+1} \otimes \mathbbm{1}_{\overline{i,i+1}}$, where $M_{i,i+1}$ is the moment operator for local unitary ensemble on the $i$-th and $(i+1)$-th qudits, and $ \mathbbm{1}_{\overline{i,i+1}}$ is the identity operator for the $(N-2)$ qudits except for them. 
Its spectral gap is defined as $\Delta_{\nu_{\rm lr}} \equiv 1- \|M_{\nu_{\rm lr}} - P_{\rm all}\|$, where $P_{\rm all}$, the projector onto the invariant subspace $\mathrm{Inv}(\mathrm{U}(d^N)^{\otimes t, t})$, is identical to the $t$-th moment operator of the global Haar measure on the $N$-qudit system.
The minimum spectral gap of the two-qudit unitary ensembles is defined as $\Delta_{\mathrm{loc},\nu_{\rm lr}} = \min_{1\leq i<N}\Delta_{i,i+1}$, where $\Delta_{i,i+1} \equiv 1 - \| M_{i,i+1} - P_{i,i+1}\|$ is the spectral gap for the local unitary ensemble on the $i$-th and $(i+1)$-th qudits. Here, $P_{i,i+1}$ denote the projector for the invariant subspace $\mathrm{Inv}(\mathrm{U}(d^2)^{\otimes t, t})$, and the residual part  $R_{i,i+1} \equiv M_{i,i+1} -P_{i,i+1}$ acts on the orthogonal subspace of $\mathrm{Inv}(\mathrm{U}(d^2)^{\otimes t, t})$.
The spectral gap for the corresponding Haar random circuit is defined as $\Delta_{\nu_{\rm lr}^{\rm H}} \equiv 1 - \|M_{\nu_{\rm lr}^{\rm H}} - P_{\rm all}\|$, where the moment operator is defined as $M_{\nu_{\rm lr}^{\rm H}} = \sum_{i=1}^{N-1} \frac{1}{N-1} P_{i,i+1} \otimes \mathbbm{1}_{\overline{i,i+1}}$.
In this setup, our goal is to derive the following proposition:
\begin{proposition}\label{prop:local_gap}
Let \(\nu_{\rm lr}\) and \(\nu_{\rm lr}^{\rm H}\) be the ensembles of a single layer of non-Haar and Haar random 1D local circuits, respectively, and let \(\Delta_{\mathrm{loc},\nu_{\rm lr}}\) be the minimum value of the spectral gap for two-qudit unitary ensembles. Then, we have
\begin{equation}
\label{eq:1d_loca_spectral_gap}
    \Delta_{\nu_{\rm lr}} \geq \frac{\Delta_{\rm loc,\nu_{\rm lr}}}{2} \Delta_{\nu_{\rm lr}^{\rm H}}.
\end{equation}
\end{proposition}
\begin{proof}
As the first step of the derivation we show the following inequality, by denoting $g \equiv 1 - \Delta_{\rm loc,\nu_{\rm lr}}$:
\begin{align}
    &M_{\nu_{\rm lr}}M_{\nu_{\rm lr}}^{\dag}
    \leq \sum_{i=1}^{N-1} \frac{1}{N-1} M_{i,i+1}M_{i,i+1}^{\dag} \otimes \mathbbm{1}_{\overline{i,i+1}} \nonumber \\
    &\leq \sum_{i=1}^{N-1} \frac{1}{N-1} \left(g^2 \mathbbm{1}_{i,i+1} + (1-g^2) P_{i,i+1} \right) \otimes \mathbbm{1}_{\overline{i,i+1}}  \nonumber \\
    &= g^2 \mathbbm{1} + (1-g^2) M_{\nu_{\rm lr}^{\rm H}}, \label{eq:1d_local_mom_op_bound}
\end{align}
where the inequality $A\leq B$ for Hermitian operators $A$ and $B$ means that $B-A$ is positive semidefinite.
In the derivation of the first line, we used the Cauchy–Schwarz inequality 
\begin{equation*}
    \sum_{i=1}^{N-1} (M_{i,i+1}\otimes \mathbbm{1}_{\overline{i,i+1}} -M_{\nu_{\rm lr}})(M_{i,i+1}\otimes \mathbbm{1}_{\overline{i,i+1}} - M_{\nu_{\rm lr}})^{\dag} \geq 0 .
\end{equation*}
In the derivation of the second line, we used the relation $M_{i,i+1}M_{i,i+1}^{\dag} = P_{i,i+1} + R_{i,i+1}R_{i,i+1}^\dag$ and the inequality $R_{i,i+1}R_{i,i+1}^\dag \leq g^2 (\mathbbm{1}_{i,i+1} - P_{i,i+1})$, which follows from $\|R_{i,i+1}\| =1-\Delta_{i,i+1} \leq g$. 

Using Eq.~\eqref{eq:1d_local_mom_op_bound}, we get
\begin{align*}
(M_{\nu_{\rm lr}} - P_{\rm all})&(M_{\nu_{\rm lr}} - P_{\rm all})^{\dag}= M_{\nu_{\rm lr}}M_{\nu_{\rm lr}}^{\dag} - P_{\rm all}\\
&\leq g^2 (\mathbbm{1}-P_{\rm all}) + (1-g^2) (M_{\nu_{\rm lr}^{\rm H}}-P_{\rm all}),
\end{align*}
where the left-hand side is positive semidefinite.
Therefore, by taking the operator norm for the both sides of this inequality, we can derive 
\begin{equation}
\label{eq:1d_local_spectral_gap_prf}
(1- \Delta_{\nu_{\rm lr}})^2 \leq g^2 + (1-g^2)(1 -\Delta_{\nu_{\rm lr}^{\rm H}}).
\end{equation}
From Eq.~\eqref{eq:1d_local_spectral_gap_prf}, as well as the inequalities $(1- \Delta_{\nu_{\rm lr}})^2 \geq 1-2 \Delta_{\nu_{\rm lr}}$ and $1-g^2 \geq \Delta_{\rm loc,\nu_{\rm lr}}$, we can derive Eq.~\eqref{eq:1d_loca_spectral_gap}.
\end{proof}

\section{1D brickwork circuit} \label{app:1D_bw}
We here derive the lower bound for the spectral gap of 1D brickwork circuit, Eq.~\eqref{eq:spectral_gap_1Dbw} in the main text. 
This derivation includes the essence of the proof of Theorem~\ref{thm:multilayer_connected_depth}, which covers circuits with generic fixed architectures. 

We consider an $N$-qudit system, assuming that $N$ is even to simplify the notation.
The unitary ensemble of the two-layer block is denoted as $\nu_{\rm bw}= \nu_{\rm e} * \nu_{\rm o}$, where $\nu_{\rm o}$ and $\nu_{\rm e}$ represent the ensemble for odd and even layers, respectively. In this appendix, we omit the superscript $(t)$ of the moment operator and the spectral gap, for the simplicity of notation.
The moment operator for the two-layer block is represented as $M_{\nu_{\rm bw}} = M_{\nu_{\rm e}}M_{\nu_{\rm o}}$, with the moment operators for odd and even layers defined as
\begin{align*}
    M_{\nu_{\rm o}} &\equiv M_{1,2} \otimes M_{3,4} \otimes \cdots \otimes M_{N-1,N}, \\
    M_{\nu_{\rm e}} &\equiv M_{2,3} \otimes M_{4,5} \otimes \cdots \otimes M_{N-2,N-1}.
\end{align*}
Here, $M_{i,i+1}$ is the moment operator for the local unitary ensemble on the $i$-th and $(i+1)$-th qudits. 
Its spectral gap is defined as $\Delta_{\nu_{\rm bw}} \equiv 1 - \| M_{\nu_{\rm bw}} - P_{\rm all}\|$.
The minumum local spectral gap is defined as $\Delta_{\mathrm{loc},\nu_{\rm bw}} = \min_{1\leq i< N}\Delta_{i,i+1}$, where $\Delta_{i,i+1} \equiv 1 - \| M_{i,i+1} - P_{i,i+1}\|$ is the spectral gap for the local unitary ensemble on the $i$-th and $(i+1)$-th qudits. 
The moment operator for the corresponding Haar random circuit is denoted as $M_{\nu_{\rm bw}^{\rm H}}= P_{\rm e} P_{\rm o}$, where $P_{\rm e}$ and $P_{\rm o}$ are the projectors defined as
\begin{align*}
    P_{\rm o} &\equiv P_{1,2} \otimes P_{3,4} \otimes \cdots \otimes P_{N-1,N}, \\
    P_{\rm e} &\equiv P_{2,3} \otimes P_{4,5} \otimes \cdots \otimes P_{N-2,N-1}.
\end{align*}
The spectral gap for the Haar random 1D brickwork circuit is then represented as $\Delta_{\nu_{\rm bw}^{\rm H}} =  1- \left\|P_{\rm e} P_{\rm o} - P_{\rm all} \right\|$.
Since this quantity is represented with a product of projectors, its lower bound can be derived by using the detectability lemma \cite{aharonov2009detectability,anshu2016simple}.

The goal of this appendix is to extend the applicability of the detectability lemma, by lower bounding $\Delta_{\nu_{\rm bw}}$ with $\Delta_{\nu_{\rm bw}^{\rm H}}$ as follows:
\begin{proposition} \label{prop:brickwork_gap}
Let \(\nu_{\rm bw}\) and \(\nu_{\rm bw}^{\rm H}\) be the ensembles of two-layer blocks for the non-Haar and Haar random 1D brickwork circuits, respectively, and let \(\Delta_{\mathrm{loc},\nu_{\rm bw}}\) be the minimum value of the spectral gap for two-qudit unitary ensembles. Then, we have
\begin{equation}
\label{eq:prop_bw_LB}
    \Delta_{\nu_{\rm bw}} \geq \Bigl( \Delta_{\mathrm{loc},\nu_{\rm bw}} \Bigr)^2 \Delta_{\nu_{\rm bw}^{\rm H}}.
\end{equation}
\end{proposition}

As a preparation for the proof of this proposition, we introduce the decomposition of $M_{\nu_{\rm o}}$ and $M_{\nu_{\rm e}}$ into three parts, where each part acts on orthogonal subspace.
First, we define the residual part as $R_{\nu_{\rm o}} \equiv M_{\nu_{\rm o}} - P_{\rm o}$ and $R_{\nu_{\rm e}} \equiv M_{\nu_{\rm e}} - P_{\rm e}$, where the relations $P_{\rm o}R_{\nu_{\rm o}} = R_{\nu_{\rm o}} P_{\rm o} =0$ and $P_{\rm e}R_{\nu_{\rm e}} = R_{\nu_{\rm e}} P_{\rm e} =0$ hold.
For these operators, the following inequalities are satisfied:
\begin{align}
    R_{\nu_{\rm o}}R_{\nu_{\rm o}}^{\dag} &\leq (1-\Delta_{\mathrm{loc},\nu_{\rm bw}})^2(\mathbbm{1}-P_{\rm o}), \label{eq:R_odd_layer_ineq}\\
    R_{\nu_{\rm e}}R_{\nu_{\rm e}}^{\dag} &\leq (1-\Delta_{\mathrm{loc},\nu_{\rm bw}})^2(\mathbbm{1}-P_{\rm e}),\label{eq:R_even_layer_ineq}
\end{align}
To derive these inequalities, we introduce the notations $\Pi_{i, 0} \equiv P_{i-1,i}$ and $\Pi_{i, 1} \equiv \mathbbm{1}_{i-1,i}-P_{i-1,i}$. Since the inequality $M_{i-1,i}M_{i-1,i}^\dag \leq \Pi_{i, 0} + (1-\Delta_{\mathrm{loc},\nu_{\rm bw}})^2 \Pi_{i, 1}$ is satisfied for all $i$, we can derive the following inequality: 
\begin{align}
    &M_{\nu_{\rm o}} M_{\nu_{\rm o}}^\dag = \bigotimes_{j=1}^{N/2} M_{2j-1,2j}M_{2j-1,2j}^\dag \nonumber \\
    &\leq \bigotimes_{j=1}^{N/2} \left[\Pi_{2j, 0} + (1-\Delta_{\mathrm{loc},\nu_{\rm bw}})^{2}\Pi_{2j, 1} \right]  \nonumber \\
    &= \bigotimes_{j=1}^{N/2}\Pi_{2j, 0} + \sum_{(l_1,l_2, ..., l_{N/2})\in \Lambda} \bigotimes_{j=1}^{N/2}(1-\Delta_{\mathrm{loc},\nu_{\rm bw}})^{2l_j} \Pi_{2j, l_j}\nonumber\\
    &\leq \bigotimes_{j=1}^{N/2}\Pi_{2j, 0} + (1-\Delta_{\mathrm{loc},\nu_{\rm bw}})^{2}\sum_{(l_1,l_2, ..., l_{N/2})\in \Lambda} \bigotimes_{j=1}^{N/2}\Pi_{2j, l_j}\nonumber\\
    &= P_{\rm o} + (1 - \Delta_{\mathrm{loc},\nu_{\rm bw}})^2 (\mathbbm{1}- P_{\rm o}), \label{eq:R_odd_ineq_prf}
\end{align}
where we define the index set $\Lambda \equiv \{0,1\}^{\frac{N}{2}} \setminus \{0\}^{\frac{N}{2}}$.
From Eq.~\eqref{eq:R_odd_ineq_prf} and the relation $M_{\nu_{\rm o}} M_{\nu_{\rm o}}^\dag = P_{\rm o} + R_{\nu_{\rm o}}R_{\nu_{\rm o}}^{\dag}$, we can derive Eq.~\eqref{eq:R_odd_layer_ineq}. 
Equation~\eqref{eq:R_even_layer_ineq} can be derived in the same way.

Furthermore, it is known that the intersection of the subspaces where $P_{\rm o}$ and $P_{\rm e}$ acts nontrivially, coincides with the subspace where $P_{\rm all}$ acts nontrivially (see e.g., Lemma~17 of Ref.~\cite{brandao2016local}).
Therefore, the projectors $P_{\rm o}$ and $P_{\rm e}$ can be decomposed into the projector on their common subspace $P_{\rm all}$ and those on the orthogonal complement, which are denoted as $Q_{\rm o} \equiv P_{\rm o} -P_{\rm all}$ and $Q_{\rm e} \equiv  P_{\rm e} -P_{\rm all}$, and satisfy the relations $Q_{\rm o} P_{\rm all} = P_{\rm all}Q_{\rm o} = 0$ and $Q_{\rm e} P_{\rm all} = P_{\rm all}Q_{\rm e} = 0$, respectively. Thus, the moment operators $M_{\nu_{\rm o}}$ and $M_{\nu_{\rm e}}$ admit orthogonal decompositions into three parts as
\begin{equation*}
    M_{\nu_{\rm o}} = P_{\rm all} + Q_{\rm o} +R_{\nu_{\rm o}}, \quad M_{\nu_{\rm e}} = P_{\rm all} + Q_{\rm e} +R_{\nu_{\rm e}},
\end{equation*}
which are utilized in the proof of Proposition~\ref{prop:brickwork_gap}.

\begin{proof}[Proof of Proposition~\ref{prop:brickwork_gap}]
By introducing the notation $\alpha \equiv 1 - (1-\Delta_{\mathrm{loc},\nu_{\rm bw}})^2$, we can derive the following inequality:
\begin{align}
    &(M_{\nu_{\rm e}}M_{\nu_{\rm o}} - P_{\rm all})(M_{\nu_{\rm e}}M_{\nu_{\rm o}} - P_{\rm all})^\dag \nonumber\\
    &=M_{\nu_{\rm e}}M_{\nu_{\rm o}}M_{\nu_{\rm o}}^\dag M_{\nu_{\rm e}}^\dag - P_{\rm all} \nonumber\\
    &\leq M_{\nu_{\rm e}}\left( (1-\alpha) \mathbbm{1} + \alpha P_{\rm o} \right) M_{\nu_{\rm e}}^\dag - \alpha P_{\rm all} \nonumber\\
    &\leq (1-\alpha) \mathbbm{1}  +  \alpha (M_{\nu_{\rm e}}P_{\rm o} M_{\nu_{\rm e}}^\dag - P_{\rm all}) \nonumber\\ 
    &=(1-\alpha) \mathbbm{1}  + \alpha \left[(Q_{\rm e} + R_{\nu_{\rm e}}) Q_{\rm o}\right] \left[(Q_{\rm e} + R_{\nu_{\rm e}}) Q_{\rm o}\right]^\dag. \label{eq:prop_bw_prf1}
\end{align}
In the third line, we used Eq.~\eqref{eq:R_odd_layer_ineq} and the fact that $0 \leq B \leq C$ implies $A B A^\dag  \leq A C A^\dag$ for any linear operator $A,B$ and $C$. In the fifth line, we used the decompositions $M_{\nu_{\rm e}} = P_{\rm all} + Q_{\rm e} +R_{\nu_{\rm e}}$ and $P_{\rm o} =  P_{\rm all} + Q_{\rm o}$, and the relations $P_{\rm all} Q_{\rm o} = 0$ and $(Q_{\rm e} + R_{\nu_{\rm e}})P_{\rm all} =0$.
By taking the operator norm of both sides, we get 
\begin{align}
    &\left(1- \Delta_{\nu_{\rm bw}}\right)^2
    =\|M_{\nu_{\rm e}}M_{\nu_{\rm o}} - P_{\rm all}\|^2 \nonumber\\
    &\leq 1-\alpha + \alpha \|\left[(Q_{\rm e} + R_{\nu_{\rm e}}) Q_{\rm o}\right] \left[(Q_{\rm e} + R_{\nu_{\rm e}}) Q_{\rm o}\right]^\dag\|\nonumber\\
    &=1-\alpha + \alpha \|\left[(Q_{\rm e} + R_{\nu_{\rm e}}) Q_{\rm o}\right]^\dag \left[(Q_{\rm e} + R_{\nu_{\rm e}}) Q_{\rm o}\right]\| \nonumber\\
    &=1-\alpha + \alpha \|Q_{\rm o} (Q_{\rm e} + R_{\nu_{\rm e}}^\dag R_{\nu_{\rm e}}) Q_{\rm o} \|, \label{eq:prop_bw_prf2}
\end{align}
where we used Eq.~\eqref{eq:prop_bw_prf1} in the second line, and $Q_{\rm e}R_{\nu_{\rm e}}=0$ in the final line. 

We can further upper bound the operator norm that appears in the final line of Eq.~\eqref{eq:prop_bw_prf2} as follows:
\begin{align}
&\|Q_{\rm o} (Q_{\rm e} + R_{\nu_{\rm e}}^\dag R_{\nu_{\rm e}}) Q_{\rm o} \| \nonumber \\
&\leq \|Q_{\rm o} [Q_{\rm e} + (1-\alpha) (\mathbbm{1} - P_{\rm all} - Q_{\rm e})] Q_{\rm o} \| \nonumber\\
&=\left\|(1-\alpha) Q_{\rm o} + \alpha Q_{\rm o} Q_{\rm e} Q_{\rm o} \right\| \nonumber\\
&\leq (1-\alpha) +\alpha \left\| Q_{\rm o} Q_{\rm e} Q_{\rm o} \right\| \nonumber\\
&= (1-\alpha) +\alpha \left\|Q_{\rm o} Q_{\rm e}\right\|^2 \nonumber\\ 
&= (1-\alpha) +\alpha \left\| P_{\rm o} P_{\rm e} - P_{\rm all}\right\|^2 \nonumber\\ 
&= (1-\alpha) +\alpha \left(1- \Delta_{\nu_{\rm bw}^{\rm H}}\right)^2, \label{eq:prop_bw_prf3}
\end{align}
where we used Eq.~\eqref{eq:R_even_layer_ineq} in the second line, $Q_{\rm o}P_{\rm all} =0$ in the third line.
From Eqs.~(\ref{eq:prop_bw_prf2}) and~(\ref{eq:prop_bw_prf3}), we can derive the following inequality, by introducing the notations $g \equiv 1-\Delta_{\mathrm{loc},\nu_{\rm bw}}=\sqrt{1- \alpha}$ and $h\equiv 1- \Delta_{\nu_{\rm bw}^{\rm H}}$:
\begin{align}
&\left(1- \Delta_{\nu_{\rm bw}}\right)^2 \leq 1-\alpha^2 \left[1- (1- \Delta_{\nu_{\rm bw}^{\rm H}})^2 \right] \nonumber \\
&= 1- (1-g^2)^2 (1- h^2) \nonumber \\
&= 1 - (1-g)^2 (1-h) (1+g)^2(1+h) \nonumber \\
&\leq  1 - (1-g)^2 (1-h) [2- (1-g)^2(1-h)] \nonumber \\
&= [1- (1-g)^2 (1-h)]^2= \left[1- \left( \Delta_{\mathrm{loc},\nu_{\rm bw}}  \right)^2 \Delta_{\nu_{\rm bw}^{\rm H}} \right]^2,\label{eq:prop_bw_prf4}
\end{align}
where the first line follows from Eqs.~(\ref{eq:prop_bw_prf2}) and~(\ref{eq:prop_bw_prf3}), and the fourth line is derived by utilizing $0\leq g\leq 1$ and $0\leq h\leq 1$.
From Eq.~\eqref{eq:prop_bw_prf4}, we can derive Eq.~\eqref{eq:prop_bw_LB}.
\end{proof}

\end{document}

\title{Supplemental Material for ``Non-Haar random circuits form unitary designs as fast as Haar random circuits''}

% ------------  AUTHORS AND AFFILIATIONS ----------
\author{Toshihiro Yada}
\affiliation{\mbox{Department of Applied Physics, The University of Tokyo, 7-3-1 Hongo, Bunkyo-ku, Tokyo 113-8656, Japan}}

\author{Ryotaro Suzuki}
\affiliation{\mbox{Dahlem Center for Complex Quantum Systems, Freie Universität Berlin, Berlin 14195, Germany}}

\author{Yosuke Mitsuhashi}
\affiliation{\mbox{Department of Basic Science, The University of Tokyo, 3-8-1 Komaba, Meguro-ku, Tokyo 153-8902, Japan}}

\author{Nobuyuki Yoshioka}
\affiliation{\mbox{International Center for Elementary Particle Physics, The University of Tokyo, 7-3-1 Hongo, Bunkyo-ku, Tokyo 113-0033, Japan}}

\setcounter{section}{0}
\setcounter{equation}{0}
\setcounter{figure}{0}
\setcounter{table}{0}
\setcounter{page}{1}
\renewcommand{\thesection}{S\arabic{section}}
\renewcommand{\theequation}{S\arabic{equation}}
\renewcommand{\thefigure}{S\arabic{figure}}
\renewcommand{\thetable}{S\arabic{table}}
\renewcommand{\bibnumfmt}[1]{[S#1]}
\renewcommand{\citenumfont}[1]{S#1}

\maketitle

\tableofcontents

\vspace{4mm}\noindent
We here outline the structure of the Supplemental Material. 
In Sec.~\ref{Ss:preliminary}, we introduce the concept of approximate unitary designs and explain the standard method for upper bounding the unitary design formation depth in random circuits. 
In Sec.~\ref{Ss:main_results}, we present the open questions addressed in this work and briefly summarize our main results. 
In Sec.~\ref{Ss:single_layer_repetition}, we address the single-layer-connected circuit and provide the proof of Theorem~1 in the main text. 
In Sec.~\ref{Ss:multilayer_circuit}, we analyze the multilayer-connected circuit and present the proof of Theorems~2 in the main text.
In Sec.~\ref{Ss:patchwork}, we discuss the unitary design formation depth in patchwork circuits. 
In Sec.~\ref{Ss:local_unitary}, we provide conditions under which the local spectral gap remains nonzero and analyze its scaling behavior.

\section{Preliminary} \label{Ss:preliminary}

\subsection{Approximate unitary design}

In this subsection, we consider the probqbility distribution on the set of unitary operations on $q$-dimensional Hilbert space, which is denoted as $\mathrm{U}(q)$.
A unitary $t$-design is the probability distribution whose $t$-th moment is identical to that of the uniform distribution (i.e., Haar distribution). 
The $t$-th moment of the probability distribution on $\mathrm{U}(q)$ is completely characterized with the $t$-th moment operator defined as
\begin{equation} \label{seq:def_moment_operator}
   M_{\nu}^{(t)} \equiv \Exnu \left[ U^{\otimes t, t} \right]
\end{equation}
where $U^{\otimes t, t}\equiv U^{\otimes t}\otimes U^{*\otimes t}$ and $U^*$ denotes the complex conjugate of unitary $U$ with respect to a fixed basis. 
The moment operator for the Haar ensemble is denoted as $M_{\rm Haar}^{(t)}$, % and the equality $M_{\nu}^{(t)} = M_{\rm Haar}^{(t)}$ is satisfied if and only if the ensemble $\nu$ is an exact unitary $t$-design.
and it is known to be the projection onto the invariant subspace $\mathrm{Inv}(\mathrm{U}(q)^{\otimes t, t})$ of $\mathrm{U}(q)^{\otimes t, t}\equiv\{U^{\otimes t, t} | U\in \mathrm{U}(q)\}$ \cite{mele2024introduction}.
The invariant subspace is explicitly written as $\mathrm{Inv}(\mathrm{U}(q)^{\otimes t, t}) = \mathrm{span} \left\{\left. \ket{\sigma} \right| \sigma \in S_t \right\}$ with $\ket{\sigma} \equiv q^{-t/2} \sum_{1\leq i_j\leq q} \ket{i_1,i_2,\dots, i_t}\ket{i_{\sigma(1)},i_{\sigma(2)},\dots, i_{\sigma(t)}}$ by the Schur-Weyl duality, where $S_t$ is the symmetric group of degree $t$. We represent this projector onto the invariant subspace as $P^{(t)}= M_{\rm Haar}^{(t)}$.
The moment operator for any unitary ensemble $\nu$ is known to be decomposed as $M_{\nu}^{(t)} = P^{(t)} + R^{(t)}_{\nu}$, where the residual part $R^{(t)}_{\nu}$ nontrivially acts only on the subspace orthogonal to $\mathrm{Inv}(\mathrm{U}(q)^{\otimes t, t})$, which means the equality $P^{(t)}R^{(t)}_{\nu} = R^{(t)}_{\nu}P^{(t)} =0$ is satisfied. 

Using the moment operator, we here introduce the definition of an approximate unitary design mainly used in this work.
\begin{dfn}[Approximate unitary design based on the moment operator] \label{def:mom_approx_unitary_design}
For $t\in\mathbb{N}$, $\varepsilon\geq0$, a probability distribution $\nu$ on $\mathrm{U}(q)$ is an $\varepsilon$-approximate unitary $t$-design if and only if 
\begin{equation}
\label{eq:def_mom_approx_unitary_design}
    \|M_{\nu}^{(t)} - M_{\rm Haar}^{(t)} \|_{\infty} \leq \frac{\varepsilon}{q^{2t}},
\end{equation}
where $\|\cdot\|_{\infty}$ is the operator norm.
\end{dfn}
\noindent
To distinguish the operator norm from other norms, we write it as $\|\cdot\|_{\infty}$ in the Supplemental Material, rather than using the simpler notation $\|\cdot\|$ used in the main text.
This definition is widely used in evaluating the rate of unitary design formation in random circuits \cite{brandao2016local,haferkamp2022random,chen2024incompressibility,haferkamp2021improved,mittal2023local,belkin2023approximate}. However, there are also other definitions of approximate unitary designs, depending on how to quantify the difference between $\nu$ and the the Haar measure. The representative examples are the definitions based on the twirling channel.

The $t$-fold twirling channel is defined as
\begin{equation}
    \Phi_{\nu}^{(t)} (\cdot ) \equiv \Exnu \left[ U^{\otimes t}\cdot U^{\dag \otimes t} \right],
\end{equation}
where $U^{\dag}$ is the Hermite conjugate of $U$. We denote the twirling channel for the Haar measure as $\Phi_{\rm Haar}^{(t)}$.
There are two frequently used definitions of approximate unitary designs based on the twirling channel, which are relative error and addtive error approximate designs.
\begin{dfn}[Approximate unitary design based on the twirling channel] \label{def:twirl_approx_unitary_design}
For $t\in\mathbb{N}$, $\varepsilon\geq0$, a probability distribution $\nu$ on $\mathrm{U}(q)$ is a relative error $\varepsilon$-approximate unitary $t$-design if and only if 
\begin{equation}
\label{seq:def_rel_t_des}
    (1-\varepsilon) \Phi_{\rm Haar}^{(t)}  \preccurlyeq  \Phi_{\nu}^{(t)} \preccurlyeq (1+\varepsilon) \Phi_{\rm Haar}^{(t)},
\end{equation}
where $\mathcal{E} \preccurlyeq \mathcal{F}$ means that the channel $\mathcal{F} -\mathcal{E}$ is completely positive.
The ensemble $\nu$ is an additive error $\varepsilon$-approximate unitary $t$-design if and only if
\begin{equation}
\label{seq:def_add_t_des}
    \| \Phi_{\nu}^{(t)} - \Phi_{\rm Haar}^{(t)} \|_{\diamond} \leq \varepsilon ,
\end{equation}
where $\| \mathcal{E} \|_{\diamond} \equiv \sup_{R} \frac{\left\|\mathcal{E} \otimes \mathrm{id}_{R} (X) \right\|_1}{\|X\|_1}$ is the diamond norm for the channel.
\end{dfn}
\noindent
These definitions are known to have operational meanings in terms of the indistinguishability \cite{schuster2024random}, while the detailed discussion is abberviated here.

Now, we show that Definitions~\ref{def:mom_approx_unitary_design} and \ref{def:twirl_approx_unitary_design} are quantitatively related to each other. Specifically, the definition with the moment operator (Definition~\ref{def:mom_approx_unitary_design}) implies the other definitions as follows:
\begin{lemma} \label{lem:mom_to_design}
Let $t\in\mathbb{N}$, $\varepsilon \geq 0$, and a probability distribution on $\mathrm{U}(q)$, denoted as $\nu$,  satisfy Eq.~\eqref{eq:def_mom_approx_unitary_design}. Then, $\nu$ is a relative error $\varepsilon$-approximate unitary $t$-design, and an additive error $q^{-t}\varepsilon$-approximate unitary $t$-design.
\end{lemma}
\noindent
The proof of this lemma is provided e.g., in Lemma~4 of Ref.~\cite{brandao2016local}. This lemma demonstrates that Definition~\ref{def:mom_approx_unitary_design} is the strongest definition which implies the others. %, due to the prefactor $q^{-2t}$ in Eq.~\eqref{eq:def_mom_approx_unitary_design}.

\subsection{Basic technique for evaluating unitary design formation depth}
In this work, we consider the random circuit on $N$-qudit system, which is composed of multiple layers.
Denoting the unitary ensemble of the $i$-th layer as $\nu_i$, the whole random circuit is characterized by $\{\nu_i\}_{i \in \mathbb{N}}$, and the convolution of $L$ layers is represented as $\chi_L = \nu_L * \cdots* \nu_2  * \nu_1$, where $*$ represents the convolution of unitary ensembles. For simplicity, we often use the notation $\chi_L = *_{1\leq i \leq L} \nu_i$. The ensemble $\chi_L$ is known to approach the global Haar ensemble as the circuit depth $L$ increases. The approaching rate is characterized by the spectral gap, which is defined as $\Delta_{\nu_i}^{(t)} \equiv 1- \|M_{\nu_i}^{(t)} - M_{\rm Haar}^{(t)}\|_{\infty}$.
Specifically, the following lemma can be derived:
\begin{lemma}\label{lem:basic}
The moment operator of the $L$-fold convolution of unitary ensemble $\chi_L \equiv *_{1\leq i \leq L} \nu_i$ satisfies 
\begin{equation}
    \|M_{\chi_L}^{(t)} -M_{\rm Haar}^{(t)} \|_{\infty} \leq \exp\left(-\sum_{i=1}^L\Delta_{\nu_i}^{(t)} \right).
\end{equation}
\end{lemma}

\begin{proof}
Here, we abbreviate superscript $(t)$ for simplicity.
Since the moment operator for $L$ layers can be represented as $M_{*_{1\leq i \leq L}\nu_i} = M_{\nu_L}M_{\nu_{L-1}}\cdots M_{\nu_2}M_{\nu_1}$, we can derive the following inequality:
\begin{equation}
\label{seq:conv_mom_op}
\begin{split}
    \| M_{*_{1\leq i \leq L}\nu_i} - M_{\rm Haar} \|_{\infty} 
    &= \left \|(M_{\nu_L} - M_{\rm Haar})(M_{\nu_{L-1}} - M_{\rm Haar}) \cdots (M_{\nu_1} - M_{\rm Haar})\right\|_{\infty} \\
    &\leq \prod_{i=1}^L \| M_{\nu_i}  - M_{\rm Haar} \|_{\infty} \\
    &=  \prod_{i=1}^L (1- \Delta_{\nu_i}) \\
    &\leq \exp\left(-\sum_{i=1}^L \Delta_{\nu_i} \right),
\end{split}
\end{equation}
where we use the equality $M_{\nu_i} M_{\rm Haar} = M_{\rm Haar} M_{\nu_i} = M_{\rm Haar}$ in the first line, the sub-multiplicity of the spectral norm in the second line, and the inequality $1-x\leq e^{-x}$ in the fourth line.
\end{proof}

This lemma shows that the moment operator monotonically approaches to that of the Haar measure as the circuit depth increases, and its approaching rate can be characterized by the spectral gap of each layer.
Since the approximate unitary design is defined based on the moment operator (Definition~\ref{def:mom_approx_unitary_design}), this lemma implies that there exists a circuit depth $L$, such that $\chi_L$ is an $\varepsilon$-approximate unitary $t$-design while $\chi_{L-1}$ is not. We call such a circuit depth $L$ as the formation depth of an $\varepsilon$-approximate $t$-design.
When the random circuit is the repetition of the same layer represented as $\nu$, Lemma~\ref{lem:basic} shows that an $\varepsilon$-approximate $t$-design is formed when the circuit depth satisfies 
\begin{equation}
    L \geq \left(\Delta_{\nu}^{(t)}\right)^{-1} (2 Nt \log d - \log \varepsilon),
\end{equation}
where $\Delta_{\nu}^{(t)}$ is the spectral gap of each layer $\nu$ and $d$ is the local dimension.

\section{Main results} \label{Ss:main_results}

\subsection{Setup and open questions} \label{Sss:setup_and_openQs}

Since this work comprehensively covers various circuit structures, we classify them into the following three types based on their essential configuration. By addressing each type collectively, we streamline the discussion and avoid overlaps. 
\begin{itemize}
    \item \textit{Single-layer-connected circuit} is a circuit where the two-qudit gates available in a single layer form a connected graph covering all sites. Here, we define the graph whose vertices represent the sites of qudits and edges indicate locations where two-qudit gates act with nonzero probability. The graph defined in such a way is said to be connected if every pair of vertices are connected by a path of edges.
    Importantly, this class does not require that all sites be deterministically connected in each layer; rather, the key condition is that the graph of possible two-qudit interactions in a single layer is connected across all sites. A representative example of this class is the 1D local random circuit.
    \item \textit{Multilayer-connected circuit} is the circuit that takes multiple layers for all sites to be connected. 
    A representative example of this class is the 1D brickwork circuit, where all sites are connected in two layers rather than a single layer. Furthermore, the circuits with fixed architectures, where the order and positions of the unitary gate applications are predetermined, are always multilayer-connected, since the number of available gates in a single layer is at most $\frac{N}{2}$, and they cannot form a connected graph. 
    \item \textit{Patchwork circuit} is the circuit with special geometry proposed in Refs.~\cite{schuster2024random,laracuente2024approximate}. This class of circuit is constructed by gluing together small patches of normal random circuits, either single-layer-connected circuits or multilayer-connected circuits. Haar random circuits with such structures are known to form approximate unitary designs in the circuit depth of $O(\log N)$ in $N$-qubit systems.
\end{itemize}

\begin{table}[t]
    \centering
    \begin{tabular}{C{3.6cm} |C{4.45cm} C{2.55cm} C{4cm} |C{3cm}}
    \hline
    & \multicolumn{3}{c|}{Previous works} & \red{Our work} \\ \cline{2-5}
         Circuit structure & \fontsize{8}{8} \selectfont \strut Haar \cite{brandao2016local,haferkamp2022random,chen2024incompressibility,haferkamp2021improved,mittal2023local,belkin2023approximate,schuster2024random,laracuente2024approximate} & \fontsize{8}{8} \selectfont \strut Inverse-closed \cite{brandao2016local}& \fontsize{8}{8} \selectfont \strut General non-Haar \cite{oszmaniec2021epsilon,oszmaniec2022saturation} & \fontsize{8}{8} \selectfont \strut \red{General non-Haar} \\
         \hline\hline
    1D local  & $L_{\rm lr}^{\rm H} = \widetilde{O}\left(N (Nt - \log \varepsilon)\right)$ & $\left(\Delta_{\rm loc}^{(t)}\right)^{-1} L_{\rm lr}^{\rm H}$  & $ (2N-2) \left(\Delta_{\rm loc}^{(t)}\right)^{-1} L_{\rm lr}^{\rm H}$ & \red{$2 \left(\Delta_{\rm loc}^{(t)}\right)^{-1} L_{\rm lr}^{\rm H}$} \\
         \hline
    1D parallel & $L_{\rm pr}^{\rm H} = \widetilde{O}\left(Nt  - \log \varepsilon\right)$ & $\left(\Delta_{\rm loc}^{(t)}\right)^{-1} L_{\rm pr}^{\rm H}$  & $4 \left(\Delta_{\rm loc}^{(t)}\right)^{-1}  L_{\rm pr}^{\rm H}$  & \red{$2\left(\Delta_{\rm loc}^{(t)}\right)^{-1} L_{\rm pr}^{\rm H}$} \\
         \hline
    All-to-all & $L_{\rm all}^{\rm H} = \widetilde{O}\left(N (Nt  - \log \varepsilon)\right)$  & $\left(\Delta_{\rm loc}^{(t)}\right)^{-1} L_{\rm all}^{\rm H}$  & $(N^2-N) \left(\Delta_{\rm loc}^{(t)}\right)^{-1}  L_{\rm all}^{\rm H} $ & \red{$2 \left(\Delta_{\rm loc}^{(t)}\right)^{-1} L_{\rm all}^{\rm H}$} \\
         \hline
    \makecell{\strut \fontsize{9.5}{10} General connectivity \\ \fontsize{7}{8}\selectfont \strut (edge number $|E|=\Omega(N)$) } &  $L_{G}^{\rm H} = \widetilde{O}\left(|E| (Nt  - \log \varepsilon)\right)$  & $\left(\Delta_{\rm loc}^{(t)}\right)^{-1} L_{G}^{\rm H}$ &  $|E| \left(\Delta_{\rm loc}^{(t)}\right)^{-1}  L_{G}^{\rm H}$ & \red{ $2\left(\Delta_{\rm loc}^{(t)}\right)^{-1} L_{G}^{\rm H}$}\\
         \hline \hline
    1D brickwork  & $L_{\rm bw}^{\rm H} = \widetilde{O}\left(Nt  - \log \varepsilon\right)$ & - & - & \red{$ \left(\Delta_{\rm loc}^{(t)}\right)^{-2} L_{\rm bw}^{\rm H}$} \\
    \hline
    \makecell{\strut \fontsize{9.5}{10} Fixed architecture \\ \fontsize{7}{8}\selectfont \strut (complete, $l$-layer-connected)} & $L_{A}^{\rm H} = \widetilde{O}\left(Nt  - \log \varepsilon\right)$ & - & - & \red{$ \left(\Delta_{\rm loc}^{(t)}\right)^{-l} L_{A}^{\rm H}$} \\
    \hline
    \makecell{\strut \fontsize{9.5}{10} Fixed architecture\\ \fontsize{7}{8}\selectfont \strut (arbitrary, $l$-layer-connected)} & \makecell{\strut \fontsize{8.5}{8.5} $L_{A^\prime}^{\rm H} = O \left(C_{N,t} (Nt -\log \varepsilon)\right)$  \\ \fontsize{6}{6}\selectfont \strut $C_{N,t} \equiv \exp \left(O(\log\log N \cdot \log \log t)\right)$}  & - & - & \red{$ \left(\Delta_{\rm loc}^{(t)}\right)^{-l} L_{A^\prime}^{\rm H}$} \\
    \hline\hline
    \makecell{\strut \fontsize{9.5}{10} Patchwork* \\ \fontsize{7}{8}\selectfont \strut (each patch is 1D brickwork)} & $L_{\rm pw}^{\rm H} = O(\log \frac{N}{\varepsilon} \cdot t\ \polylog t)$   & - & -& \red{$\left(\Delta_{\rm loc}^{(t)}\right)^{-2} L_{\rm pw}^{\rm H}$} \\
    \hline
    \end{tabular}
    \caption{Best known upper bounds for the circuit depths required to form $\varepsilon$-approximate unitary $t$-design in $N$-qubit system. Each row corresponds to the different circuit structures, and each column represents the choice of the local randomizer. While the upper bounds for Haar random circuits have been obtained for various circuit structures in previous works \cite{brandao2016local,haferkamp2022random,chen2024incompressibility,haferkamp2021improved,mittal2023local,belkin2023approximate,schuster2024random,laracuente2024approximate}, the bounds for non-Haar random circuits have remained largely unexplored or non-optimal. 
    In this work, we comprehensively investigate various circuit architectures and reveal that the formation depths in non-Haar random circuits can be upper bounded by the constant factor multiples of those in the corresponding Haar random circuits, as shown in the column of our work, highlighted with the red text. These results provide the solution to all the open problems raised in Sec.~\ref{Sss:setup_and_openQs}. 
    The prefactor is defined with the local spectral gap $\Delta_{\rm loc}^{(t)}$, which is an $N$-independent constant that quantifies the strength of the local randomizer. When the local randomizer contain a universal gate set, it can be further shown to scales as $\Delta_{\rm loc}^{(t)} =\Omega ((\log t)^{-2})$ with respect to $t$ (see Sec.~\ref{Ss:local_unitary} for the detailed discussion). 
    An asterisk in the patchwork circuit row indicates that the definition of approximate unitary design used there differs from those in the other rows: in that row, a relative error \(\varepsilon\)-approximate \(t\)-design is used, whereas the other rows employs Definition~\ref{def:mom_approx_unitary_design}, which is based on the moment operator.
    The order notation $\widetilde{O}(\cdot)$ indicates that a $\polylog t$ factor is ignored. 
    We note that the depth scalings shown in the column of the Haar random circuits are only valid for restricted values of $t$: $t \leq 2^{\frac{2}{5}N}$ in all architectures except for the 1D patchwork circuit \cite{chen2024incompressibility}, and $t \leq \poly (N)$ in the 1D patchwork circuit \cite{schuster2024random}. 
    A detailed explanation of each circuit structure is provided in the following sections, and thus is abbreviated here: the single-layer-connected circuit (i.e., 1D local, 1D parallel, all-to-all, and general connectivity circuits) is discussed in Sec.~\ref{Ss:single_layer_repetition}, the multilayer-connected circuit (i.e., 1D brickwork circuit and the circuits with fixed architectures) is addressed in Sec.~\ref{Ss:multilayer_circuit}, and the patchwork circuit is discussed in Sec.~\ref{Ss:patchwork}.
    }
    \label{tab:summary_supp}
\end{table}

Regarding single-layer-connected circuits, there are a few previous works with non-Haar local ensembles. 
In Ref.~\cite{brandao2016local}, it was shown that 1D local and parallel circuits with \textit{inverse-closed} two-qudit unitary ensembles form unitary $t$-designs in the same circuit depth order as the corresponding Haar random circuits. Here, a two-qudit unitary ensemble is said to be inverse-closed when it contains a certain unitary $U$ and its inverse $U^\dag$ as a pair with equal probabilities. Furthermore, Refs.~\cite{oszmaniec2021epsilon,oszmaniec2022saturation} investigated the situation where the local gates are chosen from general non-Haar ensembles, which are not necessarily inverse-closed. However, the depth upper bounds obtained in Refs.~\cite{oszmaniec2021epsilon,oszmaniec2022saturation} are $\poly N$-factor multiples of those for the corresponding Haar random circuits, and therefore are not optimal. 
Thus, it remains an open question whether unitary designs can be formed within the same circuit depth order as the corresponding Haar random circuits.

For multilayer-connected circuits, the rate of unitary design formation with non-Haar local ensembles has not been investigated in previous works. The main technical challenge in this class is that we need to bound the spectral gap for a multiple-layer block of the circuit. While it can be done with the detectability lemma when the local gates are chosen Haar randomly, this lemma cannot be directly applied to non-Haar random circuits. 

Regarding the patchwork circuits, the unitary design formation depth with non-Haar local randomizers has not been clarified, due to the lack of the investigations on the other classes of circuits (i.e., single-layer-connected and multilayer-connected circuits). This is because the formation depth in the patchwork circuits highly depends on the randomness generation rate of each small patch, as discussed in Refs.~\cite{schuster2024random,laracuente2024approximate}. 
It remains a significant task to elucidate whether an approximate unitary design can be formed in $O(\log N)$-depth circuit even with non-Haar local randomizers.

\subsection{Our main results}

In this work, we reveal that circuit depths required to form unitary designs in non-Haar random circuits are upper bounded by  those of the corresponding Haar random circuits up to constant prefactors. %, .
As shown in Table~\ref{tab:summary_supp}, the prefactors are defined as $\poly((\Delta_{\rm loc}^{(t)})^{-1})$, where $\Delta_{\rm loc}^{(t)}$ is the spectral gap of the two-qudit unitary ensemble, and do not depend of the system size $N$. 
When the two-qudit ensembles contain universal gate sets, it can be further shown that $\Delta_{\rm loc}^{(t)} = \Omega ((\log t)^{-2})$, which means that the prefactor $\poly((\Delta_{\rm loc}^{(t)})^{-1})$ scales as $\polylog t$ (see Sec.~\ref{Ss:local_unitary} for a detailed discussion).
While Table~\ref{tab:summary_supp} only shows the required circuit depths for $N$-qubit systems, our results are also valid for general $N$-qudit systems with local dimension $d > 2$. These results solve all the open questions raised in the previous subsection, by investigating all three classes of circuit structures, and providing the optimal upper bounds for the unitary design formation depths up to constant factors.

In deriving these results, we develop new techniques to bound the formation depth in non-Haar random circuits.
Our main technical contributions are twofold.
First, we significantly improve the spectral gap bound for single-layer-connected circuits by utilizing the Cauchy–Schwarz inequality (Lemma~\ref{lem:cs_like}).
This improvement leads to an optimal bound for the $t$-design formation depth up to a constant factor, eliminating the additional $\poly N$ factor present in previous works \cite{oszmaniec2021epsilon,oszmaniec2022saturation}.
Second, we introduce a method to lower bound the spectral gap of the multilayer block of non-Haar random circuit. For that purpose, we extend the applicability of the detectability lemma and its variant from Haar random circuits to general non-Haar random circuits (Proposition~2 in Appendix B and Lemma~\ref{lem:fixed_decomp_into_two_nonHaar}).
These techniques would serve as fundamental tools for future investigations into general non-Haar random circuits.

These results have significant implications for both practical applications and fundamental physics.
On the application side, we quantify the rate of global randomness generation in real-world experiments, where local gates are typically chosen from non-Haar random ensembles, such as discrete gate sets.
Our findings improve the flexibility of random circuit construction and ensure the robustness of the randomness generation rate against local gate deformations, for example due to coherent errors.
Specific tasks that would benefit from our results include randomized benchmarking, random circuit sampling, and randomized measurements.
On the fundamental physics side, our results highlight the universal emergence of the chaotic phenomena in random circuits, regardless of the choice of the local randomness.
Since the notion of unitary designs is closely tied to various chaotic phenomena in quantum many-body systems, such as complexity growth, information scrambling, decoupling of initial correlation, and quantum thermalization, our results provide strong evidence that these phenomena occur universally at the same circuit depth order.

\section{Single-layer-connected circuit} \label{Ss:single_layer_repetition}
\begin{figure*}[]
    \centering
    \includegraphics[width=0.95\textwidth]{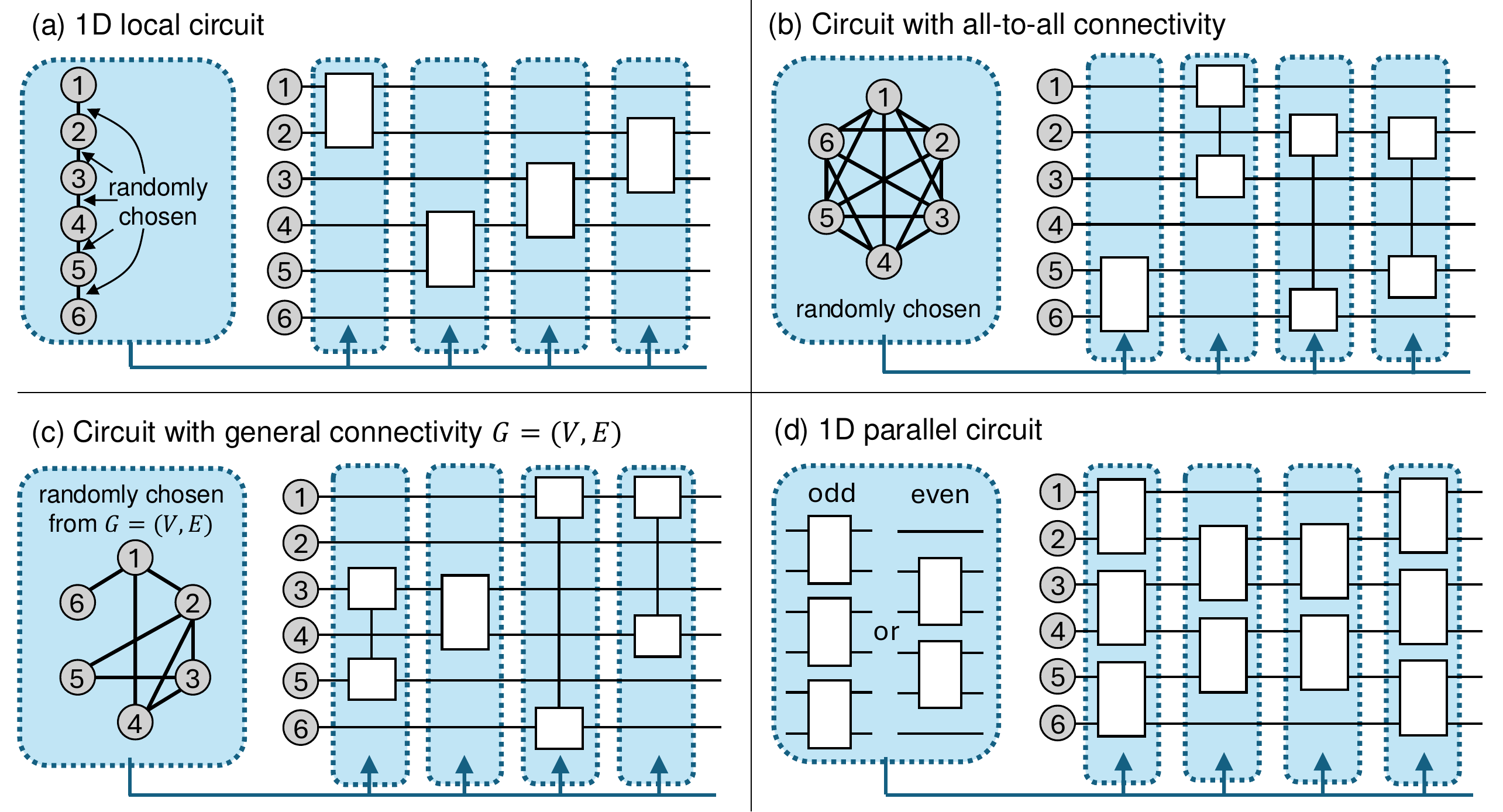}
    \caption{The schematics for (a) 1D local circuit, (b) all-to-all circuit, (c) circuit with general connectivity, and (d) 1D parallel circuit.
    All of these circuits are classified to single-layer-connected circuits, where the graph representing the interaction for a single layer covers all sites.
    }
    \label{fig:rep}
\end{figure*}

\subsection{Definition and examples of single-layer-connected circuit} \label{Sss:sing_def_ex}

The single-layer-connected circuit is a random circuit in which the two-qudit gates available in a single layer form a connected graph covering all sites.  
The unitary ensemble for a single layer is determined by a set of protocols $\{\eta\}$ and their corresponding application probabilities $\{q_\eta\}$.  
Each protocol $\eta$ specifies a set of two-qudit pairs where unitary gates are applied, along with the local unitary ensembles associated with each pair.  
The set of two-qudit pairs for protocol $\eta$ is denoted as $c(\eta) \equiv \{\{i_1, j_1\}, \{i_2, j_2\}, \dots, \{i_{n_\eta}, j_{n_\eta}\}\}$, where $n_\eta$ is the number of two-qudit unitary gates applied in protocol $\eta$.  
The moment operator for such a single-layer ensemble $\nu$ is given by  
\begin{equation}
\label{seq:sin_rep_circuit}
\begin{split}
    M_{\nu}^{(t)} &= \sum_{\eta} q_\eta M_\eta^{(t)}, \\
    M_\eta^{(t)} &\equiv \bigotimes_{\{i,j\} \in c(\eta)} M^{(t)}_{i, j} \otimes \mathbbm{1}_{\overline{c(\eta)}},
\end{split}
\end{equation}
where $M^{(t)}_{i,j}$ represents the moment operator for each local two-qudit ensemble, and $\mathbbm{1}_{\overline{c(\eta)}}$ is the identity operator on the $N-2n_\eta$ qudits that are not included in $c(\eta)$.  

We further introduce a local Haar-random ensemble corresponding to $\nu$.  
First, we define a protocol $\eta^{\rm H}$ in which all two-qudit local ensembles in $\eta$ are replaced by local Haar ensembles.  
Here, the superscript $\rm H$ represents local Haar randomness.  
Then, by replacing every protocol $\eta$ contained in the ensemble $\nu$ with its local Haar-random counterpart $\eta^{\rm H}$, we can define the local Haar-random ensemble $\nu^{\rm H}$ corresponding to $\nu$. 
The moment operator for $\nu^{\rm H}$ is given by  
\begin{equation}
\label{seq:sin_rep_circuit_lHaar}
\begin{split}
    M_{\nu^{\rm H}}^{(t)} &= \sum_{\eta^{\rm H}} q_{\eta^{\rm H}} M_{\eta^{\rm H}}^{(t)}, \\
    M_{\eta^{\rm H}}^{(t)} &\equiv \bigotimes_{\{i,j\} \in c(\eta)} P^{(t)}_{i, j} \otimes \mathbbm{1}_{\overline{c(\eta)}},
\end{split}
\end{equation}
where $q_{\eta^{\rm H}} = q_{\eta}$ for corresponding protocols $\eta$ and $\eta^{\rm H}$, and $P^{(t)}_{i,j}$ represents the projector onto the invariant subspace $\mathrm{Inv}(\mathrm{U}(d^2)^{\otimes t, t})$ of the $i$-th and $j$-th qudits with local dimension $d$.

Many of random circuits studied in the previous works are classified to the single-layer-connected circuit.
We here raise the representative examples illustrated in Figs.~\ref{fig:rep}(a-d).
\begin{itemize}
    \item 1D local circuit: In each layer, choose a single pair of neighboring sites uniformly at random and apply two-qudit gate (i.e., $q_\eta = \frac{1}{N-1}$, $n_\eta=1$ for all protocols $\eta$ in Eqs.~(\ref{seq:sin_rep_circuit}) and (\ref{seq:sin_rep_circuit_lHaar})). The moment operators for Haar and non-Haar random 1D local circuits are represented as follows:
    \begin{align}
        M_{\nu_{\rm lr}^{\rm H}}^{(t)} &\equiv \sum_{i=1}^{N-1} \frac{1}{N-1} P_{i,i+1}^{(t)} \otimes \mathbbm{1}_{\overline{i,i+1}}, \\ 
        M_{\nu_{\rm lr}}^{(t)} &\equiv \sum_{i=1}^{N-1} \frac{1}{N-1} M_{i,i+1}^{(t)} \otimes \mathbbm{1}_{\overline{i,i+1}}, 
    \end{align}
    The unitary design formation depth of Haar random 1D local circuit was studied e.g., in Refs.~\cite{brandao2016local,haferkamp2022random,chen2024incompressibility}.
    \item Circuit with all-to-all connectivity: In each layer, choose a single pair of qudits uniformly at random and apply two-qudit gate (i.e., $q_\eta = \frac{2}{N(N-1)}$, $n_\eta=1$ for all protocols $\eta$). The moment operators for Haar and non-Haar random all-to-all circuits are as follows:
    \begin{align}
        M_{\nu_{\rm all}^{\rm H}}^{(t)} &\equiv \sum_{1\leq i < j \leq N} \frac{2}{N(N-1)} P_{i,j}^{(t)} \otimes \mathbbm{1}_{\overline{i,j}}, \\ 
        M_{\nu_{\rm all}}^{(t)} &\equiv \sum_{1\leq i < j \leq N} \frac{2}{N(N-1)} M_{i,j}^{(t)} \otimes \mathbbm{1}_{\overline{i,j}}.   
    \end{align}
    The unitary design formation in Haar random all-to-all circuit was studied e.g., in Refs.~\cite{haferkamp2021improved,chen2024incompressibility}.
    \item Circuit with general connectivity: The connectivity of circuit is characterized with graph $G(V,E)$, where $V$ is the set of vertices corresponding to each site (i.e., $|V|=N$) and $E$ is the set of the edges connecting a pair of vertices. In each layer, choose a single edge from $E$ uniformly at random, and apply two-qudit gate (i.e., $q_\eta = \frac{1}{|E|}$, $n_\eta =1$ for all protocols $\eta$). The moment operators for Haar and non-Haar random circuit on graph $G(V,E)$ are as follows:
    \begin{align}
        M_{\nu_{G}^{\rm H}}^{(t)} &\equiv \sum_{\{i,j\}\in E } \frac{1}{|E|} P_{i,j}^{(t)} \otimes \mathbbm{1}_{\overline{i,j}}, \\ 
        M_{\nu_{G}}^{(t)} &\equiv \sum_{\{i,j\}\in E} \frac{1}{|E|} M_{i,j}^{(t)} \otimes \mathbbm{1}_{\overline{i,j}}.   
    \end{align}
    While the depth required for \(t\)-design formation clearly depends on the specific structure of the connectivity graph, Ref.~\cite{mittal2023local} classified them into several types and derived the depth upper bounds for every classes. The upper bound shown in Table~\ref{tab:summary_supp} corresponds to the case where the graph has a spanning tree of constant depth and height. The results of Ref.~\cite{chen2024incompressibility} are also incorporated to improve the scaling with respect to \(t\). For a broader discussion including other connectivity graph types, see Ref.~\cite{mittal2023local}.
    \item 1D parallel circuit: In each layer, apply unitary gates $U_{1,2}\otimes U_{3,4}\otimes \cdots \otimes U_{N-1,N}$ or $U_{2,3}\otimes U_{4,5}\otimes \cdots \otimes U_{N-2,N-1}$ randomly (i.e., $q_\eta = \frac{1}{2}$ for these two protocols, and $n_\eta=\frac{N}{2}$ or $\frac{N}{2}-1$). We here assume the number of qudits $N$ is even, while the the same discussion applies to the case of odd number of qudits.
    The moment operators for Haar and non-Haar 1D parallel circuits are as follows:
    \begin{align}
        M_{\nu_{\rm pr}^{\rm H}}^{(t)} &\equiv \frac{1}{2} P_{1,2}^{(t)} \otimes P_{3,4}^{(t)} \otimes \dots \otimes P_{N-1,N}^{(t)} + \frac{1}{2} P_{2,3}^{(t)} \otimes P_{4,5}^{(t)} \otimes \dots \otimes P_{N-2,N-1}^{(t)}, \\ 
        M_{\nu_{\rm pr}}^{(t)} &\equiv \frac{1}{2} M_{1,2}^{(t)} \otimes M_{3,4}^{(t)} \otimes \dots \otimes M_{N-1,N}^{(t)} + \frac{1}{2} M_{2,3}^{(t)} \otimes M_{4,5}^{(t)} \otimes \dots \otimes M_{N-2,N-1}^{(t)}.
    \end{align}
    The unitary design formation in Haar random 1D parallel circuit was studied e.g., in Ref.~\cite{brandao2016local}.
\end{itemize}

\subsection{Unitary design formation in non-Haar random circuit}
We here provide the upper bound for the unitary design formation depth in non-Haar random single-layer-connected circuits.
The important notion in bounding the formation depth is the spectral gap of the local randomizers.
The local spectral gap for the protocol $\eta$ is defined as 
\begin{equation}
\label{seq:loc_spec_gap_eta}
 \Delta_{\rm loc, \eta}^{(t)} \equiv 1 - \max_{\{i,j\} \in c(\eta)} \left\|M^{(t)}_{i,j}- P^{(t)}_{i,j} \right\|_{\infty},
\end{equation}
where the maximization is taken over all the two-qudit pairs. By using $\Delta_{\rm loc, \eta}^{(t)}$, the local spectral gap for a single layer unitary ensemble $\nu$ is defined as
\begin{equation}
    \label{seq:loc_spec_gap_nu}
    \Delta_{\rm loc, \nu}^{(t)} \equiv \min_{\eta} \Delta_{\rm loc, \eta}^{(t)},
\end{equation}
where the minimization is taken over all the protocols $\{\eta\}$ constituting the ensemble $\nu$.
With this quantity, we can bound the spectral gap of $\nu$ as follows:
\begin{proposition} \label{prop:repetition_circuit} 
Let $\nu$ be a unitary ensemble for local non-Haar single-layer-connected circuit, and $\nu^{\rm H}$ be the ensemble of corresponding local Haar random circuit.
Then, the gap of moment operator of non-Haar random circuit can be lower bounded as 
\begin{equation}
\label{seq:slrep_thm}
    \Delta_{\nu}^{(t)} \geq \frac{\Delta_{\rm loc,\nu}^{(t)}}{2} \Delta_{\nu^{\rm H}}^{(t)},
\end{equation}
where $\Delta_{\rm loc,\nu}^{(t)}$ is defined as Eq.~\eqref{seq:loc_spec_gap_nu}.
\end{proposition}

In the proof of this theorem, we use two key lemmas, Lemmas~\ref{lem:ltog_gap} and \ref{lem:cs_like}.
Lemma~\ref{lem:ltog_gap} relates the moment operators $M_\eta^{(t)}$ and $M_{\eta^{\rm H}}^{(t)}$ as follows:
\begin{lemma} \label{lem:ltog_gap}
Consider a protocol $\eta$ where random two-qudit gates are applied to the set of $n_\eta$ two-qudit pairs $c(\eta) \equiv \{\{i_1, j_1\}, \{i_2, j_2\}, \dots, \{i_{n_\eta}, j_{n_\eta}\}\}$ without any overlaps. 
Then, for the moment operators $M_\eta^{(t)}$ and $M_{\eta^{\rm H}}^{(t)}$, defined as Eqs.~\eqref{seq:sin_rep_circuit} and \eqref{seq:sin_rep_circuit_lHaar}, respectively, the following inequality is satisfied:
\begin{equation}
    M_{\eta}^{(t)} M_{\eta}^{(t)\dag} \leq \left(1 - \Delta_{\rm loc,\eta}^{(t)}\right)^2 \mathbbm{1} +  \left[1 - \left(1 - \Delta_{\rm loc,\eta}^{(t)}\right)^2\right] M_{\eta^{\rm H}}^{(t)},
\end{equation}
where the local spectral gap $\Delta_{\rm loc,\eta}^{(t)}$ is defined as Eq.~\eqref{seq:loc_spec_gap_eta}.
\end{lemma}

\begin{proof}
Here, we omit superscript $(t)$ for simplicity. 
We define  $R_{i, j}\equiv M_{i, j}-P_{i, j}$. 
Then, $R_{i, j}$ satisfies $P_{i,j}R_{i,j} = R_{i,j}P_{i,j} = 0$ and $\|R_{i,j}\|_{\infty} \leq 1 - \Delta_{\rm loc,\eta}$, which implies 
\begin{align}
\label{seq:sing_protocol_prf1}
    P_{i_k, j_k}+R_{i_k, j_k}R_{i_k, j_k}^\dag
    \leq P_{i_k, j_k}+(1-\Delta_{\rm loc,\eta})^2(\mathbbm{1}_{i_k, j_k}-P_{i_k, j_k})
    =\sum_{l\in\{0, 1\}} (1-\Delta_{\rm loc,\eta})^{2l}\Pi_{k, l}, 
\end{align}
where $\Pi_{k, 0}\equiv P_{i_k, j_k}$ and $\Pi_{k, 1}\equiv \mathbbm{1}_{i_k, j_k}-P_{i_k, j_k}$. 
By using this property, we get
\begin{align}
    M_\eta M_\eta^\dag &= \left[\bigotimes_{k=1}^{n_\eta} \left(P_{i_k, j_k} + R_{i_k, j_k}R_{i_k, j_k}^\dag \right)\right] \otimes \mathbbm{1}_{\overline{c(\eta)}} \nonumber\\
    &\leq \left[ \sum_{(l_1, ..., l_{n_\eta})\in\{0, 1\}^{n_\eta}} \bigotimes_{k=1}^{n_\eta} (1-\Delta_{\rm loc,\eta})^{2l_k}\Pi_{k, l_k}\right] \otimes \mathbbm{1}_{\overline{c(\eta)}} \nonumber\\
    &=\left[\bigotimes_{k=1}^{n_\eta} \Pi_{k, 0}
    +\sum_{(l_1, ..., l_{n_\eta})\in\{0, 1\}^{n_\eta}\setminus\{0\}^{n_\eta}} \bigotimes_{k=1}^{n_\eta} (1-\Delta_{\rm loc,\eta})^{2l_k}\Pi_{k, l_k}\right] \otimes \mathbbm{1}_{\overline{c(\eta)}} \nonumber\\
    &\leq\left[\bigotimes_{k=1}^{n_\eta} \Pi_{k, 0}
    +(1-\Delta_{\rm loc,\eta})^2\sum_{(l_1, ..., l_{n_\eta})\in\{0, 1\}^{n_\eta}\setminus\{0\}^{n_\eta}} \bigotimes_{k=1}^{n_\eta} \Pi_{k, l_k}\right] \otimes \mathbbm{1}_{\overline{c(\eta)}} \nonumber\\
    &=\left\{\bigotimes_{k=1}^{n_\eta} \Pi_{k, 0}
    +(1-\Delta_{\rm loc,\eta})^2\left[\bigotimes_{k=1}^{n_\eta} (\Pi_{k, 0}+\Pi_{k, 1})-\bigotimes_{k=1}^{n_\eta} \Pi_{k, 0}\right]\right\} \otimes \mathbbm{1}_{\overline{c(\eta)}} \nonumber\\
    &=\left\{(1-\Delta_{\rm loc,\eta})^2 \mathbbm{1}_{i_1, j_1, ..., i_{n_\eta}, j_{n_\eta}}+\left[1-(1-\Delta_{\rm loc,\eta})^2\right]\bigotimes_{k=1}^{n_\eta} P_{i_k, j_k}\right\}\otimes\mathbbm{1}_{\overline{c(\eta)}} \nonumber\\
    &= (1 - \Delta_{\rm loc,\eta})^2 \mathbbm{1} + (1- (1 - \Delta_{\rm loc,\eta})^2 ) M_{\eta^{\rm H}},
\end{align}
where the second line follows from Eq.~\eqref{seq:sing_protocol_prf1}.
\end{proof}

\begin{lemma} \label{lem:cs_like}
Let $n\in\mathbb{N}$, $\{M_x\}_{x=1}^n$ be a set of general (not necessarily Hermitian) matrices, and $\{q_x\}_{x=1}^n$ be a probability distribution. 
Then,
\begin{equation}
\label{eq:cs_like}
    \left(\sum_{x=1}^n q_x M_x \right)\left(\sum_{x^\prime =1}^n q_{x^\prime} M_{x^\prime} \right)^\dag \leq \left(\sum_{x=1}^n q_x M_x M_x^\dag \right).
\end{equation}
\end{lemma}

\begin{proof}
By introducing the representation $M\equiv \sum_{x=1}^n q_x M_x$, Eq.~(\ref{eq:cs_like}) is derived as follows, by utilizing Cauchy-Schwarz inequality:
\begin{equation}
    \left(\sum_{x=1}^n q_x M_x M_x^\dag\right) -MM^\dag
    =\sum_{x=1}^n q_x (M_x-M)(M_x-M)^\dag
    \geq 0.
\end{equation}
\end{proof}

By utilizing Lemmas~\ref{lem:ltog_gap} and \ref{lem:cs_like}, we can derive Proposition~\ref{prop:repetition_circuit}. 
In the proof, we denote the projection to the invariant subspace $\mathrm{Inv}(\mathrm{U}(d^N)^{\otimes t, t})$ as $P^{(t)}_{\rm all}$, which is identical to the $t$-th moment operator for the global Haar ensemble on $N$-qudit system, i.e., $P^{(t)}_{\rm all} = M^{(t)}_{\rm Haar}$.
\begin{proof}[Proof of Proposition~\ref{prop:repetition_circuit}]
In this proof, we abbreviate the superscript $(t)$ for simplicity.
Since the moment operator for non-Haar random circuit is not necessarily Hermitian, we evaluate its Gram matrix instead of the moment operator itself as follows:
\begin{equation}
\label{seq:slrep_them_prf1}
\begin{split}
(M_{\nu} - M_{\rm Haar})(M_{\nu} - M_{\rm Haar})^\dag &= M_{\nu} M_{\nu}^{\dag} - M_{\rm Haar} \\
&= \left(\sum_{\eta} q_\eta M_\eta \right) \left(\sum_{\eta^\prime} q_{\eta^\prime} M_{\eta^\prime} \right)^\dag - P_{\rm all} \\
    &\leq \sum_{\eta} q_\eta M_\eta M_\eta^\dag -P_{\rm all} \\
    &\leq \sum_\eta q_\eta \left\{(1 - \Delta_{\rm loc,\eta})^2  \mathbbm{1} + \left[1 -\left(1 - \Delta_{\rm loc,\eta}\right)^2 \right] M_{\eta^{\rm H}} \right\} -P_{\rm all}\\
    &\leq (1 - \Delta_{\rm loc,\nu})^2 (\mathbbm{1}-P_{\rm all}) + \left[1-(1 - \Delta_{\rm loc,\nu})^2\right]  (M_{\nu^{\rm H}} - P_{\rm all}),
\end{split}
\end{equation}
where the third line follows from Lemma~\ref{lem:cs_like}, the fourth line follows from Lemma~\ref{lem:ltog_gap}, and the final line follows from the inequality $\Delta_{\rm loc,\eta} \geq \Delta_{\rm loc,\nu} \geq 0$, which can be seen from the definition Eq.~\eqref{seq:loc_spec_gap_nu}.

By taking the operator norm of the both sides of Eq.~(\ref{seq:slrep_them_prf1}), we can derive the following inequality:
\begin{equation}
\begin{split}
    (1-\Delta_\nu)^2 &\leq  (1 - \Delta_{\rm loc,\nu})^2  + \left[1-(1 - \Delta_{\rm loc,\nu})^2\right] (1-\Delta_{\nu^{\rm H}}) \\
    &= 1 - \left[1-(1- \Delta_{\rm loc,\nu})^2\right] \Delta_{\nu^{\rm H}} \\
    &\leq 1 -\Delta_{\rm loc,\nu}\Delta_{\nu^{\rm H}}.
\end{split}
\end{equation}
Since the left-hand side is lower bounded as $(1-\Delta_\nu)^2 \geq 1 -2 \Delta_\nu$, we can derive Eq.~(\ref{seq:slrep_thm}).
\end{proof}

By utilizing Proposition~\ref{prop:repetition_circuit} and Lemma~\ref{lem:basic}, we can upper bound the unitary design formation depth in a non-Haar random circuit whose $i$-th layer is denoted as $\nu_i$.
We here consider the situation where the Haar random counterpart of each layer $\nu_i^{\rm H}$ is the same for all $i$, and is simply denoted as $\nu^{\rm H}$. In this case, the probability distribution $q_\eta$ and the set of two-qudit pairs $c(\eta)$ associated with the protocol $\eta$ are the same for all layers. However, the two-qudit unitary ensembles can be different accross layers.

To upper bound the formation depth in such a non-Haar random circuit, the key quantities are the averaged local spectral gap $\Delta_{\mathrm{loc}}^{(t)}$ and the circuit depth at which the corresponding Haar random circuit forms an approximate unitary design, $L^{\rm H}$.
The averaged local spectral gap is defined as 
\begin{equation}
\label{seq:ave_loc_spec_single}
    \Delta_{\mathrm{loc}}^{(t)} \equiv \inf_{K\in \mathbb{N}} \frac{1}{K} \sum_{i=1}^K \Delta_{\mathrm{loc}, \nu_i}^{(t)},
\end{equation}
where $\Delta_{\mathrm{loc}, \nu_i}^{(t)}$ is the local spectral gap of $i$-th layer defined as Eq.~\eqref{seq:loc_spec_gap_nu}.
When the circuit is repetitive and the ensembles of all layers are the same, i.e., $\nu_i=\nu$ for all $i$, Eq.~\eqref{seq:ave_loc_spec_single} is simply reduced to $\Delta_{\rm loc, \nu}^{(t)}$. This is the situation discussed in the main text.

Another key quantity $L^{\rm H}$ is defined as 
\begin{equation}
\label{seq:depth_sin_Haar}
    L^{\rm H} \equiv \left( \Delta_{\nu^{\rm H}} \right)^{-1} \cdot (2Nt \log d - \log \varepsilon),
\end{equation}
where $\Delta_{\nu^{\rm H}}$ is the spectral gap for the corresponding Haar random circuit $\nu^{\rm H}$.
From Lemma~\ref{lem:basic}, we get that $L^{\rm H}$ is the formation depth of $\varepsilon$-approximate unitary $t$-design.
Given these quantities, we can now state the main theorem of this section.
\begin{theorem}[Theorem 1 in the main text] \label{thm:depth_sin_rep}
Consider a single-layer-connected circuit where each layer is represented as $\nu_i$. 
This circuit achieves an $\varepsilon$-approximate unitary $t$-design at a circuit depth of
\begin{equation}
\label{seq:depth_sing}
    L \geq 2\left(\Delta_{\mathrm{loc}}^{(t)}\right)^{-1} \cdot L^{\rm H},
\end{equation}
where the averaged local spectral gap $\Delta_{\mathrm{loc}}^{(t)}$ is defined as Eq.~(\ref{seq:ave_loc_spec_single}), and $L^{\rm H}$ is defined as Eq.~\eqref{seq:depth_sin_Haar}.
\end{theorem}
\begin{proof}
For $L$ defined as Eq.~\eqref{seq:depth_sing}, we can obtain
\begin{align}
    \sum_{i=1}^{L} \Delta_{\nu_i}^{(t)} &\geq  \frac{1}{2} \Delta_{\nu^{\rm H}}^{(t)}\sum_{i=1}^{L}\Delta_{\rm loc,\nu_i}^{(t)}\nonumber \\
    &\geq  \frac{1}{2} \Delta_{\nu^{\rm H}}^{(t)}\cdot L \ \Delta_{\rm loc}^{(t)} \nonumber\\
    &\geq 2 Nt \log d - \log \varepsilon,
\end{align}
where the first line follows from Proposition~\ref{prop:repetition_circuit}, and the second line follows from the definition of the averaged local spectral gap. With this inequality and Lemma~\ref{lem:basic}, we can complete the proof.
\end{proof}

\begin{rem}[Comparison with previous works \cite{oszmaniec2021epsilon,oszmaniec2022saturation}]
Proposition~\ref{prop:repetition_circuit} provides a tighter bound on the spectral gap $\Delta_{\nu}^{(t)}$ by a factor of $\poly N$ compared to previous works~\cite{oszmaniec2021epsilon,oszmaniec2022saturation}. 
The primary reason for this improvement lies in our use of Lemma~\ref{lem:cs_like} in the third line of Eq.~(\ref{seq:slrep_them_prf1}), where we bound the Gram matrix of the moment operator.
This improvement is reflected in the better upper bound of the unitary design formation depth in our work as shown in Table~\ref{tab:summary_supp}.
\end{rem}

\section{Multilayer-connected circuit} \label{Ss:multilayer_circuit}

\begin{figure*}[]
    \centering
    \includegraphics[width=1.0\textwidth]{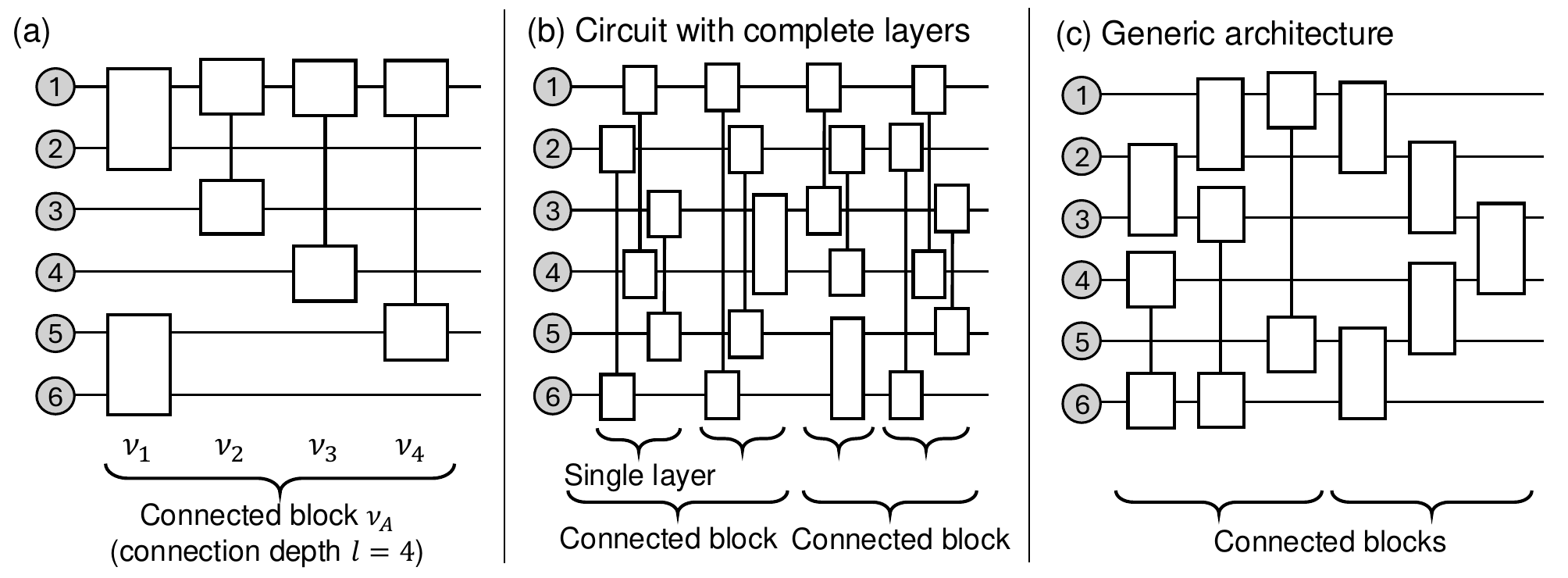}
    \caption{Schematics for circuits with generic fixed architectures. (a) An example of a single connected block of circuit. Here, 6-qudit system is connected in circuit depth $l=4$. A unitary ensemble of $i$-th layer is denoted as $\nu_i$. (b) An example of circuit with complete layers. In every layers, $3$ two-qudit unitary gates are applied in parallel. (c) An example of the circuit with incomplete layers. In this circuit, each layer involves less than $3$ unitary gates. }
    \label{fig:gen_arch}
\end{figure*}

In this section, we discuss the unitary design formation depths in multilayer-connected circuits with generic fixed architectures. Here, a circuit is said to have a fixed architecure when the order and position of unitary gate applications are predetermined. For example, it encompasses both one- and higher-dimensional brickwork circuits.
In Sec.~\ref{Ssss:gen_fixed_Haar}, we introduce the notation for describing the circuit with a generic fixed architecture, and briefly review the method to upper bound the $t$-design formation depth in such a circuit composed of Haar random two-qudit unitary gates. This method was introduced in Ref.~\cite{belkin2023approximate}.  
In Sec.~\ref{Ssss:gen_fixed_nonHaar}, we derive an upper bound on the unitary design formation depth in non-Haar random circuits with a generic fixed architecture.

\subsection{Description of circuits with generic fixed architectures}\label{Ssss:gen_fixed_Haar}
A convenient way to describe the circuit structure is to decompose it into connected blocks.
The connected block is a block of circuit where the gates involved form a connected graph over all sites, and the number of layers involved in each  block is called connection depth.
For example, the circuit shown in Fig.~\ref{fig:gen_arch}(a) consists of a single connected block with connection depth $l=4$, while the circuits of Figs.~\ref{fig:gen_arch}(b) and (c) are composed of two connected blocks.
In this work, each connected block is treated as a fundamental unit of the circuit, and evaluating its spectral gap serves as the primary focus of the analysis.

We here introduce the notation to describe the unitary ensemble formed by each connected blocks.
When the connected block is of depth $l$, its ensemble $\nu_A$ is described as $\nu_A = *_{1\leq j \leq l} \nu_j$ with the unitary ensemble corresponding to $j$-th layer $\nu_j$, as illustrated in Fig.~\ref{fig:gen_arch}(a).
The sets of sites connected by the unitary gates in $\nu_j$ is called the cluster of $\nu_j$ and is denoted as $c(\nu_j)$. For example, in Fig.~\ref{fig:gen_arch}(a), the clusters $c(\nu_1)$ and $c(\nu_2)$ are $c(\nu_1)= \{\{1,2\},\{5,6\}\}$ and $c(\nu_2)= \{\{1,3\}\}$, respectively. 
Corresponding to the cluster $c(\nu_j)$, we introduce the operator $P_{c(\nu_j)}^{(t)}$, which represents the projector onto the invariant subspace of the cluster $c(\nu_j)$. For example, in the circuit shown in Fig.~\ref{fig:gen_arch}(a), $P_{c(\nu_1)}^{(t)} = P_{1,2}^{(t)}\otimes P_{5,6}^{(t)}\otimes \mathbbm{1}_{3,4}$ and $P_{c(\nu_2)}^{(t)} = P_{1,3}^{(t)}\otimes\mathbbm{1}_{2,4,5,6}$, where $P_{1,2}^{(t)}$, $P_{5,6}^{(t)}$, and $P_{1,3}^{(t)}$ are the projectors onto the subspace $\mathrm{Inv}(\mathrm{U}(d^2)^{\otimes t, t})$ of the corresponding qudit pair.
As can be seen from here, for a single-layer ensemble $\nu_j$, the operator $P_{c(\nu_j)}^{(t)}$ coincides with the moment operator for corresponding Haar random circuit $\nu_j^{\rm H}$, i.e., $P_{c(\nu_j)}^{(t)} = M_{\nu_j^{\rm H}}^{(t)}$.

These notations of the cluster and the corresponding projector are also used for the multiple-layer convolution. For example, regarding the convolution of the first two layers in Fig.~\ref{fig:gen_arch}(a), the cluster is represented as $c(\nu_2*\nu_1)= \{\{1,2,3\},\{5,6\}\}$ and the corresponding projector is defined as $P_{c(\nu_2*\nu_1)}^{(t)} = P_{1,2,3}^{(t)}\otimes P_{5,6}^{(t)}\otimes \mathbbm{1}_{4}$. 
About the whole connected block $\nu_{A} \equiv *_{1\leq i \leq 4} \nu_i$, the cluster is represented as $c(\nu_A) =\{\{1,2,3,4,5,6\}\}$, and the projector becomes $P_{c(\nu_A)}^{(t)} = P_{\rm all}^{(t)}$.
Here, it should be noted that the operator $P_{c(\nu_j*\nu_i)}^{(t)}$ is generally different from the moment operator of the corresponding Haar random circuit $\nu_j^{\rm H}*\nu_i^{\rm H}$, represented as $M_{\nu_j^{\rm H}*\nu_i^{\rm H}}^{(t)} = M_{\nu_j^{\rm H}}^{(t)}M_{\nu_i^{\rm H}}^{(t)} =P_{c(\nu_j)}^{(t)}P_{c(\nu_i)}^{(t)}$. 

In Ref.~\cite{belkin2023approximate}, the method to evaluate the spectral gaps of such connected blocks were developed. Specifically, they focus on the situation where each two-qudit gate is drawn from the Haar measure, and derive the lower bound for the spectral gap, defined as
\begin{equation}
\label{seq:spectral_gap_fixed_Haar}
    \Delta_{\nu^{\rm H}_A}^{(t)} \equiv 1 - \|M_{\nu^{\rm H}_A}^{(t)} - P_{\rm all}^{(t)} \|_{\infty} = 1 - \| P_{c(\nu_l)}^{(t)} \cdots P_{c(\nu_2)}^{(t)} P_{c(\nu_1)}^{(t)} - P_{\rm all}^{(t)} \|_{\infty}.
\end{equation}
They divide the circuit architectures into two cases and derive different lower bounds for each: those for circuits with complete architectures, where $\lfloor\frac{N}{2}\rfloor$ two-qudit gates are applied in every layer as shown in Fig.~\ref{fig:gen_arch}(b), and those for circuits with incomplete architectures, where some layers consist of fewer than $\lfloor\frac{N}{2}\rfloor$ gates, as shown in Fig.~\ref{fig:gen_arch}(c).
\begin{proposition} \label{prop:Haar_fixed_arch}
Let \(\nu^{\rm H}_A\) be a unitary ensemble for an \(l\)-layer connected block of a Haar random circuit with a fixed architecture, denoted as \(\nu^{\rm H}_A \equiv *_{1\leq j\leq l} \nu_j^{\rm H}\). Then, its spectral gap defined in Eq.~\eqref{seq:spectral_gap_fixed_Haar} is lower bounded as
\begin{equation}
    \Delta_{\nu^{\rm H}_A}^{(t)} \geq \mathrm{LB}_{A}^{(t)}.
\end{equation}
Here, when the architecture $A$ is complete, $\mathrm{LB}_{A}^{(t)}$ is defined as 
\begin{equation}
    \mathrm{LB}_{A}^{(t)} \equiv f_{\rm c}(N,t,d), \quad f_{\rm c}(N,t,d) \equiv \min_{1\leq m\leq N} \Delta_{\nu_{\rm bw}^{\rm H}}^{(t)}(m,d), \label{eq:spec_LB_comp_Haar_random}
\end{equation}
where $\Delta_{\nu_{\rm bw}^{\rm H}}^{(t)}(m,d)$ is the spectral gap for the two-layer block of Haar random 1D brickwork circuit on $m$-qudit system with local dimension $d$. When the architecture $A$ is incomplete, $\mathrm{LB}_{A}^{(t)}$ is defined as 
\begin{equation}
    \mathrm{LB}_{A}^{(t)} \equiv f_{\rm i}(N,t,d), \quad f_{\rm i}(N,t,d) \equiv \left( \min_{1\leq m\leq N} \Delta_{\nu_{\rm bw}^{\rm H}}^{(t)}(m,d) \right)^h, \label{eq:spec_LB_incomp_Haar_random}
\end{equation}
where $h$ is defined as \(h \equiv 8 \lceil \log_2 \lfloor \log_2 (N+1) \rfloor \rceil + 1\).
\end{proposition}

In the case where the local dimension is \(d=2\), the spectral gap for a 1D brickwork circuit is lower bounded as \(\min_{m\leq N} \Delta_{\nu_{\rm bw}^{\rm H}}^{(t)}(m,d) \geq (\polylog t)^{-1}\), as shown in Ref.~\cite{chen2024incompressibility}. Therefore, Proposition~\ref{prop:Haar_fixed_arch} implies that for a circuit on an \(N\)-qubit system with complete layers, we have \(\Delta_{\nu^{\rm H}_A}^{(t)} \geq (\polylog t)^{-1}\). On the other hand, in the case of incomplete layers, we have \(\Delta_{\nu^{\rm H}_A}^{(t)} \geq \exp\Bigl(-\mathrm{Const}\cdot\log \log N \cdot \log \log t\Bigr)\). 
While the known bounds for the incomplete architectures are much looser than those for the complete architectures due to a subtle technical issue, it is conjectured in Ref.~\cite{belkin2023approximate} that a tighter lower bound, independent of the system size \(N\), exists even for incomplete architectures.

For later convenience, we here briefly overview the derivation of the Proposition~\ref{prop:Haar_fixed_arch}.
The derivation essentially consists of Lemmas~\ref{lem:fixed_decomp_into_two} and \ref{lem:fixed_clustermerge}, which are introduced below.
\begin{lemma}[Bound for the spectral gap with $\gamma^{(t)}$. Theorem 7 of Ref.~\cite{belkin2023approximate} using our terminology] \label{lem:fixed_decomp_into_two}
Let $\nu^{\rm H}_{A}$ be the unitary ensemble of an $l$-layer connected block of a Haar random circuit with a fixed architecture, defined as $\nu^{\rm H}_{A} \equiv *_{1\leq j\leq l} \nu_j^{\rm H}$, where $\nu_j^{\rm H}$ is the ensemble of the $j$-th layer. Then, its spectral gap $\Delta^{(t)}_{\nu^{\rm H}_{A}}$ satisfies
\begin{equation}
\label{seq:lem_fixed_thm7}
    \left(1- \Delta_{\nu^{\rm H}_{A}}^{(t)} \right)^2  \leq 1 - \left(\prod_{j=2}^{l} \gamma^{(t)}(\nu_j^{\rm H},\ast_{1\leq k < j} \nu_k^{\rm H}) \right),
\end{equation}
where $\gamma^{(t)} (\nu,\mu)$ is defined as 
\begin{equation}
\label{seq:gamma_def}
    \gamma^{(t)} (\nu,\mu) \equiv 1- \left\|P^{(t)}_{c(\mu)}P^{(t)}_{c(\nu)}-P^{(t)}_{c(\nu * \mu)} \right\|_{\infty}^2.
\end{equation}
\end{lemma}
The term $\gamma^{(t)}(\nu,\mu)$ is entirely determined by the configurations of the clusters $c(\mu)$ and $c(\nu)$ and satisfies the inequality $0\leq \gamma^{(t)} \leq 1$. 
While the spectral gap $\Delta_{\nu^{\rm H}_{A}}^{(t)}$ itself is defined using the $l$-fold product of projectors, $M_{\nu^{\rm H}_{A}} \equiv  P_{c(\nu_l)}\cdots P_{c(\nu_2)}P_{c(\nu_1)}$, and is difficult to evaluate directly, Lemma~\ref{lem:fixed_decomp_into_two} decomposes it into the terms $\{\gamma^{(t)}(\nu_j^{\rm H},\ast_{1\leq k < j} \nu_k^{\rm H})\}_{j=2}^{l}$, which is defined solely by a two-fold product of projectors and is thus much easier to handle.
The derivation of Lemma~\ref{lem:fixed_decomp_into_two} is contained in the proof of its generalized version, Lemma~\ref{lem:fixed_decomp_into_two_nonHaar}, which is discussed later.

In Ref.~\cite{belkin2023approximate}, the upper bounds of $\gamma^{(t)}$ were obtained by developing a proof technique called the cluster-merging method.
By utilizing this technique, we can derive the following lemma:
\begin{lemma}[Upper bound of $\gamma^{(t)}$. Theorem 8 and Appendix D of Ref.~\cite{belkin2023approximate} using our terminology] \label{lem:fixed_clustermerge}
Suppose $\{\nu_i\}_i$ are unitary ensembles for a single layer of a random circuit with a fixed architecture on an $N$-qudit system with local dimension $d$.
Then, if $\nu_j$ consists of $\lfloor\frac{N}{2}\rfloor$ two-qudit unitary gates, we have
\begin{equation}
    \gamma^{(t)}(\nu_j,\ast_{1\leq k < j} \nu_k)  \leq 1- \left(1- f_{\rm c}(N,t,d) \right)^2,
\end{equation}
where $f_{\rm c}$ is defined as Eq.~\eqref{eq:spec_LB_comp_Haar_random}.
If $\nu_j$ consists of fewer than $\lfloor\frac{N}{2}\rfloor$ two-qudit gates, we have
\begin{equation}
    \gamma^{(t)}(\nu_j,\ast_{1\leq k < j} \nu_k)  \leq 1- \left(1- f_{\rm i}(N,t,d) \right)^2,
\end{equation}
where $f_{\rm i}$ is defined as Eq.~\eqref{eq:spec_LB_incomp_Haar_random}.
\end{lemma}
\noindent
By combining Lemmas~\ref{lem:fixed_decomp_into_two} and \ref{lem:fixed_clustermerge} and further performing some simple algebraic calculations, we can derive Proposition~\ref{prop:Haar_fixed_arch}.

\subsection{Unitary design formation in non-Haar random circuit} \label{Ssss:gen_fixed_nonHaar}

In this subsection, we derive the upper bound for the unitary design formation depth in non-Haar random circuit with generic fixed architectures.
The significant quantity here is the local spectral gap defined as 
\begin{equation}
\label{eq:Del_A_loc}
    \Delta_{\mathrm{loc},\nu_{A}}^{(t)} \equiv \min_{
   \substack{ \{j,k\}\in c(\nu_i),\\ 1\leq i \leq l} } \Delta^{(t)}_{j,k},
\end{equation}
where $\Delta^{(t)}_{j,k}$ is the spectral gap for the two-qudit unitary ensemble on the $j$-th and $k$-th qudits. The minimum in Eq.~\eqref{eq:Del_A_loc} is taken over all the two-qudit unitary ensembles involved in the $l$-layer connected block.
With this quantity we can derive the lower bound for the spectral gap as follows:
\begin{proposition}[Spectral gap of circuits with generic architectures] \label{prop:gap_fixed_lnonHaar}
Let $\nu_{A}$ be a unitary ensemble of connected block $A$ of non-Haar random circuit, and $l$ be its connection depth.
Then, 
\begin{equation}
    \Delta_{\nu_A}^{(t)} \geq \left(\Delta_{\mathrm{loc},\nu_{A}}^{(t)}\right)^{l} \cdot\mathrm{LB}_{A}^{(t)},
\end{equation}
where the local spectral gap  $\Delta_{\mathrm{loc},\nu_{A}}^{(t)}$ is defined as Eq.~\eqref{eq:Del_A_loc}, and $\mathrm{LB}_{A}^{(t)}$ is defined as Eqs.~\eqref{eq:spec_LB_comp_Haar_random} and \eqref{eq:spec_LB_incomp_Haar_random}.
\end{proposition}

The skeleton of the derivation of Proposition~\ref{prop:gap_fixed_lnonHaar} follows the same structure as that of Proposition~\ref{prop:Haar_fixed_arch}. It consists of two key steps, the first is to relate the spectral gap with $\gamma^{(t)}$ and the second is to upper bound $\gamma^{(t)}$. Especially in the second step, we can directly repurpose Lemma~\ref{lem:fixed_clustermerge}.
Therefore, the only new task here concerns the first step.
Specifically, we derive a generalized version of Lemma~\ref{lem:fixed_decomp_into_two}, which extends its applicability from Haar random circuits to general non-Haar random circuits, as follows:
\begin{lemma}[Generalization of Lemma~\ref{lem:fixed_decomp_into_two} to non-Haar random circuit] \label{lem:fixed_decomp_into_two_nonHaar}
Let $\nu_A \equiv *_{1\leq j\leq l} \nu_j$ be the $l$-layer connected block of a non-Haar random circuit with a fixed architecture, where $\nu_j$ represents the $j$-th layer. Then, its spectral gap $\Delta_{\nu_A}^{(t)}$ satisfies 
\begin{equation}
\label{seq:lem_fixed_thm7_nHaar}
    \left(1- \Delta_{\nu_A}^{(t)} \right)^2  \leq 1 - \left( \prod_{j=1}^l \alpha_{\nu_j}^{(t)}\right) \left(\prod_{j=2}^{l} \gamma^{(t)}(\nu_j,\ast_{1\leq k < j} \nu_k) \right)
\end{equation}
where $\gamma^{(t)}$ is defined as Eq.~(\ref{seq:gamma_def}) and $\alpha_{\nu_j}^{(t)}$ is defined as
\begin{equation}
\label{seq:alpha_clu_def}
    \alpha_{\nu_j}^{(t)} \equiv 1- \|M_{\nu_j}^{(t)} - P_{c(\nu_j)}^{(t)}\|_{\infty}^2.
\end{equation}
\end{lemma}
\noindent
When $\nu$ is an ensemble for Haar random circuit, Eq.~(\ref{seq:lem_fixed_thm7_nHaar}) is reduced to Eq.~(\ref{seq:lem_fixed_thm7}) since $\alpha^{(t)}_{\nu_j} =1$ for all $j$. This means that Lemma~\ref{lem:fixed_decomp_into_two_nonHaar} is the generalization of Lemma~\ref{lem:fixed_decomp_into_two}.

For the proof of Lemma~\ref{lem:fixed_decomp_into_two_nonHaar}, the decomposition of the projector $P_{c(\nu)}^{(t)}$ into the orthogonal subspaces is essential.
It is known that the intersection of the subspaces where $P_{c(\nu)}^{(t)}$ and $P_{c(\mu)}^{(t)}$ act nontrivially coincides with the subspace that $P_{c(\nu*\mu)}^{(t)}$ acts nontrivially (see e.g., Lemma 17 of Ref.~\cite{brandao2016local}). 
Therefore, the projectors $P_{c(\nu)}^{(t)}$ and $P_{c(\mu)}^{(t)}$ can be decomposed into the projector on the common subspace $P_{c(\nu*\mu)}^{(t)}$ and those on the orthogonal complement, defined as
\begin{equation}
\label{seq:Q_clus_def}
    Q_{\nu\setminus \mu}^{(t)} \equiv P_{c(\nu)}^{(t)} - P_{c(\nu*\mu)}^{(t)}, \quad 
    Q_{\mu\setminus\nu}^{(t)} \equiv P_{c(\mu)}^{(t)}- P_{c(\nu*\mu)}^{(t)}.
\end{equation}
Here, the orthogonal relations $Q_{\nu\setminus \mu}^{(t)} P_{c(\nu*\mu)}^{(t)} = P_{c(\nu*\mu)}^{(t)}Q_{\nu\setminus \mu}^{(t)} =0$ and $Q_{\mu\setminus \nu}^{(t)} P_{c(\nu*\mu)}^{(t)} = P_{c(\nu*\mu)}^{(t)}Q_{\mu\setminus \nu}^{(t)} =0$ are satisfied. By further decomposing the moment operator as $M_{\nu}^{(t)} = P_{c(\nu)}^{(t)} +R_\nu^{(t)}$, where $P_{c(\nu)}^{(t)}R_\nu^{(t)} = R_\nu^{(t)}P_{c(\nu)}^{(t)}=0$, we can obtain the expression
\begin{equation}
\label{seq:mom_tri_decomp}
    M_{\nu}^{(t)} = P_{c(\nu*\mu)}^{(t)} + Q_{\nu\setminus \mu}^{(t)} + R_\nu^{(t)},
\end{equation}
where all pairwise products vanishing.
With $\alpha_{\nu}^{(t)}$ defined as Eq.~\eqref{seq:alpha_clu_def}, we can derive the inequality
\begin{equation}
\label{eq:R_bound_fixed_arch}
    R_\nu^{(t)} R_\nu^{(t)\dag} \leq (1-\alpha_{\nu}^{(t)}) (\mathbbm{1}-P^{(t)}_{c(\nu)})
\end{equation}

Using this decomposition, the following lemma is derived, which is recursively used in the proof of Lemma~\ref{lem:fixed_decomp_into_two_nonHaar}:
\begin{lemma}\label{lem:for_fixed_circuits}
Let $\nu$ and $\mu$ be unitary ensembles composed of local gates, and $\gamma^{(t)} (\nu,\mu)$ and $\alpha_\nu^{(t)}$ be defined as Eqs.~(\ref{seq:gamma_def}) and~(\ref{seq:alpha_clu_def}), respectively. Then, the following inequality is satisfied:
\begin{equation}
\label{seq:lem_fixed}
M_\nu^{(t)} P_{c(\mu)}^{(t)}  M_\nu^{(t) \dag} 
\leq \left[1 - \alpha_\nu^{(t)} \gamma^{(t)} (\nu,\mu)\right] \mathbbm{1} + \alpha_\nu^{(t)} \gamma^{(t)} (\nu,\mu)  P^{(t)}_{c(\nu * \mu)}.
\end{equation}
\end{lemma}

\begin{proof}
Here, we omit the superscript $(t)$ for simplicity. By utilizing the decomposition in Eq.~(\ref{seq:mom_tri_decomp}), we can derive the following inequality:
\begin{equation}
\label{seq:lem_fixed_prf1}
\begin{split}
M_{\nu} P_{c(\mu)} M_{\nu}^{\dag} 
&=  (P_{c(\nu * \mu)} + Q_{\nu \setminus \mu} + R_\nu )( P_{c(\nu * \mu)} + Q_{\mu \setminus \nu})(P_{c(\nu * \mu)} + Q_{\nu \setminus \mu} + R_\nu^\dag) \\ 
&= P_{c(\nu * \mu)} + (Q_{\nu \setminus \mu} + R_\nu ) Q_{\mu \setminus \nu}(Q_{\nu \setminus \mu} + R_\nu^\dag)  \\
&\leq P_{c(\nu * \mu)} +  \| (Q_{\nu \setminus \mu} + R_\nu ) Q_{\mu \setminus \nu}\|_{\infty}^2 (\mathbbm{1} - P_{c(\nu * \mu)} ) \\
&= P_{c(\nu * \mu)} +  \| Q_{\mu \setminus \nu} (Q_{\nu \setminus \mu} + R_\nu^\dag R_\nu)Q_{\mu \setminus \nu}\|_{\infty} (\mathbbm{1} - P_{c(\nu * \mu)} ) .
\end{split}
\end{equation}
In the second line, we use $P_{c(\nu * \mu)}(Q_{\nu \setminus \mu}+R_\nu)=0$ and $P_{c(\nu * \mu)} Q_{\mu \setminus \nu} = 0$.
In the third line, we used the fact that the second term in the second line is Hermite positive semidefinite with a largest eigenvalue given by $\| (Q_{\nu \setminus \mu} + R_\nu ) Q_{\mu \setminus \nu}\|_{\infty}^2$, and it acts on the subspace orthogonal to $\mathrm{supp}[P_{c(\nu * \mu)}]$.

We further upper bound the operator norm appeared in the final line as follows:
\begin{equation}
\label{seq:lem_fixed_prf2}
\begin{split}
\|Q_{\mu \setminus \nu} (Q_{\nu \setminus \mu} + R_\nu^\dag R_\nu)Q_{\mu \setminus \nu} \|_{\infty} 
&\leq \left\|Q_{c(\mu \setminus \nu)} \left(Q_{\nu \setminus \mu} + (1 -\alpha_\nu)(\mathbbm{1}-P_{c(\nu * \mu)} - Q_{\nu \setminus \mu}) \right)Q_{c(\mu \setminus \nu)} \right\|_{\infty}  \\
&=\left\|(1-\alpha_\nu) Q_{\mu \setminus \nu} + \alpha_\nu Q_{\mu \setminus \nu} Q_{\nu \setminus \mu} Q_{\mu \setminus \nu} \right\|_{\infty} \\
&\leq (1-\alpha_\nu) +\alpha_\nu \left\| Q_{\mu \setminus \nu} Q_{\nu \setminus \mu} Q_{\mu \setminus \nu} \right\|_{\infty} \\
&= (1-\alpha_\nu) +\alpha_\nu \left\| Q_{\mu \setminus \nu} Q_{\nu \setminus \mu}\right\|_{\infty}^2 \\ 
&= (1-\alpha_\nu) +\alpha_\nu \left\| P_{c(\nu)} P_{c(\mu)} - P_{c(\nu * \mu)} \right\|_{\infty}^2 \\ 
&= (1-\alpha_\nu) +\alpha_\nu (1-\gamma (\nu,\mu)) \\
&= 1-\alpha_\nu \gamma (\nu,\mu),
\end{split}
\end{equation}
where in the first line, we used Eq.~\eqref{eq:R_bound_fixed_arch} and $\left\| A B A^\dag \right\|_{\infty} \leq \left\| A C A^\dag \right\|_{\infty}$ for all linear operators $A$, $B$, and $C$ satisfying $0 \leq B \leq C$. 
From Eqs.~(\ref{seq:lem_fixed_prf1}) and (\ref{seq:lem_fixed_prf2}), we can derive Eq.~(\ref{seq:lem_fixed}) as follows:
\begin{equation}
\begin{split}
    M_\nu P_{c(\mu)} M_\nu^{\dag} 
    &\leq  P_{c(\nu * \mu)} +  (1-\alpha_\nu \gamma (\nu,\mu)) (\mathbbm{1} - P_{c(\nu * \mu)} ) \\
    &= (1- \alpha_\nu \gamma (\nu,\mu)) \mathbbm{1} +  \alpha_\nu \gamma (\nu,\mu) P_{c(\nu * \mu)}.
\end{split}
\end{equation}
\end{proof}

\begin{proof}[Proof of Lemma~\ref{lem:fixed_decomp_into_two_nonHaar}]
We omit the superscript $(t)$ for simplicity. The moment operator satisifies the following inequality:
\begin{equation}
\label{seq:prf_thm7_nHaar1}
\begin{split}
    (M_{\nu} -P_{\rm all}) (M_{\nu} -P_{\rm all})^\dag  &= M_{\nu_l} \cdots M_{\nu_2} M_{\nu_1}M_{\nu_1}^\dag M_{\nu_2}^\dag  \cdots M_{\nu_l}^\dag - P_{\rm all} \\
    &\leq (1-\alpha_{\nu_1}) \mathbbm{1} + \alpha_{\nu_1} M_{\nu_l} \cdots M_{\nu_2} P_{c(\nu_1)}M_{\nu_2}^\dag  \cdots M_{\nu_l}^\dag - P_{\rm all} \\
    &\leq \left[ 1- \left( \prod_{j=1}^l \alpha_{\nu_j}\right) \left(\prod_{j=2}^{l} \gamma(\nu_j,\ast_{1\leq k < j} \nu_k) \right) \right] (\mathbbm{1} - P_{\rm all}),
\end{split}
\end{equation}
where we used Lemma \ref{lem:for_fixed_circuits} recursively in the third line. We can derive Eq.~(\ref{seq:lem_fixed_thm7_nHaar}) by taking the operator norm of both sides of Eq.~(\ref{seq:prf_thm7_nHaar1}).
\end{proof}

By utilizing Lemmas~\ref{lem:fixed_clustermerge} and \ref{lem:fixed_decomp_into_two_nonHaar}, we can derive Proposition~\ref{prop:gap_fixed_lnonHaar} as follows:
\begin{proof}[Proof of Proposition~\ref{prop:gap_fixed_lnonHaar}]
We here omit the superscript $(t)$ for simplicity.
By utilizing Lemma~\ref{lem:ltog_gap}, we have 
\[
\alpha_{\nu_j} = 1 - \|M_{\nu_j} -  P_{c(\nu_j)}\|_{\infty}^2 \geq 1 - (1-\Delta_{\mathrm{loc},\nu_j})^2 \geq 1 - (1-\Delta_{\mathrm{loc},\nu_A})^2,
\]
where the final inequality follows from the definition of $\Delta_{\mathrm{loc},\nu_A}$. 
Therefore, we can derive the following inequality:
\begin{equation}
\label{seq:con_block_gap1}
\begin{split}
\left( 1- \Delta_{\nu_A} \right)^2 
&\leq 1 - \left[1- (1-\Delta_{\mathrm{loc},\nu_A})^2 \right]^{l} \left(\prod_{j=2}^{l} \gamma (\nu_j,\ast_{1\leq k < j} \nu_k) \right) \\
&\leq  1 - \left[1- (1-\Delta_{\mathrm{loc},\nu_A})^2 \right]^{l} \left[1 - (1-f)^2\right]^{l-1} \\
&\leq \left[1- \left(\Delta_{\mathrm{loc},\nu_A}\right)^{l} f^{l-1} \right]^2 ,
\end{split}
\end{equation}
where we used Lemma~\ref{lem:fixed_decomp_into_two_nonHaar} in the first line, Lemma~\ref{lem:fixed_clustermerge} in the second line, and  Lemma~\ref{lem:simple_alg} (discussed below) in the third line. 
Here, the function $f$ is $f_{\rm c}$ when the architecture $A$ is complete, and $f$ is $f_{\rm i}$ otherwise.
From Eq.~(\ref{seq:con_block_gap1}), we can obtain the inequality
\begin{equation*}
    \Delta_{\nu_A} \geq \left(\Delta_{\mathrm{loc},\nu_A} \right)^{l}  \mathrm{LB}_{A}^{(t)},
\end{equation*}
which complete the proof.
\end{proof}

We here introduce Lemma~\ref{lem:simple_alg}, which is used in the proof of Proposition~\ref{prop:gap_fixed_lnonHaar}.
\begin{lemma} \label{lem:simple_alg}
Let $k, l\in\mathbb{Z}$ and $x, y\in[0, 1]$. 
Then,
\begin{equation}
\label{seq:alg_simple}
    \left(1- (1-x)^l(1-y)^k\right)^2 \geq 1 - (1-x^2)^l (1-y^2)^k.
\end{equation}
\end{lemma}

\begin{proof}
The difference between the left- and right-hand sides of Eq.~(\ref{seq:alg_simple}) can be evaluated as follows:
\begin{align}
    &\left[1- (1-x)^l(1-y)^k\right]^2-\left[1 - (1-x^2)^l (1-y^2)^k \right]  \nonumber\\
    &= 2(1-x)^l(1-y)^k \left[\frac{(1+x)^l(1+y)^k + (1-x)^l(1-y)^k}{2} -1\right] \nonumber\\
    &= 2(1-x)^l(1-y)^k \left[\frac{(1+x)^l + (1-x)^l}{2}\cdot\frac{(1+y)^k + (1-y)^k}{2} + \frac{(1+x)^l - (1-x)^l}{2}\cdot\frac{(1+y)^k - (1-y)^k}{2} -1\right] \nonumber\\
    &\geq 2(1-x)^l(1-y)^k \left\{\left[\frac{(1+x) + (1-x)}{2}\right]^l\cdot\left[\frac{(1+y) + (1-y)}{2}\right]^k+0-1\right\}  \nonumber\\
    &= 0, 
\end{align}
where we used Jensen's inequality in the fourth line.
\end{proof}

Finally, we provide an upper bound for the unitary design formation depth in a non-Haar random circuit with generic fixed architecture. 
We here consider the circuit where the $l$-layer connected blocks are convoluted, and denote the unitary ensemble for $i$-th block as $\nu_{A,i}$.
There are two key quantities in upper bounding the formation depth. The first is the averaged local spectral gap defined as 
\begin{equation}
\label{seq:ave_loc_spec_arb}
    \Delta_{\mathrm{loc}}^{(t)} \equiv \inf_{K\in \mathbb{N}} \frac{1}{K} \sum_{i=1}^K \Delta_{\mathrm{loc}, \nu_{A,i}}^{(t)},
\end{equation}
where $\Delta_{\mathrm{loc}, \nu_{A,i}}^{(t)}$ is the local spectral gap, defined as Eq.~\eqref{eq:Del_A_loc}. When all blocks have the same unitary ensemble, i.e., $\nu_{A,i} = \nu_{A}$ for all $i$, the averaged local spectral gap is reduced to $\Delta_{\mathrm{loc}}^{(t)} = \Delta_{\mathrm{loc}, \nu_{A}}^{(t)}$. This is the situation discussed in the main text. %In Theorem~\ref{thm:depth_generic}, we assume $\Delta_{\mathrm{loc}}^{(t)} >0$ to derive the upper bound of the formation depth.

Another key quantity is the circuit depth $L_{A}^{\rm H}$ defined as
\begin{equation}
\label{seq:depth_arb_Haar}
    L_{A}^{\rm H} \equiv [f(N,t,d)]^{-l+1}\cdot l \left(2Nt \log d - \log \varepsilon \right),
\end{equation}
where $f=f_{\rm c}$ if the architecture is complete and $f=f_{\rm i}$ if the architecture is incomplete.
This is the upper bound for the circuit depth required for the Haar random circuit composed of the blocks $\{\nu_{A,i}^{\rm H}\}$ to form an $\varepsilon$-approximate unitary $t$-design, derived from Proposition~\ref{prop:Haar_fixed_arch} and Lemma~\ref{lem:basic} in Ref.~\cite{belkin2023approximate}.
Note that a single block $\nu_{A,i}^{\rm H}$ is composed of $l$ layers and thus $L_{A}^{\rm H}$ has an additional factor $l$.
With these quantities, we present the main theorem of this subsection as follows:
\begin{theorem}[Theorem 2 in the main text] \label{thm:depth_generic}
Consider a non-Haar random circuit with generic fixed architecture consisting of $l$-layer connected blocks. 
Then, this circuit forms an $\varepsilon$-approximate unitary $t$-design at a circuit depth of 
\begin{equation}
\label{seq:depth_gen}
    L \geq (\Delta_{\mathrm{loc}}^{(t)})^{-l} L^{\rm H}_{A}, 
\end{equation}
where the averaged local spectral gap $\Delta_{\mathrm{loc}}^{(t)}$ is defined as Eq.~\eqref{seq:ave_loc_spec_arb}, and $L^{\rm H}_{A}$ is defined as Eq.~\eqref{seq:depth_arb_Haar}.
\end{theorem}

\begin{proof}
We here abbreviate the superscript $(t)$ and assume that $L/l$ is a natural number for the simplicity of notation.
Then, we can derive the following inequality:
\begin{equation}
\label{seq:fixed_depth}
\begin{split}
    \sum_{i=1}^{L/l} \Delta_{\nu_{A,i}}
    &\geq \sum_{i=1}^{L/l} (\Delta_{\mathrm{loc},\nu_{A,i}})^{l} f^{l-1} \\
    &\geq f^{l-1} \cdot \left(\Delta_{\mathrm{loc}}\right)^l \frac{L}{l}\\ 
    &= 2Nt \log d - \log \varepsilon.
\end{split}
\end{equation}
Here, we used Proposition~\ref{prop:gap_fixed_lnonHaar} in the first line, and the concavity of $x^l$ with respect to $x$ in the second line.
Therefore, from Lemma~\ref{lem:basic}, we can show that this circuit forms an $\varepsilon$-approximate unitary $t$-design after circuit depth $L$.
\end{proof}

\begin{rem}
When we take $A$ as the 1D brickwork architecture, Proposition~\ref{prop:gap_fixed_lnonHaar} becomes
\begin{equation*}
    \Delta_{\nu_{\rm bw}}^{(t)} \geq \left(\Delta_{\mathrm{loc},\nu_{\rm bw}}^{(t)}\right)^{2} \cdot \min_{1\leq m\leq N} \Delta_{\nu_{\rm bw}^{\rm H}}(m,d),
\end{equation*}
which is almost identical to Proposition~2 derived in Appendix B, except for the minimization over the number of qudits \(m\). Furthermore, the proof techniques employed in Appendix B and in Proposition~\ref{prop:gap_fixed_lnonHaar} are analogous, as evidenced by the correspondence between Eqs.~(B5) and (B7) in Appendix B and Eqs.~\eqref{seq:lem_fixed_prf1} and \eqref{seq:lem_fixed_prf2}, respectively.
\end{rem}

\begin{rem}[Extension of the applicability of the detectability lemma]
In Ref.~\cite{belkin2023approximate}, a lower bound for the spectral gap of generic fixed-architecuture circuit $\Delta_{\nu_{A}^{\rm H}}^{(t)}$ is provided with $\mathrm{LB}_{A}^{(t)}$ (Proposition~\ref{prop:Haar_fixed_arch} in the Supplemental Material). This term $\mathrm{LB}_{A}^{(t)}$ is evaluated with the detectability lemma, since it is defined with the spectral gap for the 1D brickwork circuit as Eqs.~\eqref{eq:spec_LB_comp_Haar_random} and \eqref{eq:spec_LB_incomp_Haar_random}.
In this section, we lower bound the spectral gap for non-Haar random circuit $\Delta_{\nu_{A}}^{(t)}$ with $\mathrm{LB}_{A}^{(t)}$. 
This result further extends the known applicability of the detectability lemma to non-Haar random circuits.
\end{rem}

\section{Patchwork circuit} \label{Ss:patchwork}
\begin{figure*}[]
    \centering
    \includegraphics[width=0.65\textwidth]{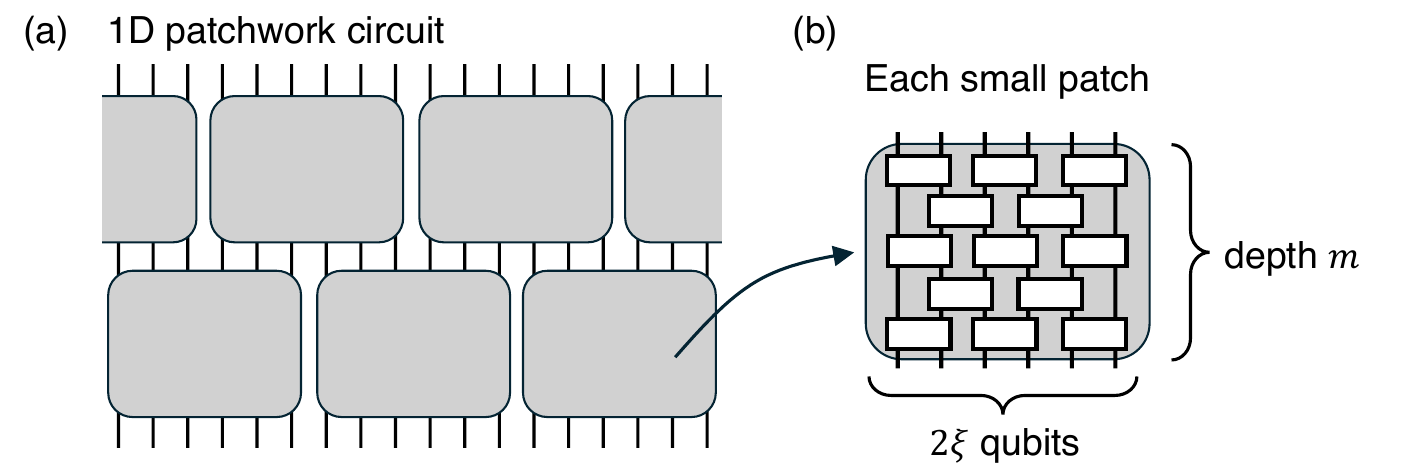}
    \caption{Schematics for 1D patchwork circuit introduced in Ref.~\cite{schuster2024random}. (a) Construction of patchwork circuit with small patches. Small patch of circuits on $2\xi$-qubit system with depth $m$ are glued together in a two-layer brickwork form. (b) The construction of each small patch. Each small patch is composed of normal random circuit, e.g., 1D brickwork circuit.}
    \label{fig:patch}
\end{figure*}

In this section, we discuss the unitary design formation in patchwork circuit structure proposed in Refs.~\cite{schuster2024random,laracuente2024approximate}.
Specifically, we demonstrate that $O(\log N)$-depth formation with local Haar ensemble in this structure, shown in Refs.~\cite{schuster2024random,laracuente2024approximate}, can be directly extended to the situation with local non-Haar ensemble, by utilizing the results of Secs.~\ref{Ss:single_layer_repetition} and \ref{Ss:multilayer_circuit}.

The essential idea of their shallow-depth formation is to glue together small patches of normal random circuits in a two-layer brickwork form, as shown in Fig.~\ref{fig:patch}(a).  
They prove the following theorem which states that circuits with this special structure form approximate unitary designs efficiently:
\begin{theorem}[Theorem 1 of Ref.~\cite{schuster2024random}]\label{thm:schuster}
Given any approximation error $\varepsilon \leq 1$, suppose each small random unitary in the two-layer brickwork ensemble is drawn from a relative error $\frac{\varepsilon}{N}$-approximate unitary $t$-design on $2\xi$ qubits with circuit depth $m$.  
Then, this circuit forms a relative error $\varepsilon$-approximate unitary $t$-design on $N$ qubits with depth $2m$, whenever the local patch size is at least $\xi \geq \log_2(Nt^2/\varepsilon)$.
\end{theorem}
\noindent
This theorem holds regardless of how the random unitary is formed within each small patch of $2\xi$ qubits.  
The only requirements are:  
(i) the size of each small subsystem must satisfy $\xi \geq \log_2(Nt^2/\varepsilon)$, and  
(ii) each small patch on $2\xi$ qubits must create a relative error $\frac{\varepsilon}{N}$-approximate unitary $t$-design.  

A simple way to satisfy these requirements is to first fix the subsystem size as $\xi = \lceil\log_2(Nt^2/\varepsilon) \rceil$ to meet condition (i), and then construct each small patch with Haar random 1D brickwork circuit with sufficient depth $m^{\rm H}$ to meet condition (ii).
From Ref.~\cite{chen2024incompressibility}, $m^{\rm H}$ can be taken as
\begin{equation}
\label{seq:l_patch_Haar}
    m^{\rm H} = O\left(\left(\xi t + \log (N/\varepsilon)\right) \polylog\ t \right) = O\left(\log (N/\varepsilon) \cdot t \ \polylog \ t \right),
\end{equation}
which implies $O(\log N)$-depth formation of an approximate $t$-design.

Since Theorem~\ref{thm:schuster} does not impose any restriction to each small patch, we can demonstrate the unitary design formation with non-Haar random circuits in essentially the same way.  
We again choose the size of each small subsystem as $\xi = \lceil\log_2(Nt^2/\varepsilon)\rceil$, satisfying condition (i), and construct each small patch with a non-Haar random 1D brickwork circuit. 
From Appendix B, condition (ii) is satisfied by taking the circuit depth $m$ as 
\begin{equation}
\label{eq:patch_1d_brick_nonHaar}
    m = \left(\Delta_{\rm loc}^{(t)}\right)^{-2} \cdot m^{\rm H} = O\left(\log (N/\varepsilon) \cdot t \ \polylog \ t \right),
\end{equation}
where $m^{\rm H}$ is the circuit depth for corresponding Haar random circuit. 
We note that approximate unitary design used in Theorem~\ref{thm:schuster} is the definition with the relative error, not Definition~\ref{def:mom_approx_unitary_design} employed in the discussion of Secs.~\ref{Ss:single_layer_repetition} and \ref{Ss:multilayer_circuit}. However, we can show the formation with the circuit depth $m$ (eq.~\eqref{eq:patch_1d_brick_nonHaar}), since Definition~\ref{def:mom_approx_unitary_design} implies a relative error $\varepsilon$-approximate $t$-design.

Finally, we comment on the extensibility of non-Haar random circuits with this patchwork structure.
While we have focused on the 1D circuit structure shown in Fig.\ref{fig:patch}, patchwork circuits in higher spatial dimensions have also been studied in Refs.\cite{schuster2024random,laracuente2024approximate}.
Using the same reasoning as in this section, we can show that non-Haar random circuits with such higher-dimensional patchwork structures achieve unitary designs at the same circuit depth order as their Haar random counterparts.
Moreover, although we have considered each small patch as a 1D brickwork circuit, other architectures can also be used.
Even in such cases, it can be shown that the required circuit depth scales in the same order as that of the corresponding Haar random circuit, by utilizing Theorems~\ref{thm:depth_sin_rep} and \ref{thm:depth_generic}.

\section{Spectral gap of local unitary ensemble} \label{Ss:local_unitary}
In this section, we investigate how the spectral gap $\Delta_{\nu}^{(t)}$ of a unitary ensemble $\nu$ on $\mathrm{U}(q)$ depends on the moment order $t$. This analysis is crucial for understanding the $t$-dependence of $\Delta_{\rm loc}^{(t)}$, with which the prefactor of the depth upper bounds are defined, as shown in Table~\ref{tab:summary_supp}. 
The key tool in this discussion is Lemma~\ref{lem:spec_gap_univ}, which provides the necessary and sufficient condition for the spectral gap to satisfy $ \Delta_{\nu}^{(t)} \geq \Omega((\log t)^{-2}) $. In Sec.~\ref{Sss:univ_spec}, we introduce Lemma~\ref{lem:spec_gap_univ} and discuss its implications, illustrating them with some comprehensive examples. In Sec.~\ref{Sss:prf_univ_spec}, we present the proof of Lemma~\ref{lem:spec_gap_univ}. 

\subsection{Condition for nonzero spectral gap} \label{Sss:univ_spec}

In this section, we assume that the ensemble $\nu$ consists of a finite gate set and is described as $\nu \equiv \{(p_i, U_i)\}_i$, where unitary operation $U_i$ is applied with the probability $p_i$.
However, we expect that a similar discussion applies more generally, e.g., when $\nu$ is supported on a continuous gate set. 

When discussing the condition for nonzero spectral gap, the concept of universal gate set is crucial.
\begin{dfn}[Universal gate set]
A discrete gate set $\{U_i\}_i$ is said to form a universal gate set on $\mathrm{U}(q)$ if and only if the following condition is satisfied: For any $V \in \mathrm{U}(q)$ and $\varepsilon > 0$, there exists $ V_{\rm approx} \equiv U_{i_m} U_{i_{m-1}} \cdots U_{i_1} $, where each $U_{i_j}$ is drawn from that gate set, such that
\begin{equation}
    D(V,V_{\rm approx}) \equiv \min_{\varphi \in [0,2\pi)} \| V - e^{i\varphi} V_{\rm approx}\|_{\infty} \le \varepsilon.
\end{equation}
\end{dfn}

We further introduce the notion of adjoint measure of a certain ensemble $\nu \equiv \{(p_i,U_i)\}_i$, which is defined as $\nu^\dag \equiv \{(p_i,U_i^\dag)\}_i$. 
Based on these concepts, the following lemma is formulated:
\begin{lemma}\label{lem:spec_gap_univ}
Let $\nu\equiv \{(p_i,U_i)\}_i$ be a unitary ensemble, and $\nu^\dag$ be its adjoint ensemble.
Then, if $ \nu^\dag * \nu = \{(p_i p_j, U_i^\dag U_j)\}_{i,j}$ has support on a universal gate set, there exists a positive constant $B>0$ such that for all $t$, 
\begin{equation}
\label{seq:LB_spec_gap_univ}
    \Delta_{\nu}^{(t)} > B(\log t)^{-2} .
\end{equation}
Conversely, if $ \nu^\dag * \nu $ does not have support on a universal gate set, then $\Delta_{\nu}^{(t)} = 0$ for sufficiently large $t$.
\end{lemma}

A subtle but important point here is that even if $\nu$ itself has support on a universal gate set, its spectral gap can still vanish: we need a slightly more strong requirement that the measure $\nu^\dag * \nu$ does.
The illustrative example is the following:
\begin{ex}[$\nu$ has support on a universal gate set, but its spectral gap becomes zero]\label{ex:two_element}
Consider an ensemble $\nu\equiv \{(p,U_1),(1-p,U_2)\}$, where $U_1$ and $U_2$ form a universal gate set on $\mathrm{U}(q)$. Then the ensemble $ \nu^\dag * \nu $ includes only the unitaries  $\{\mathbbm{1}, U_1^\dag U_2, U_2^\dag U_1\}$. These gates do not form a universal gate set because $ U_1^\dag U_2 = e^{iA} $ and $ U_2^\dag U_1 = e^{-iA} $ for some Hermitian operator $A$, making it impossible to approximate any unitary $V$ that does not commute with $A$. By Lemma~\ref{lem:spec_gap_univ}, the spectral gap of such an ensemble thus becomes $ \Delta_{\nu}^{(t)} = 0$ for sufficiently large $t$, even though $\nu$ has support on a universal gate set.
\end{ex} 
\noindent
Example~\ref{ex:two_element} includes the set of the $T$ gate and the Hadamard gate in a qubit system ($q=2$). 
For general $q\geq 2$, Ref.~\cite{lloyd1995almost} shows that two random unitaries in $\mathrm{U}(q)$ almost always form a universal gate set. Hence, Example~\ref{ex:two_element} certainly exists for general $q$.

However, introducing a small amount of redundancy can immediately ensure a nonzero spectral gap.
For example, if an ensemble $\nu \equiv \{(p_i,U_i)\}_{i=0}^m$ has support on the identity $U_0 \equiv \mathbbm{1}$ in addition to a universal gate set $\{U_i\}_{i=1}^m$, $\nu^\dag * \nu$ has support on a universal gate set, and therefore satisfies $\Delta_{\nu}^{(t)} > \Omega((\log t)^{-2})$ for all $t$ from Lemma~\ref{lem:spec_gap_univ}.
In practical situations, such inclusion of the identity is easy since it simply means to do nothing.
Hence, we adopt the following terminology, which is also empolyed in the main text:
\begin{dfn} \label{def:univ_ens}
A unitary ensemble $\nu$ is said to contain a universal gate set, if it has support on a universal gate set \textit{along with the identity operator}.
\end{dfn}

\subsection{Proof of Lemma~\ref{lem:spec_gap_univ}} \label{Sss:prf_univ_spec}

In this subsection, we provide the proof of Lemma~\ref{lem:spec_gap_univ}. 
We first introduce Lemmas~\ref{lem:slow_decay_prev} and \ref{lem:univ_eig}, which play important roles in the proof of Lemma~\ref{lem:spec_gap_univ}.
\begin{lemma} [Theorem 5 in Ref.~\cite{oszmaniec2021epsilon}] \label{lem:slow_decay_prev}
Let $\nu$ be an arbitrary unitary ensemble on $\mathrm{U}(q)$, and $\Delta_{\nu}^{(t)} \equiv 1 - \|M_{\nu}^{(t)} - M_{\rm Haar}^{(t)} \|_{\infty}$ be its spectral gap for $t$-th moment operator.
Then, there exists a natural number $t_0$ and a constant $c>0$, which are dependent only on $q$, such that for arbitrary $t>t_0$ the following inequality is satisfied:
\begin{equation}
\label{seq:slow_decay_Dloc_prev}
    \Delta_{\nu}^{(t)} \geq \frac{c\cdot\Delta_{\nu}^{(t_0)}}{(\log t)^2}.
\end{equation}
\end{lemma}
\noindent
The proof of this lemma is provided in Refs.~\cite{oszmaniec2021epsilon,varju2013random}, and we here abbreviate it.

In Lemma~\ref{lem:univ_eig}, we use the spectral radius of the operator $A$, denoted by $\rho(A)$, which is defined as the largest absolute value among the eigenvalues of $A$.
\begin{lemma} \label{lem:univ_eig}
Let $\nu \equiv \{(p_i, U_i)\}_i$ be a unitary ensemble, and $M_{\nu}^{(t)}$ be its $t$-th moment operator.  
Then, $\rho(M_{\nu}^{(t)} - M_{\rm Haar}^{(t)}) < 1$ for all $t$ if and only if $\nu$ has support on a universal gate set.
\end{lemma}
\noindent
While Lemma~\ref{lem:spec_gap_univ} characterizes the spectral gap of the moment operator, Lemma~\ref{lem:univ_eig} is about its spectral radius. These two quantities satisfy the relation $1 -\Delta_{\nu}^{(t)} \geq \rho(M_{\nu}^{(t)}-M_{\rm Haar}^{(t)})$ for general $\nu$. Therefore, Lemma~\ref{lem:univ_eig} and the discussion in the previous subsection (such as Example~\ref{ex:two_element}) do not contradict: when an ensemble $\nu$ has support on a universal gate set, its spectral radius is always strictly less than $1$, while its sepctral gap may vanish in some cases.

While a part of the proof of Lemma~\ref{lem:univ_eig} has already appeared in the previous works \cite{harrow2009random,brown2010convergence}, we here provide a full proof to make the manuscript self-contained.
\begin{proof}[Proof of Lemma~\ref{lem:univ_eig}]
We first prove the ``if'' part. 
Assume that $\rho (M_{\nu}^{(t)} - M_{\rm Haar}^{(t)}) = 1$ for some $t$, and suppose $\nu$ has support on a universal gate set. We will derive a contradiction under these conditions.

From $\rho (M_{\nu}^{(t)} - M_{\rm Haar}^{(t)}) = 1$, we obtain a normalized vector $\ket{\psi}$ which acts on the orthogonal complement of the subspace $\mathrm{Inv}(\mathrm{U}(q)^{\otimes t, t})$ and satisfies $M_{\nu}^{(t)} \ket{\psi} = e^{i \theta} \ket{\psi}$ for some $\theta \in [0,2\pi)$.  
Therefore, for any natural number $K$, we have
\begin{equation}
\label{seq:mom_eig_1}
1 = \bigl|\!\bra{\psi} M_{\nu^{*K}}^{(t)} \ket{\psi}\!\bigr|
  = \Bigl|\ExnuK\!\bigl[\bra{\psi} U^{\otimes t, t} \ket{\psi}\bigr]\Bigr|
  \;\leq\; \ExnuK\!\Bigl[\bigl|\bra{\psi}U^{\otimes t, t}\ket{\psi}\bigr|\Bigr].
\end{equation}
Since $\bigl|\bra{\psi}U^{\otimes t, t}\ket{\psi}\bigr| \leq 1$ for any unitary $U$, it follows from Eq.~\eqref{seq:mom_eig_1} that 
$\left| \bra{\psi} U^{\otimes t, t}\ket{\psi}\right| =1$ for all $U \sim \nu^{*K}$. 
Furthermore, by the equality condition in Eq.~\eqref{seq:mom_eig_1}, there exists a fixed $\theta_K \in [0,2\pi)$ such that, for all $U \sim \nu^{*K}$,
\begin{equation}
\label{seq:mom_eig_contra}
\bra{\psi} U^{\otimes t, t} \ket{\psi} \;=\; e^{i \theta_K}.
\end{equation}

On the other hand, since $\ket{\psi}$ acts on the orthogonal complement of $\mathrm{Inv}(\mathrm{U}(q)^{\otimes t, t})$, there must be a unitary $W \in \mathrm{U}(q)$ and a constant $\varepsilon > 0$ such that
\begin{equation}
\label{seq:schur_weyl}
\Bigl\|W^{\otimes t,t}\ket{\psi} \;-\; \ket{\psi}\Bigr\| \;>\; \varepsilon,
\end{equation}
where $\|\ket{v}\| \equiv \sqrt{\braket{v|v}}$.  
From Eq.~\eqref{seq:schur_weyl}, we immediately obtain
\begin{equation}
\label{seq:mom_eig_2}
\begin{split}
\Bigl|1 - \bra{\psi} W^{\otimes t,t} \ket{\psi}\Bigr|
&\;\ge\; 1 - \mathrm{Re}\bigl[\bra{\psi} W^{\otimes t,t}\ket{\psi}\bigr]\\
&\;=\; \tfrac12 \Bigl\|W^{\otimes t,t}\ket{\psi} - \ket{\psi}\Bigr\|^2\\
&\;>\; \tfrac{\varepsilon^2}{2}.
\end{split}
\end{equation}

Next, we use the fact that $\nu$ has support on a universal gate set. By this property, there exists a unitary $W_{\rm approx} \sim \nu^{*K}$ such that $D(W,W_{\rm approx}) < \frac{\varepsilon^2}{8t}$, for sufficiently large $K$. This implies
\begin{equation}
\label{seq:mom_eig_3}
\begin{split}
\Bigl|\bra{\psi}W^{\otimes t,t}\ket{\psi} \;-\; \bra{\psi}W_{\rm approx}^{\otimes t, t}\ket{\psi}\Bigr|
&\;=\;\Bigl|\mathrm{tr}\Bigl[\ket{\psi}\bra{\psi}\bigl(W^{\otimes t,t} - W_{\rm approx}^{\otimes t, t}\bigr)\Bigr]\Bigr|\\
&\;\le\;\|\ket{\psi}\bra{\psi}\|_1 \;\bigl\|\,W^{\otimes t,t} - W_{\rm approx}^{\otimes t, t}\bigr\|_{\infty}\\
&\leq 2t D(W,W_{\rm approx})\;<\;\tfrac{\varepsilon^2}{4}.
\end{split}
\end{equation}
Similarly, we can choose $J \sim \nu^{*K}$ such that $D(\mathbbm{1}, J) < \frac{\varepsilon^2}{8t}$, yielding
\begin{equation}
\label{seq:mom_eig_4}
\;\Bigl|\bra{\psi} \mathbbm{1}^{\otimes t,t}\ket{\psi} - \bra{\psi} J^{\otimes t,t}\ket{\psi}\Bigr|
\;<\;\tfrac{\varepsilon^2}{4}.
\end{equation}
Since we have $\bra{\psi} W_{\rm approx}^{\otimes t,t}\ket{\psi} = \bra{\psi} J^{\otimes t,t}\ket{\psi}= e^{i\theta_K}$ from Eq.~\eqref{seq:mom_eig_contra}, Eqs.~\eqref{seq:mom_eig_3} and \eqref{seq:mom_eig_4} can be reformulated as 
\begin{equation*}
    \Bigl|\bra{\psi}W^{\otimes t, t}\ket{\psi} - e^{i\theta_K} \Bigr| < \tfrac{\varepsilon^2}{4}, 
    \quad \Bigl| 1- e^{i\theta_K}\Bigr| < \tfrac{\varepsilon^2}{4}.
\end{equation*}
These two inequalities contradict with Eq.~\eqref{seq:mom_eig_2} from the triangle inequality, which completes the proof of the ``if'' part.

We next provide the proof of the ``only if'' part.
For that purpose, we assume $\rho (M_{\nu}^{(t)} - M_{\rm Haar}^{(t)}) < 1$ for all $t$, and suppose $\nu$ does not have support on a universal gate set. We will derive a contradiction under these conditions.

Since $\nu$ does not have support on a universal gate set, there exists a unitary $V \in \mathrm{U}(q)$ and a constant $\varepsilon > 0$ such that the following statement is satisfied:
\begin{equation}
    {}^\forall U \in G(\nu), \quad D(V,U) > \varepsilon,
\end{equation}
where $G(\nu) \equiv \{U| {}^\exists K \in \mathbb{N} \ \mathrm{s.t.} \ U \sim \nu^{*K}\}$ is the set of unitaries generated by $\nu$. 
Therefore, for any unitary operator $U \in G(\nu)$, the value of the function $f_V(U) \equiv \frac{|\tr[V^\dag U]|}{q}$ can be upper bounded as follows:
\begin{equation}
\label{seq:mom_eig_onlyif_1}
\begin{split}
    f_V(U) &= \frac{1}{q} \max_{\theta \in [0,2\pi) } \mathrm{Re} \left[ e^{i\theta} \tr[V^\dag U]\right] \\
    &= \frac{1}{q} \left( q - \frac{1}{2} \min_{\theta \in [0,2\pi)}\|V - e^{i\theta} U\|_{2}^2 \right) \\
    &\leq \frac{1}{q} \left( q - \frac{1}{2q} \min_{\theta \in [0,2\pi)}\|V - e^{i\theta} U\|_{\infty}^2 \right) \\
    &\leq 1 - \frac{1}{2q^2} D(V,U)^2 < 1 - \frac{\varepsilon^2}{2q^2}.
\end{split}
\end{equation}
By defining $\varepsilon^\prime \equiv \frac{\varepsilon^2}{4q^2}$ and 
\begin{math}
    g_{s,V}(U) \equiv \left(\varepsilon^\prime + f_V(U)^2 \right)^s,
\end{math}
we can derive the following inequality from Eq.~(\ref{seq:mom_eig_onlyif_1}):
\begin{equation}
\label{seq:mom_eig_onlyif_contra}
    \ExnuK [g_{s,V}(U)] < (1-\varepsilon^\prime)^s.
\end{equation}

On the other hand, by introducing the notion of $\alpha$-ball around $V\in \mathrm{U}(q)$ as
\begin{math}
    B(V,\alpha) \equiv \{U | D(U,V) \leq \alpha \},
\end{math}
we can derive the lower bound for the average of $g_{s,V}(U)$ over Haar measure as follows:
\begin{equation}
\label{seq:mom_eig_g_LB}
\begin{split}
    \ExHaar [g_{s,V}(U)] &= \int_{\mathrm{U}(q)} g_{s,V}(U) d\mu_{\rm Haar} \\
    &\geq \int_{B(V,\sqrt{\varepsilon^\prime})} g_{s,V}(U) d\mu_{\rm Haar} \\
    &\geq \int_{B(V,\sqrt{\varepsilon^\prime})} d\mu_{\rm Haar} \\
    &\geq \left(\frac{\sqrt{\varepsilon^\prime}}{C}\right)^{q^2-1},
\end{split}
\end{equation}
where the inequality in the second line follows from the positivity of $g_{s,V}$, and the inequality in the final line follows from the fact that
\begin{math}
    \mathrm{Vol}(B(V,\alpha)) \equiv \int_{B(V,\alpha)} d\mu_{\rm Haar} \geq \left(\frac{\alpha}{C}\right)^{q^2-1},
\end{math}
with $C=\frac{1}{9\pi}$ (see Ref.~\cite{oszmaniec2021epsilon} for the derivation).
The inequality in the third line follows from the fact that $g_{s,V}(U)>1$ is satisfied for any $U\in B(V,\sqrt{\varepsilon^\prime})$, which is derived from the following inequality:
\begin{equation*}
\begin{split}
    f_V(U) &= \frac{1}{q} \left( q - \frac{1}{2} \min_{\theta \in [0,2\pi)}\|V - e^{i\theta} U\|_{2}^2 \right) \\
    &\geq \frac{1}{q} \left( q - \frac{q}{2} \min_{\theta \in [0,2\pi)}\|V - e^{i\theta} U\|_{\infty}^2 \right) \\
    &= 1 - \frac{D(U,V)^2}{2}.
\end{split}
\end{equation*}

Furthermore, since we now assume $\rho ( M_{\nu}^{(s)} -  M_{\rm Haar}^{(s)}) < 1$ for any $s$, we can take natural numbers $t$ and $K$ so that the following inequalities are satisfied:
\begin{gather}
    (1-\varepsilon^\prime)^{t} < \frac{1}{2} \left(\frac{\sqrt{\varepsilon^\prime}}{C}\right)^{q^2-1}, \label{seq:mom_eig_onlyif_t_fix}\\
    \delta_{t,K} \equiv \max_{1\leq s \leq t} \left\|M_{\nu^{*K}}^{(s)} -M_{\rm Haar}^{(s)} \right\|_{\infty} < \frac{1}{2(1+\varepsilon^\prime)^t} \left(\frac{\sqrt{\varepsilon^\prime}}{C}\right)^{q^2-1}. \label{seq:mom_eig_onlyif_K_fix}
\end{gather}
Here, we first fix $t$ so that Eq.~(\ref{seq:mom_eig_onlyif_t_fix}) holds, and then take sufficiently large $K$, which satisfies Eq.~(\ref{seq:mom_eig_onlyif_K_fix}).
For such $t$ and $K$, the following inequality can be derived:
\begin{equation}
\label{seq:mom_eig_K_Haar_gap}
\begin{split}
    \left|\ExnuK \left[ g_{t,V} (U) \right] - \ExHaar \left[ g_{t,V}(U) \right]  \right|
    &= \left|  \sum_{s=0}^t \binom{t}{s} (\varepsilon^\prime )^{t-s} \left( \ExnuK \left[ f_V(U)^{2s} \right] - \ExHaar \left[ f_V(U)^{2s} \right] \right)\right| \\
    &\leq \sum_{s=0}^t \binom{t}{s} (\varepsilon^\prime )^{t-s} \left| \ExnuK \left[ f_V(U)^{2s} \right] - \ExHaar \left[ f_V(U)^{2s} \right] \right| \\
    &\leq \left( \sum_{s=0}^t \binom{t}{s} (\varepsilon^\prime )^{t-s} \right) \cdot \delta_{t,K}\\
    &= (1+\varepsilon^\prime)^t \delta_{t,K} \\
    &< \frac{1}{2} \left(\frac{\sqrt{\varepsilon^\prime}}{C}\right)^{q^2-1},
\end{split}
\end{equation}
where the inequality in the third line is derived from the fact that the following inequality is satisfied for any $1\leq s\leq t$:
\begin{equation}
\begin{split}
    \left|\ExnuK \left[ f_V(U)^{2s} \right] - \ExHaar \left[ f_V(U)^{2s} \right]  \right| 
    &= \frac{1}{q^{2s}} \left| \ExnuK \left[ \tr\left[\left(UV^\dag\right)^{\otimes s, s}\right]\right] - \ExHaar \left[  \tr\left[\left(UV^\dag\right)^{\otimes s, s}\right] \right]  \right| \\
    &= \frac{1}{q^{2s}} \left| \tr\left[\left(M_{\nu}^{(s)} -M_{\rm Haar}^{(s)}\right)  V^{\dag \otimes s,s}\right] \right| \\
    &\leq \frac{1}{q^{2s}} \left\|M_{\nu}^{(s)} -M_{\rm Haar}^{(s)} \right\|_{\infty} \left\| V^{\dag \otimes s,s}\right\|_1 \\
    &\leq \delta_{t,K} .
\end{split}
\end{equation}

From Eqs.~(\ref{seq:mom_eig_g_LB}) and (\ref{seq:mom_eig_K_Haar_gap}), we can derive the inequality
\begin{equation}
\ExnuK \left[ g_{t,V}(U) \right] \geq \ExHaar [g_{t,V}(U)] - \left|\ExnuK \left[ g_{t,V}(U) \right] - \ExHaar \left[ g_{t,V}(U) \right]  \right|  > \frac{1}{2} \left(\frac{\sqrt{\varepsilon^\prime}}{C}\right)^{q^2-1},
\end{equation}
which contradicts the following inequality derived from Eqs.~(\ref{seq:mom_eig_onlyif_contra}) and (\ref{seq:mom_eig_onlyif_t_fix}):
\begin{equation}
    \ExnuK \left[ g_{t,V}(U) \right] <  (1-\varepsilon^\prime)^t < \frac{1}{2} \left(\frac{\sqrt{\varepsilon^\prime}}{C}\right)^{q^2-1}. 
\end{equation}
This completes the proof of the ``only if'' part.
\end{proof}

Finally, by utilizing Lemmas~\ref{lem:slow_decay_prev} and \ref{lem:univ_eig}, we prove Lemma~\ref{lem:spec_gap_univ}.

\begin{proof}[Proof of Lemma~\ref{lem:spec_gap_univ}]
We begin by noting a key relation that holds for any unitary ensemble $\nu$ and moment order $t$:
\begin{equation}
\label{seq:spec_gap_univ_prf1}
    \rho ( M_{\nu^\dag * \nu}^{(t)} - M_{\rm Haar}^{(t)} ) = (1 - \Delta_{\nu}^{(t)})^2.
\end{equation}
This follows from the fact that $M_{\nu^\dag * \nu}^{(t)}$ is Hermitian since $\nu^\dag * \nu$ is an inverse-closed ensemble, and that $M_{\nu^\dag * \nu}^{(t)} - M_{\rm Haar}^{(t)} = (M_{\nu}^{(t)} - M_{\rm Haar}^{(t)})(M_{\nu}^{(t)} - M_{\rm Haar}^{(t)})^\dag.$

We now prove the latter half of Lemma~\ref{lem:spec_gap_univ}. Suppose that $\nu^\dag * \nu$ does not have support on a universal gate set. Then, by Lemma~\ref{lem:univ_eig}, we have $\rho ( M_{\nu^\dag * \nu}^{(t)} - M_{\rm Haar}^{(t)} ) = 1$ for sufficiently large $t$. Substituting into Eq.~\eqref{seq:spec_gap_univ_prf1}, we obtain $\Delta_{\nu}^{(t)} = 0$ for such $t$.

Next, we turn to the first half of the lemma. Suppose that $\nu^\dag * \nu$ has support on a universal gate set. Then Lemma~\ref{lem:univ_eig} implies that $\rho ( M_{\nu^\dag * \nu}^{(t)} - M_{\rm Haar}^{(t)} ) < 1$ for all $t$, and therefore Eq.~\eqref{seq:spec_gap_univ_prf1} yields $\Delta_{\nu}^{(t)} > 0$ for all $t$. Finally, applying Lemma~\ref{lem:slow_decay_prev}, we conclude that $\Delta_{\nu}^{(t)} > B (\log t)^{-2}$ for some constant $B > 0$, since the spectral gap $\Delta_{\nu}^{(t_0)}$ appearing in Eq.~\eqref{seq:slow_decay_Dloc_prev} must be positive.
\end{proof}

\bibliography{bibliography}